\documentclass[%
 reprint,
 superscriptaddress,
nofootinbib,
 amsmath,amssymb,
 aps,
 prx
 unpublished,
 noarxiv
]{quantumarticle}
\usepackage[T1]{fontenc}
\usepackage{graphicx}% Include figure files
\usepackage{dcolumn}% Align table columns on decimal point
\usepackage{bm}% bold math
\usepackage[colorlinks=true, allcolors=blue]{hyperref}
\usepackage{braket}
\usepackage[table]{xcolor}
\usepackage{xfrac}
\usepackage{comment}
\usepackage[normalem]{ulem}
\usepackage{mathtools}
\usepackage{layouts}
\usepackage{quantikz}
\usepackage{tabularx}
\usepackage{multirow}
\usepackage{cellspace} %for better spaced tables
\usepackage{esint}
\usepackage{eczoo}
\usepackage{enumitem}% for custom enumerate labels
\usepackage{makecell}% for multiline cells in tables

\usepackage{tikz}
\usetikzlibrary{shapes.geometric,arrows}
\tikzstyle{startstop} = [rectangle, rounded corners, minimum width=6cm, text width = 6 cm, inner xsep=10 pt,minimum height=0cm, draw=black, fill=red!30]
\tikzstyle{io} = [trapezium, trapezium left angle=70, trapezium right angle=110, minimum width=3.5cm, text width = 3.2 cm, inner xsep=0 pt, minimum height=0cm,text centered, draw=black, fill=blue!30]
\tikzstyle{process} = [rectangle, minimum width=3.5cm, text width = 3.5 cm, inner xsep=-10 pt,minimum height=0cm,text centered, draw=black, fill=orange!30]
\tikzstyle{decision} = [diamond, minimum width=3.5cm, text width = 3.5 cm, inner xsep=-10 pt, minimum height=0cm,text centered, draw=black, fill=green!30]
\tikzstyle{arrow} = [thick,->,>=stealth]

\usepackage{amsthm}% for theorems and stuff
\usepackage[capitalise]{cleveref}% load this as late as possible
\Crefname{criterion}{Criterion}{Criteria}
\crefname{criterion}{Crit.}{Crit.}
\Crefname{definition}{Definition}{Definitions}
\crefname{definition}{Def.}{Defs.}
\Crefname{theorem}{Theorem}{Theorems}
\crefname{theorem}{Thm.}{Thms.}
\Crefname{proposition}{Proposition}{Propositions}
\crefname{proposition}{Prop.}{Props.}
\Crefname{conjecture}{Conjecture}{Conjectures}
\crefname{conjecture}{Conj.}{Conjs.}
\Crefname{heuristic}{Heuristic}{Heuristics}
\crefname{heuristic}{Heur.}{Heurs.}
\newcommand{\vect}[1]{\boldsymbol{\mathbf{#1}}}
\newcommand{\approptoinn}[2]{\mathrel{\vcenter{
  \offinterlineskip\halign{\hfil$##$\cr
    #1\propto\cr\noalign{\kern2pt}#1\sim\cr\noalign{\kern-2pt}}}}}

\begin{document}
\title{Optimising Quantum Error Correction Using Morphing Circuits}
\author{Mackenzie H. Shaw}
\affiliation{QuTech, Delft University of Technology, P.O. Box 5046, 2600 GA Delft, The Netherlands}
\affiliation{Faculty of EEMCS, Delft University of Technology, Mekelweg 4, 2628 CD Delft, The Netherlands}
\author{Barbara M. Terhal}
\affiliation{QuTech, Delft University of Technology, P.O. Box 5046, 2600 GA Delft, The Netherlands}
\affiliation{Faculty of EEMCS, Delft University of Technology, Mekelweg 4, 2628 CD Delft, The Netherlands}
\newtheorem{theorem}{Theorem}
\newtheorem{corollary}[theorem]{Corollary}
\newtheorem{lemma}[theorem]{Lemma}
\newtheorem{proposition}[theorem]{Proposition}
\newtheorem{criterion}{Criterion}
\theoremstyle{definition}
\newtheorem{assumption}{Assumption}
\theoremstyle{remark}
\newtheorem*{remark}{Remark}
\theoremstyle{definition}
\newtheorem{definition}[theorem]{Definition}
\newtheorem{heuristic}[theorem]{Heuristic}
\newtheorem{taggedpropositionx}{Proposition}
\newenvironment{taggedproposition}[1]
 {\renewcommand\thetaggedpropositionx{#1}\taggedpropositionx}
 {\endtaggedpropositionx}
%\numberwithin{equation}{section} %numbers equations by section, i.e. (sec.eqno)
%\renewcommand\qedsymbol{\emoji{smiling-face-with-sunglasses}}
\begin{abstract}
Quantum error correction (QEC) codes are traditionally defined and searched for without specifying the manner in which its syndrome extraction circuits are executed using elementary gates and measurements. We show how morphing circuits introduced in Refs.~\cite{McEwen23,Gidney2023,ST:morphing} provide a way of optimising syndrome extraction circuits and codes directly in terms of connectivity, choice of two-qubit gate (ISWAP versus CNOT) and number of physical qubits. We discuss morphing circuits in code optimisation among Abelian two-block group algebra (2BGA) codes, handling boundaries for 2D codes, codes with single-shot properties, and improving performance in stability experiments against measurement and reset errors. We show that alternating syndrome extraction circuits -- executed with alternating time-reversed rounds -- can be viewed as a two-round morphing circuit whose fault-tolerant properties are computationally much easier to examine than non-alternating syndrome extraction circuits. Our methods find new codes and syndrome extraction circuits of practical interest, including Abelian 2BGA morphing circuits with better code parameters and connectivity than existing circuits.
\end{abstract}
\tableofcontents
\maketitle

\section{Introduction}

One of the building blocks of any fault-tolerant quantum computer is quantum error correction (QEC), which can preserve quantum information and perform logical transformations at low logical error rates using noisy physical operations. QEC codes are usually defined in terms of multi-qubit Pauli operators called stabilisers that are repeatedly measured to detect -- and then correct -- errors that occur on a physical device. To implement a QEC code on a physical device, one needs to design a syndrome extraction circuit that measures each stabiliser using only single-qubit resets, single-qubit measurements, and two-qubit gates. Typically, for quantum LDPC codes, this is done by using one additional qubit per stabiliser in what we call a ``bare-ancilla'' (syndrome extraction) circuit.

However, bare-ancilla circuits are not the only approach to design physically-implementable QEC circuits. One particularly promising alternative is to use a design principle that we call \textit{morphing}, which has already been applied to the surface~\cite{McEwen23}, colour~\cite{Gidney2023}, and \eczoo[Abelian Two-Block Group Algebra (2BGA)]{2bga} codes~\cite{ST:morphing}, -- a highly interesting subclass of lifted-product codes. Morphing circuits -- which have also been referred to as middle-out, dynamic \cite{eickbusch2024:dynamic}, or ancilla-free circuits -- use a QEC code $C$ as input, but output syndrome extraction circuits that move between a different, possibly novel, set of QEC codes $C_{j}^{\mathrm{end}}$ measuring their stabilisers. This approach has a remarkable number of advantages compared to bare-ancilla circuits, including lower connectivity requirements~\cite{McEwen23,Gidney2023,ST:morphing}, simple ways to reduce leakage~\cite{McEwen23,yoshida2025low}, a convenient method for adapting to qubit and coupler dropouts~\cite{debroy2024:luci,higgott2025handling,wolanski2026automated,anker2025:optimized}.

%On the other hand, one ongoing challenge of morphing circuits is that they are sometimes vulnerable to lower weight logical errors than the original QEC code, and there are no existing techniques in the literature that can systematically avoid this issue by designing fault-tolerant morphing circuits.

In this work we perform a deep-dive into morphing circuits.
%to demonstrate that they provide an attractive approach to performing quantum error correction for a wide range of QEC codes and across various physical platforms. We do this in three ways: by generalising known advantages of morphing circuits beyond the specific examples that exist in the literature, by designing techniques to improve the fault-tolerance of morphing circuits, and by uncovering previously-unexplored advantages of morphing circuits when performing lattice surgery experiments.
%To complement our general discussion we present concrete proof-of-principle case studies of morphing circuits applied to specific code families. 
Let us mention a few highlights of this work. 
\begin{enumerate}
\item We numerically search through morphing circuits for Abelian 2BGA codes and find circuits with code parameters that strictly improve upon the bare-ancilla circuits -- which is the first time this has been systematically demonstrated for morphing circuits in the literature as far as we are aware (\cref{subsec:Abelian_2BGA}).
\item We manually optimise the boundary geometry of morphing circuits for the colour code, improving the parameters of the colour code morphing circuits from Ref.~\cite{Gidney2023} and providing multiple circuits that may be of use in QEC experiments (\cref{subsec:colour_code}).
\item Given a repeatedly-applied syndrome extraction circuit, we prove that a circuit that alternates between time-reversed versions of this syndrome extraction circuit: (a) can be viewed as a two-round morphing circuit, (b) can never have a lower circuit-level distance than its non-alternating counterpart, and (c) has a circuit-level distance set by errors within a single syndrome extraction round. This last property drastically reduces the computational cost of calculating the circuit-level distance of a given circuit (\cref{subsec:from_normal_circuits}).
\item We describe and analyse how to set detectors for morphing circuits in a general manner, including for codes with single-shot properties. We use this to prove bounds on the number of rounds required to fault-tolerantly implement lattice surgery operations (\cref{sec:optimising_time_like}) and to improve the logical performance of morphing circuits under a noise model where the measurement and reset error rates are biased (\cref{sec:measurement_reset_bias}).
\end{enumerate}

Here is a more detailed overview of what to expect in this paper. We conclude this section with a literature review in~\cref{subsec:literature}. Then, many of the general insights of this paper are expressed in \cref{sec:definition}, which contains the definition and key properties of morphing circuits, and its relationship to alternating syndrome extraction circuits. 
In \cref{sec:constructing}, we summarise a number of techniques that can be used to construct morphing circuits, including both general and heuristic approaches, while in \cref{sec:manipulating} we generalise the techniques that are available to morphing circuits to reduce leakage. Then, in \cref{sec:optimising_end_cycle_distance} we tackle the important issue of fault-tolerance by proposing heuristic techniques to improve this for morphing circuits. Finally, in \cref{sec:optimising_time_like,sec:measurement_reset_bias} we examine the circuit-level fault-tolerance properties of morphing circuits for lattice-surgery-like operations (including single-shot codes), and explore a previously unexamined advantage that morphing circuits can have over traditional circuits under a noise model in which the reset and measurement noise rates are not equal. The two case studies for Abelian 2BGA codes and colour codes are contained in \cref{subsec:Abelian_2BGA,subsec:colour_code} respectively.

\subsection{Relationship to Existing Literature}\label{subsec:literature}

The first example of a morphing circuit appeared in Ref.~\cite{McEwen23} for the surface code, and this circuit already demonstrated many of the advantages of morphing circuits: the reduced connectivity requirements, its ability to be compiled using CXSWAP gates instead of CNOT gates, and the ability to swap data and ancilla qubits each QEC round. Later, the framework of morphing circuits was used to develop a flexible and efficient scheme to combat qubit and coupler dropouts in the surface code~\cite{debroy2024:luci,higgott2025handling,anker2025:optimized}, which is important in the presence of fabrication errors.

Morphing circuits were first applied beyond the surface code in the colour code in Ref.~\cite{Gidney2023} with only hex-grid (degree-3) connectivity requirements and the ability to compile the circuit in terms of CXSWAP gates~\cite{yoshida2025low}. 

In our previous work~\cite{ST:morphing} we designed morphing circuits for the Bivariate Bicycle (BB) code family~\cite{Bravyi24}: in particular, we found new so-called ``end-cycle codes'' that had a distance and number of qubits that matched the parameters of another code in the BB code family. Because the morphing circuit has a lower connectivity requirement than a corresponding bare ancilla circuit, this represents a total win-win situation for using morphing circuits. More recently, the scheme to combat qubit and coupler dropouts using morphing was generalised from the surface code to arbitrary CSS codes~\cite{wolanski2026automated}.

The majority of the results presented in this work are novel and go beyond the existing results in the literature. However, in the interests of completeness and pedagogy, in \cref{subsec:boundaries,sec:manipulating} we also include some results that are simple generalisations of already existing techniques in the literature; namely, the technique to add boundaries to morphing circuits which was already used in the colour code circuits in Ref.~\cite{Gidney2023}, the technique to swap data and ancilla qubits which was already generalised in Ref.~\cite{ST:morphing}, and the rewriting of morphing circuits in terms of CXSWAP gates which was already done in the colour code in Ref.~\cite{yoshida2025low}.

%In previous work, morphing circuits have been designed for the surface code~\cite{McEwen23}, colour code~\cite{Gidney23}, and Bivariate Bicycle (BB) codes~\cite{ST:morphing} to reduce the connectivity required to execute quantum error correction, without incurring any loss in logical performance. Moreover, morphing circuits have been designed that use CXSWAP gates instead of CNOT gates for the surface code~\cite{McEwen23} and for colour code~\cite{yoshida2025low}. Finally, a strategy for handling qubit and/or coupler defects has been developed by utilising the flexibility of morphing circuits~\cite{debroy2024:luci,higgott2025handling,wolanski2025automated,anker2025:optimized}. These results show the interest in the morphing design strategy. This paper aims to further investigate and demonstrate the use of morphing circuits in optimizing quantum error correction.

%all with only two contraction rounds $F_1$ and $F_2$ and with each contraction circuit compiled to CNOT gates. ; this strategy is also particularly interesting from a theoretical perspective because the mid-cycle code is a subsystem code and the morphing circuit uses more than two contraction rounds. 

\subsection{Some Definitions}

Before moving forward, we briefly define some key terms. We will refer to various types of syndrome extraction circuits that we define here -- we will sometimes omit the words ``syndrome extraction'' for brevity when it is clear from context. We call a syndrome extraction circuit for a stabiliser code {\em normal} when it measures a complete or overcomplete set of stabiliser generators a single time using ancillary qubits.
%, and repeats this identically over time using the same order of single- and two-qubit gates, without change. 
The most common examples of normal circuits are \textit{bare-ancilla} syndrome extraction circuits, which specifically refers to syndrome extraction circuits that use one ancilla qubit for each stabiliser that is measured. Examples of syndrome extraction circuits that are normal but not bare-ancilla are Steane QEC and the superdense colour code circuit~\cite{Gidney2023}. 

We will explicitly refer to resets (denoted by the letter $R$) and measurements (denoted by the letter $M$) in syndrome extraction circuits, which in experiments may have different error rates. The {\em support} of a Pauli string or set of Pauli locations in a quantum circuit is where it acts non-trivially.
In this paper, for simplicity, we restrict ourselves to CSS codes with purely $X$- and purely $Z$ checks.

%  %The CXSWAP gate is ${\rm SWAP}\circ {\rm CNOT}$ and it is single-qubit Clifford equivalent to the ISWAP gate.
% A syndrome extraction ends at the measurement of ancillary qubits followed by a reset which then starts a next syndrome extraction 
% round. 
% define ``bare-ancilla syndrome extraction''
% define complete or normal syndrome extraction.
% define alternating circuit, or locally reflection-symmetric around R-M points.
% define the infinite codes somewhere?

\section{Morphing Circuits: Features and Variations}\label{sec:definition}

%We begin by reviewing the definition, design, and key properties of morphing circuits. 
%In particular, we will make a distinction between the definition of morphing circuits -- which obey a particular structure defined in \cref{subsec:morphing_circuits} -- and the morphing \textit{design principle} (\cref{subsec:morphing_design}) that takes a code $C$ as input and outputs a morphing circuit based on that code. We do this since there are other ways of obtaining morphing circuits that \textit{do not} use the morphing design principle, and we give one example of this in \cref{subsec:from_normal_circuits}. 
Let's provide an overview of this section. In \cref{subsec:morphing_design,subsec:morphing_circuits} we describe morphing in the simple case where $C$ is a CSS stabiliser code, and the morphing circuit has two contraction rounds (as we will define shortly). We also present Prop.~\ref{prop:two_round_circuit_distance} that shows that the circuit-level distance of a two-round morphing circuit can be calculated by considering a memory experiment with only one QEC round -- significantly reducing the computational cost. We explain how to define the detectors of a morphing circuit and how this contrasts to bare-ancilla circuits in \cref{subsec:two_round_morphing_detectors,subsec:two_round_morphing_redundancies}. A convenient subclass of morphing circuits -- purely contracting morphing circuits -- is introduced in \cref{subsec:purely_contracting}, whose defining property is that they can be conveniently represented using ``contraction tree diagrams'' (a generalisation of the LUCI diagrams from Ref.~\cite{debroy2024:luci}). We end in \cref{subsec:general_morphing} with a more general definition of the morphing design principle that can be applied more broadly to subsystem and even Floquet codes.

\subsection{The Morphing Design Principle}\label{subsec:morphing_design}

The key idea behind the morphing design principle is that we wish to measure the stabilisers of a code $C$ without the use of additional ancilla qubits (`ancilla-qubit free'). Let $S$ be a complete (or overcomplete) set of stabiliser generators of the code $C$ that we wish to measure. Then, we split this set $S$ into $J$ \textit{contracting subsets} $S_{j}\subseteq S$, and we will say that $J$ is the number of contraction rounds in the morphing circuit. Within each contracting subset, we require that $S_{j}$ contains an independent set of stabilisers. Moreover, we require that each stabiliser $s\in S$ is contained in \textit{at least} one contracting subset $S_{j}$. Then, we design a set of $J$ Clifford circuits $F_{j}$ (each of which ideally should be constant depth) that we call \textit{contraction circuits}, which have the property that the weight of the ``contracted'' Pauli operator $F_{j}^{\vphantom{\dag}}sF_{j}^{\dag}$ is one for all contracting stabilisers $s\in S_{j}$. The fact that $S_{j}$ is an independent set of stabilisers guarantees that each of these single-qubit Pauli operators is supported on a different qubit. The contracting circuit $F_{j}$ may not be fault-tolerant, as we will discuss later.

These two pieces of information -- the contracting subsets $S_{j}$ and the contraction circuits $F_{j}$ -- are sufficient to specify a morphing circuit as follows, see \cref{fig:overview}. Beginning in the code $C$, we implement the first contracting subset $F_{1}$, followed by a set of measurements $M_{1}$ that measure the single-qubit Pauli operators $F_{1}^{\vphantom{\dag}}sF_{1}^{\dag}$ for all $s\in S_{1}$. Then, we reset these same qubits (and we label this operation $R_{1}$), before applying $F_{1}^{\dag}$ to return to the code $C$. We then repeat this set of operations for the remaining contraction rounds $j>1$. Note that when $C$ is a CSS code, we can take each $F_{j}$ to be circuit made from CNOT gates.
%circuit, we say that the morphing circuit is CSS -- all morphing circuits discussed in this paper are CSS.

Let us briefly consider a few extreme cases for the number of contraction rounds $J$. On one extreme, we could put only one stabiliser in each contracting subset, so that $J=|S|$. However, this would result in stabilisers being measured extremely infrequently, adding to the time overhead of QEC. On the other extreme, we could put all of the stabilisers into one contracting subset, so that $J=1$ and $S_{1}=S$. In this case, $S_{1}$ must be an independent generating set of stabilisers, so the contracting circuit $F_{1}$ must be a ``decoding circuit'' that maps the code $C$ to a set of unencoded qubits, removing any protection from errors. Therefore, the smallest value of $J$ that we could hope for is $J=2$, in this case we say that we have a \textit{two-round} morphing circuit. For most of the rest of this section, we will only consider two-round morphing circuits due to their wide applicability and their simplicity, before explaining and motivating a more general formulation of morphing circuits in \cref{subsec:general_morphing}.

\subsection{Morphing Circuits}\label{subsec:morphing_circuits}

\begin{figure}
\centering
\includegraphics[width=\linewidth]{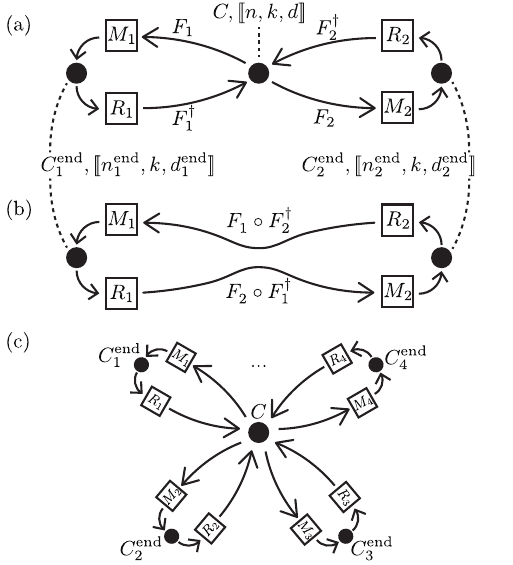}
\caption{Morphing with a stabiliser code. (a) A ``mid-cycle-code'' perspective of a two-round morphing circuit, showing each contracting circuit $F_{j}$ morphing from the mid-cycle code $C$ to each end-cycle code $C^{\rm end}_{j}$. (b) An ``end-cycle-code'' perspective of a two-round morphing circuit, where the contraction circuits have been combined into syndrome extraction circuits that morph from one end-cycle code $C^{\rm end}_{j}$ to the next $C^{\rm end}_{j+1}$. The physical state of the system may or may not explicitly pass through the mid-cycle code $C$, depending on the compilation of the circuit. (c) A morphing circuit with $J> 2$ contraction rounds, e.g. $J=4$.}
\label{fig:overview}
\end{figure}

We illustrate the operation of a two-round morphing circuit in \cref{fig:overview}(a). More generally, we say that any circuit that follows the structure of \cref{fig:overview}(a) is a two-round morphing circuit -- regardless of whether it was designed using the morphing design principle starting from code $C$ or not. The period of time between a set of resets and a set of measurements we call a \textit{measurement round}, and this is the main unit of time we will use throughout this paper. We call $C$ the \textit{mid-cycle} code of the morphing circuit, because it occurs in the middle of any given measurement round.

During each round of measurements $M_{j}$ and resets $R_{j}$, we can also define an \textit{end-cycle} code $C_{j}^{\mathrm{end}}$ with $j=1,\ldots, J$ by tracking the ``expanding'' stabiliser generators $P\in S\setminus S_{j}$ through the contraction and measurement circuit $M_{j}\circ F_{j}$. The code parameters of the end-cycle codes $[\![{n}^{\mathrm{end}}_{j},k^{\mathrm{end}}_{j},d^{\mathrm{end}}_{j}]\!]$ are related to those of the mid-cycle code $[\![n,k,d]\!]$ (for general $J$) as follows:
\begin{itemize}
    \item the number of physical qubits in the end-cycle code satisfies $n^{\mathrm{end}}_{j}=n-|M_{j}|$ where $|M_{j}|$ denotes the number of qubits that are measured in $M_{j}$;
    \item we are guaranteed that $k^{\mathrm{end}}_{j}=k$ because the logical operators of $C$ are not contained in any of the contracting subsets and are therefore not measured by $M_{j}$; and
    \item the distance $d^{\mathrm{end}}_{j}$ is lower-bounded by $d/2^{\mathrm{depth}(F_{j})}$~\cite{ST:morphing}, but is otherwise unrelated to $d$.
\end{itemize}  

For fault-tolerance, we are interested in the \textit{circuit-level} distance of the morphing circuit, which is the smallest number of circuit-level Pauli errors that result in an undetectable logical error. The circuit-level distance is of course upper-bounded by the mid-cycle distance $d$ and each of the end-cycle distances $d^{\mathrm{end}}_{j}$. In our studies, we have found that the end-cycle distances $d^{\mathrm{end}}_{j}$ are often (but not always) the limiting factor in the circuit-level distance of a given morphing circuit, an observation that we will use in \cref{sec:optimising_end_cycle_distance} to construct morphing circuits with larger circuit-level distance.
This is not surprising since the end-cycle codes contain many fewer qubits than the mid-cycle code. In some sense, it is only when the mid-cycle code is {\em non-optimal} in its use of physical qubits versus distance that one can expect that the end-cycle code distance matches that of the mid-cycle code. This is precisely what is happening for the two-round surface code morphing circuits in \cite{McEwen23} where the mid-cycle is the (inefficient) unrotated surface code and the end-cycle codes are rotated surface codes. In contrast, in the reduced-connectivity 2D hexagonal colour code circuits from Ref.~\cite{Gidney2023} the end-cycle codes have a lower distance than a colour code that has the same number of qubits: in Section \ref{subsec:colour_code} we seek to find improved morphing circuits for the colour code.

\subsubsection{The End-Cycle-Code Perspective}\label{subsec:end_cycle_perspective}

So far morphing circuits are described as proceeding from the ancilla-qubit-free mid-cycle code $C$ via an end-cycle code $C_{j}^{\mathrm{end}}$ back to the mid-cycle code $C$. This ``mid-cycle-code'' perspective is useful when designing morphing circuits, but the real use and purpose of morphing circuits is that it provides new end-cycle codes {\em together} with their syndrome extraction circuits. In this perspective, we can consider each circuit $M_{j+1}^{\vphantom{\dag}}\circ F_{j+1}^{\vphantom{\dag}}\circ F_{j}^{\dag}\circ R_{j}^{\vphantom{\dag}}$ as a syndrome extraction circuit for the code $C_{j}^{\mathrm{end}}$ that \textit{also} happens to transform the code from $C_{j}^{\mathrm{end}}$ to $C_{j+1}^{\mathrm{end}}$, as illustrated in \cref{fig:overview}(b).

This end-cycle-code perspective is important as we run memory experiments from end-cycle code to end-cycle code via the mid-cycle code \footnote{To understand why this is important, consider a $Z$-basis memory experiment for a CSS code in which we initialise all qubits in the $\ket{0}$ state, run some number of QEC rounds, then measure all qubits in the $Z$ basis. If we were to run the QEC rounds beginning in the mid-cycle code, then the first contraction circuit that we run wouldn't do anything because each CNOT gate would be acting on a pair of qubits in the $\ket{00}$ state. It therefore only makes sense to begin in an end-cycle code $C^{\mathrm{end}}_{1}$ by applying the resets $R_{1}$ and resetting all data qubits that do not participate in $R_{1}$ in the $\ket{0}$ state. The same arguments apply for the end of the memory experiment.}.
Similarly, when using morphing to perform, say, lattice surgery, one switches from executing memory morphing circuits to lattice surgery morphing circuits when one is in one of the \textit{end-cycle} codes (see for example Appendix F of Ref.~\cite{ST:morphing}).

Another observation related to this end-cycle-code perspective is that it may be possible to compile the syndrome extraction circuit $F_{j+1}^{\vphantom{\dag}}\circ F_{j}^{\dag}$ more efficiently by combining the last gate layer of $F_{j}^{\dag}$ with the first gate layer of $F_{j+1}$, as was recently observed in Ref.~\cite{yoshida2025low}. When this happens, the physical state of the system may \textit{never} be in the mid-cycle code, despite the fact that the morphing circuit was designed using the mid-cycle code as its input.

\subsubsection{Computational Cost of Determining The Circuit-level Distance}

Calculating the circuit-level distance of a normal circuit typically requires simulating a memory experiment with $d$ measurement rounds; however, remarkably, the circuit-level distance of a \textit{two-round} morphing circuit can be determined from a memory experiment containing only \textit{one} measurement round. Since calculating the circuit-level distance is a hard computational problem, these savings drastically reduce the computational cost of the calculation. We consider a $T$-round memory experiment of a two-round morphing circuit defined by ideally preparing an arbitrary logical state $\ket{\overline{\psi}}$ of one of the end-cycle codes $C_{j}^{\mathrm{end}}$, then performing $T$ measurement rounds of the memory circuit following \cref{fig:overview}(a), and then ideally measuring the stabilisers of the last end-cycle code. We define the \textit{circuit-level distance} $d_{\mathrm{circ}}(T)$ as the smallest number of circuit-level errors that cause an undetectable logical error in such $T$-round memory experiment. 
\begin{proposition}\label{prop:two_round_circuit_distance}
    For any two-round morphing circuit for a stabiliser code, we have
    \begin{equation}\label{eq:back_and_forth_fault_tolerance}
        d_{\mathrm{circ}}(T)= d_{\mathrm{circ}}(1).
    \end{equation}
\end{proposition}
\begin{figure}
    \centering
    \includegraphics[width=\linewidth]{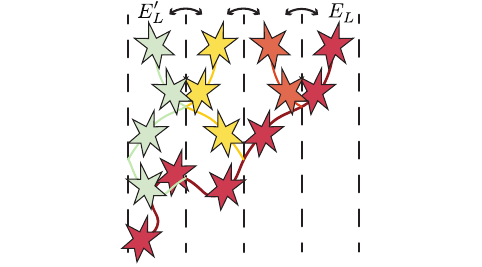}
    \caption{The key step in the proof of Prop.~\ref{prop:two_round_circuit_distance} takes arbitrary undetectable circuit-level logical error $E_{L}$ (dark red) spanning some number of measurement rounds (in this example, four) and ``folds'' it up into an undetectable circuit-level logical error $E_{L}'$ (light green) that only spans one measurement round by repeatedly applying the time-reversal symmetry.}
    \label{fig:time_reversal-main}
\end{figure}
\begin{proof}[Proof Sketch]
    The proof (which can be found in Appendix~\ref{sec:proofs}) relies on the ``time-reversal'' symmetry that exists within two-round morphing circuits. In particular, given any undetectable circuit-level logical error $E_{L}$ in a morphing circuit that contains errors spanning multiple measurement rounds, we can ``fold up'' $E_{L}$ by repeatedly applying the time-reversal symmetry to construct a new undetectable circuit-level logical error $E_{L}'$ that contains errors spanning only a single measurement round, see the sketch in \cref{fig:time_reversal-main}. The folding procedure does not increase the weight of the logical error, which leads to the result.
\end{proof}

Thus, since the circuit-level distance for a memory experiment is set by errors within a single measurement round, we can visualise all the possible circuit-level errors in one place by forward- or backward-propagating each circuit-level error, including reset and measurement errors, to the mid-cycle code, as illustrated in \cref{fig:circuit_level_errors}. Then, any set of circuit-level errors is a logical error if and only if the corresponding set of propagated errors make up a logical error in the mid-cycle code.

\begin{figure}
    \centering
    \includegraphics[width=\linewidth]{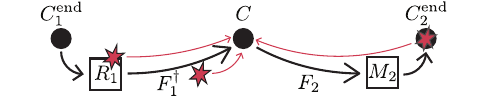}
    \caption{Circuit-level errors occurring anywhere (red stars) in a single-round memory experiment for a two-round morphing circuit can be propagated (red lines) to the mid-cycle code $C$ to completely determine whether the errors constitute a logical error.}
    \label{fig:circuit_level_errors}
\end{figure}

%We point out that result like a Prop.~\ref{prop:two_round_circuit_distance} is not expected to hold for the circuit-level distance of a stability or lattice surgery experiment due to the existence of ``open time-boundaries'': the existence of these boundaries with missing detectors breaks the ``detector-trigger'' preservation property of the folding-up procedure in \cref{fig:time_reversal-main}.

\subsubsection{Connectivity}

As noted in the introduction, one of the advantages of morphing circuits is that they often have lowered connectivity requirements compared to the bare-ancilla circuits for the mid-cycle code, and we will see multiple examples in this paper, in particular in \cref{subsec:Abelian_2BGA,subsec:colour_code}. Indeed, in previous work~\cite{ST:morphing} we designed morphing circuits for 2BGA mid-cycle codes with weight $w$, in which each qubit is only required to be connected to $w-1$ other qubits, hence reducing the degree by one as compared to a bare-ancilla circuit. 
These reduced-connectivity morphing circuits were obtained for most but not all BB codes in \cite{Bravyi24}.
% although note that the morphing circuit construction in Ref.~\cite{ST:morphing} does not work for all 2BGA codes.

However, there is no a priori reason that an arbitrary morphing circuit for an arbitrary (LDPC) code has a lowered connectivity requirement. Indeed, the morphing circuit for the 3D toric code that we present in \cref{subsec:3D_TC} has \textit{increased} connectivity requirement compared to a bare-ancilla circuit measuring the stabilisers of the 3D toric code. 
%While these circuits are therefore less attractive for fixed-connectivity qubit platforms such as superconducting qubits, they may still be of interest in mobile qubit platforms such as neutral atoms or trapped ions.

We also note that there are other ways to reduce connectivity requirements that are not designed using the morphing design principle, such as the Bell-flagging and superdense colour code circuits~\cite{baireuther2019neural,Gidney2023}, and ideas using CXSWAP gates such as directional codes~\cite{geher2025directional} and Louvre~\cite{zhou2025louvre}.

\subsection{Detectors and the Time-Overhead of Morphing Circuits}\label{subsec:two_round_morphing_detectors}

\begin{figure*}
    \includegraphics[width = \linewidth]{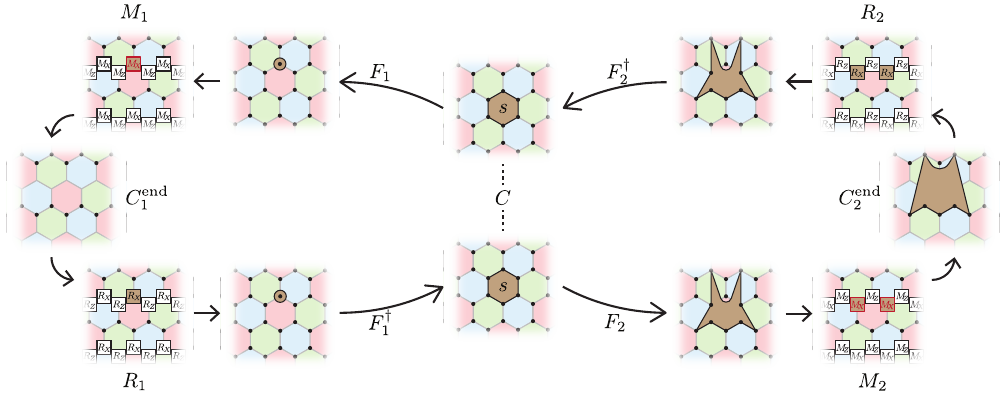}
    \caption{An example of an $X$-detector in a morphing circuit for the colour code which is constructed by considering how an $X$-stabiliser $s\in S_1$ in the mid-cycle code $C$ morphs to the two end-cycle codes and how its support is affected by measurements. Starting from $C_{1}^{\rm end}$, the detector for $s$ is formed by taking the XOR of the $M_X$ measurement outcomes in $M_2$ and the $M_X$ measurements in $M_1$ given by the red boxes. The first error the detector can detect is a reset error in $R_1$, and the last error it can detect is a measurement error in the following $M_{1}$, so we say that the duration of its detection region is two QEC cycles -- the same as in a bare-ancilla syndrome extraction circuit. Note that each $M_X$ in $M_2$ is also the measurement which measures a contracted $s'\in S_2$ and will therefore also have a trivial outcome in the absence of errors.}\label{fig:two_round_detector}
\end{figure*}

Decoding of error correction circuits such as morphing circuits proceeds by utilising the fact that some products of measurement outcomes in these circuits are deterministic in the absence of errors. Following the nomenclature of Stim~\cite{Gidney21} we call such products \textit{detectors}. For example, in a bare-ancilla syndrome extraction circuit, detectors can be formed by taking the product of two consecutive measurement outcomes corresponding to the same stabiliser, which gives rise to a corresponding \textit{detecting region}~\cite{McEwen23} that is local in time. We say that the \textit{duration} (of the detecting region) of the detector is the time-duration of the space-time region in which errors can trigger the detector. Thus, for normal bare-ancilla circuits, the duration of a detector is two syndrome extraction measurement rounds.

In this subsection we will sketch how to construct detectors with time-local detecting regions for two-round morphing circuits, leaving a more general and mathematical description for \cref{subsec:morphing_detectors}.

The structure of detectors in morphing circuits in general can be more complicated than simply the product of two consecutive measurement outcomes, and this was indeed the case for the BB codes in our previous work~\cite{ST:morphing}. In particular, the detector corresponding to a given mid-cycle stabiliser $s\in S_{j}$ will contain the measurement outcome corresponding to the contraction of $s$ in the measurement round $M_{j}$, \textit{as well as} all the measurements the stabiliser $s$ has support on in the other measurement round(s) $M_{j'}$, $j'\neq j$, as shown in the colour code example in~\cref{fig:two_round_detector}. This is despite the fact that, at first sight, it seems that the subset of stabilisers which are measured in $M_1$ have nothing to do with what is measured in $M_2$! To understand this, consider a stabiliser $s\in S_{1}$ of the mid-cycle code and consider how it morphs and is affected by each set of measurements in \cref{fig:two_round_detector}. The detection region starts with one $X$-reset in $R_1$, morphs to $s$ via $F_1^{\dagger}$, gets further expanded via $F_2$ and {\em partially measured} by (in this example) \textit{two} $X$-measurements in $M_2$ that we include in the detector. After this, the stabiliser contracts to finally end back on a single $X$-measurement in $M_1$. Thus, in two-round morphing circuits these detectors involve measurements from \textit{two} measurement rounds and their corresponding detecting regions have a duration of two, the same as for bare-ancilla circuits \footnote{One exception to this is when the stabiliser $s$ is contracting in \textit{both} contraction rounds (i.e.~$s\in S_{1}\cap S_{2}$), in which case the duration of the corresponding detector for $s$ is just one measurement round.}.

Even in two-round morphing circuits, this means that detectors may contain more or less than two measurement outcomes, and measurement errors may trigger more than two detectors. For example, any error on a measurement in $M_{1}$ will trigger one detector corresponding to the contracting stabiliser in the set $S_1$ that is being measured, as well as triggering any number detectors in the following time-step corresponding to expanding stabilisers in the set $S_2$ that {\em happen to also have support on the measured qubit}. If there is more than one such expanding stabiliser in $S_2$, then the measurement error will trigger more than two defects and one will not be able to form `string-like' errors only out of measurement errors. Note that errors on resets in $R_{1}$ will cause a syndrome that is a ``mirrored'' version of a measurement error: the detectors corresponding to expanding stabilisers in $S_{2}$ are triggered in the round \textit{preceding} the detector corresponding to the contracting stabiliser in $S_{1}$. Although this can complicate the decoding problem, it also can have some advantages that we will explore later in \cref{sec:measurement_reset_bias}.

Due to detectors lasting the same duration as in bare-ancilla circuits, the time-overhead of two-round morphing circuits to achieve full fault-tolerance in lattice surgery operations is the same as the overhead of normal, bare-ancilla circuits. However, if there are more than two contraction rounds in the morphing circuit as in \cref{fig:overview}(c), then typically one needs more than one measurement round to measure all of the stabilisers of each end-cycle code. This can have consequences for the amount of time taken to perform lattice surgery which we discuss in much more detail in~\cref{sec:optimising_time_like}.

\subsubsection{Decoding}\label{subsubsec:morphing_decoding}

How do we decode morphing circuits using the detectors we have just defined? Unfortunately, an optimised decoder to decode bare-ancilla syndrome extraction circuits for the mid-cycle code $C$ does not immediately map to an optimised decoder for a morphing circuit derived from this mid-cycle code $C$ via the morphing design principle. As we have argued, in morphing circuits, measurement errors can trigger more than two detectors and decoding methods have to be adapted to optimally handle these. For example, matching decoders~\cite{sahay2022decoder,lee2025color} have been developed for 2D hexagonal lattice colour codes, but similar decoder strategies for colour code morphing circuits would have to be rethought or adjusted.

Despite this, there is a useful relationship between the decoding problem for morphing circuits and the much simpler decoding problem for phenomenological noise which is, by definition, independent of the manner in which stabiliser generators are measured. More precisely, we define the phenomenological noise model as one in which there are only errors on the data qubits of $C$, and errors in the stabiliser generator measurements of $C$. We consider any memory, stability, or even lattice surgery experiment of any depth, so long as the morphing and phenomenological experiments can be mapped to each other (as we explain in more detail in Appendix~\ref{sec:proofs}). Consider the decoding hypergraph $G_{\rm phen}^C=(V_{\mathrm{phen}}^{C},H_{\mathrm{phen}}^{C})$ for this phenomenological noise model, in which detectors are given by vertices and errors are hyperedges with endpoints corresponding to the detectors that are triggered by the error. Then, we can relate the hypergraph $G_{\rm circ}^{C, \rm morph}=(V_{\mathrm{circ}}^{C, \rm morph},H_{\mathrm{circ}}^{C, \rm morph})$ for a morphing circuit based on $C$ subjected to circuit-level noise, to $G_{\rm phen}^C$ using the following proposition.

\begin{proposition}\label{prop:decode}
    Given a quantum LDPC stabiliser code $C$ with a constant-depth $J$-round morphing circuit where every stabiliser generator is measured exactly once in the $J$ rounds of the morphing circuit. Every hyperedge of the hypergraph $G_{\rm circ}^{C, \rm morph}$ can be written as a product of a constant number of hyperedges in $G_{\rm phen}^C$.
\end{proposition}
\begin{proof}[Proof Sketch]
    The proof in Appendix~\ref{sec:proofs} first involves constructing a map between the detectors in a morphing circuit and those in a phenomenological circuit. With this map, given any hyperedge in the morphing hypergraph we show how to explicitly construct a constant-size set of hyperedges in the phenomenological hypergraph that trigger the same detectors.
\end{proof}

Prop.~\ref{prop:decode} is helpful because it implies that any decoder for the phenomenological noise model can be applied to the morphing circuit after decomposing the hyperedges of the morphing circuit into those of the phenomenological noise model. However, this decomposition strategy may in practice lead to worse performance for the morphing circuits due to the loss of information about the correlations between those hyperedges in the decomposition. Note that in our previous work \cite{ST:morphing} we used a general LDPC decoding method, BP-OSD~\cite{Panteleev21}, to decode the morphing circuits for BB codes. However here, for our toric code morphing circuits in \cref{sec:measurement_reset_bias}, we use Belief Matching~\cite{higgott2023improved} for decoding, which uses Belief Propagation to better handle the correlations between hyperedges. Meanwhile, for our colour code morphing circuits in \cref{subsec:colour_code,sec:measurement_reset_bias}, we do not attempt to adapt any bespoke colour code decoders and simply use an (inefficient) minimum-weight decoder implemented with Gurobi~\cite{Gurobi}. 

\subsubsection{Redundancies and Stability Experiments}\label{subsec:two_round_morphing_redundancies}

An important property of QEC codes is that of \textit{redundancies} in the stabiliser generators -- that is, the to-be-measured set of stabiliser generators is not independent. Such redundancies are necessary to define stability experiments~\cite{gidney2022stability}, and can enable single-shot error correction for some codes. More precisely, a redundancy $r\subseteq S$ is any subset of the stabiliser generators of a code whose product is the identity.

Here, we are interested in redundancies of the mid-cycle code $C$, and we wish to see how these redundancies behave in the morphing circuit. One important observation is that the stabilisers in each contracting subset $S_{j}$ must be \textit{independent}; so, any redundancy of the mid-cycle code cannot be contained in just one contracting subset. Writing the stabilisers in the redundancy $r$ that are contained in the contracting set $S_j$ as $r_{j} \equiv S_{j}\cap r$, this means that in a two-round morphing circuit the sets $r_{1}$ and $r_{2}$ must both be non-empty.

Interestingly, in morphing circuits, each redundancy of the mid-cycle code typically gives rise to one redundancy in each end-cycle code $C_{j}^{\mathrm{end}}$, as we now explain. Consider propagating each of the stabilisers in a redundancy $r$ from the mid-cycle code $C$ to the end-cycle code $C_{2}^{\mathrm{end}}$. By definition, the contraction circuit $F_{2}$ contracts each contracting stabiliser in $r_{2}$ onto a single-qubit. Because $F_{2}$ is unitary, we must have that
\begin{equation}\label{eq:two_round_redundancy_product}
    \prod_{s\in r}F_{2}^{\vphantom{\dag}}sF_{2}^{\dag}=I.
\end{equation}
Then, to obtain the stabilisers of $C_{2}^{\mathrm{end}}$, we simply perform a Gottesman-Knill update on each pre-measurement stabiliser by only retaining its support on the qubits that are not measured in $M_{2}$ (as, say, in \cref{fig:two_round_detector}). Doing this removes each of the contracted stabilisers, leaving behind only the stabilisers in the set $r_{1}$. However, the product of these stabilisers must still be the identity on the non-measured qubits from \cref{eq:two_round_redundancy_product}. One trivial way this can be satisfied is when $r_{1}$ contains only one stabiliser $s$: in this case, $s$ must necessarily be the identity during $C_{2}^{\mathrm{end}}$ and therefore is not a stabiliser generator of the end-cycle code $C_{2}^{\mathrm{end}}$. However, more typically, there are multiple stabilisers in $r_{1}$, each of which is a non-trivial stabiliser generator of $C_{2}^{\mathrm{end}}$, in which case $r_{1}$ is itself a redundancy in $C_{2}^{\mathrm{end}}$. In fact, we will see later in \cref{subsec:3D_TC} that in the 3D toric code, \textit{every} redundancy of the mid-cycle code leads to one redundancy in each end-cycle code.

Regardless of whether or not $r$ explicitly induces an end-cycle redundancy, each mid-cycle redundancy $r$ will always give rise to additional detectors for decoding the circuit because they define products of measurements occurring at the same time that are deterministic in the absence of errors. How do we set these detectors given a redundancy $r$ in the mid-cycle code? In two-round morphing circuits, we can define a new detector that contains the measurements corresponding to the contraction of each stabiliser in the set $r_{1}$. Importantly, this detector contains only measurements in a single measurement round, and does not involve the stabilisers in $r_{2}$. This works because the stabilisers in $r_{1}$ necessarily product to the identity in $C_{2}^{\mathrm{end}}$, so there is no need to include any measurements corresponding to stabilisers in $r_{2}$ in the detector. Likewise, a separate detector can be constructed from just the contracting stabilisers in $r_{2}$. The duration of each of these detectors is just one measurement round, unlike the standard detectors discussed in \cref{subsec:two_round_morphing_detectors}, but like detectors defined from redundancies in normal syndrome extraction circuits. Thus one can use these detectors to define stability experiments for the end-cycle code in the standard way. These also enable single-shot error correction in morphing circuits when starting with a single-shot mid-cycle code. We will discuss all of this in more detail in \cref{sec:optimising_time_like}.

\subsection{Purely Contracting Morphing Circuits as Contraction Tree Diagrams}
\label{subsec:purely_contracting}

A particularly convenient subset of morphing circuits are what we call \textit{purely contracting} morphing circuits, which can be specified visually by \textit{a contraction tree diagram}. Consider a morphing circuit with contraction circuits $F_{j}$ and contracting subsets $S_{j}$ and each contraction circuit $F_{j}$ is compiled into a series of {\em layers} of CNOT gates that are implemented sequentially, given by
\begin{equation}
    F_{j}=F_{j,L}\circ\cdots\circ F_{j,1}.
\end{equation}
With this, we say that a compilation of a morphing circuit is \textit{purely contracting} if:
\begin{enumerate}
    \item every CNOT gate in each gate layer $F_{j,\ell}$ causes at least one contracting stabiliser to contract (i.e. reduce its support); and
    \item for every contracting stabiliser $s\in S_{j}$, each gate layer $F_{j,\ell}$ causes $s$ either to contract or to remain unchanged, but {\em not} expand.
\end{enumerate} 

%[To be more precise, let us write $s_{\ell}$ to denote the contracting stabiliser after $\ell$ layers of the contracting circuit $F_{j}$, given inductively by $s_{0}=s$ and $s_{\ell}^{\vphantom{\dag}}=F_{j,\ell}^{\vphantom{\dag}}s_{\ell-1}^{\vphantom{\dag}} F_{j,\ell}^{\dag}$ for $\ell=1,\dots,L$. By the definition of a morphing circuit, $s_{L}$ will be a single-qubit Pauli operator. Then, for the morphing circuit to be purely contracting we require that, for all $s\in S_{j}$ and $\ell=1,\dots,L-1$, there exists a set of qubits $Q$ such that $s_{\ell+1}$ is equal to the projection of $s_{\ell}$ onto the qubits $Q$. Note that this allows $s_{\ell+1}=s_{\ell}$, but does not allow $s_{\ell}$ to ``grow'' or ``expand'' in any gate layer, nor does it allow $s_{\ell}$ to change Pauli type on any of its qubits\footnote{For the definition of purely contracting morphing circuits to make sense, it is therefore advisable \textit{not} to decompose the morphing circuits into the \textit{physically executable} set of gates, which may for example include single-qubit rotations and CZ gates instead of CNOT gates. It may be possible to generalise the definition of purely contracting morphing circuits to allow for such physical decompositions, but for simplicity we do not discuss this here.}.]

\begin{figure*}
    \includegraphics[width = \linewidth]{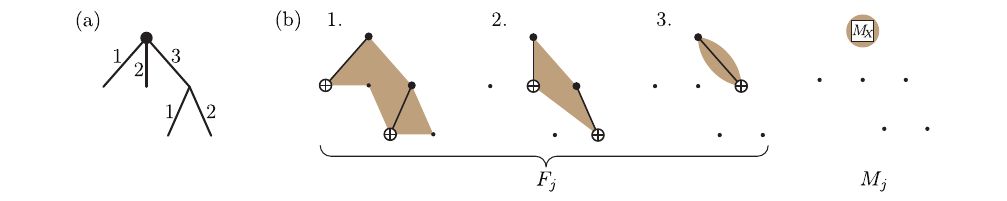}
    \caption{An example contraction tree for a given $X$-stabiliser in a purely contracting morphing circuit; this particular example corresponds to the contraction tree for the circuits in Ref.~\cite{ST:morphing}. The contraction tree is in (a) while the corresponding CNOT and measurement layers are shown in (b), with qubits placed on vertices of the graph. The gates in (b) make up the $s$-contraction circuit $F_{s}$ for the $X$-stabiliser. The support of the contracting $X$-stabiliser {\em before} the gate layer is also shown in brown. Note that only the CNOT gates that correspond to edges in the contraction tree are shown in (b), and other CNOT gates may be performed to contract other stabilisers at the same time so long as they do not inadvertently expand the contracting stabiliser.}\label{fig:contraction_tree}
\end{figure*}

Given a purely contracting morphing circuit, the contraction (by $F_j$) of some stabiliser $s$ can be represented by a contraction tree $T_{s}=(V_{s},E_{s})$: it is simply a graph with vertices corresponding to the qubits in $s$ and with edges corresponding to the CNOT gates that contract the stabiliser $s$, see \cref{fig:contraction_tree}. Each edge $e\in E_{s}$ is additionally labelled by an integer $\ell$ that corresponds to the layer in which the corresponding CNOT gate occurs in during the contraction circuit $F_{j}$ (although this integer can sometimes be inferred from context). One of the vertices of the graph is the ``root'' of the graph, which corresponds to the qubit that the stabiliser $s$ contracts onto to get measured. Note that, in principle, we could also assign each edge a direction corresponding to the direction along which the stabiliser is contracted (so that we have a directed acyclic graph), however this information can be inferred from the root of the graph. Going in the opposite direction, given a stabiliser $s$ and a contraction tree $T_{s}$, we define the $s$-contraction circuit $F_{s}$ as the circuit that only contains CNOT gates involved in the contraction tree $T_{s}$, as shown in \cref{fig:contraction_tree}.

\begin{figure}
    \centering
    \includegraphics[width = \linewidth]{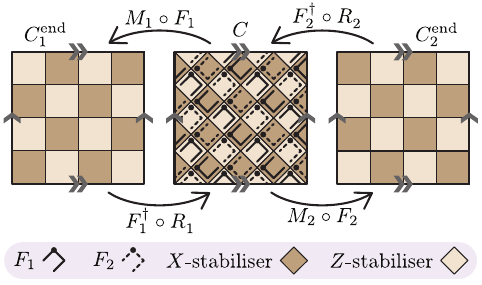}\\
    \caption{The contraction tree diagram and the corresponding end-cycle codes for the distance-4 hex-grid toric code~\cite{McEwen23,debroy2024:luci}, with qubits on vertices, $X$-stabilisers shaded in brown, and $Z$-stabilisers shaded in beige. The grey chevrons indicate the periodic boundary conditions of the toric code; the mid-cycle unrotated toric code $C$ can also be written as an Abelian 2BGA code, see \cref{tab:Abelian_2BGA_code_examples}. The two contraction circuits are represented using solid and dashed contraction trees respectively. Technically, each contraction tree should also be accompanied by labels representing the order in which the CNOT gates are executed, however here the two layers of CNOT gates can be deduced from the contraction trees themselves. Both end-cycle codes are distance-4 \textit{rotated} toric codes, with the two end-cycle codes related by a translation of the qubit lattice.}
    \label{fig:hex_grid_TC}
\end{figure}

We construct a contraction tree diagram by drawing all of the contraction trees corresponding to each stabiliser in the code together on the same diagram, while adding additional labels to indicate which contraction round the stabiliser is contracting in. Moreover, a contraction tree diagram is enough to uniquely specify a purely contracting morphing circuit. As a standard example, we show a contraction tree diagram for the hex-grid toric code in \cref{fig:hex_grid_TC}. Almost all of the morphing circuits we construct in this paper are purely contracting, with only a few exceptions in \cref{sec:measurement_reset_bias}. Such contraction tree diagrams are a generalisation of the LUCI diagrams that were introduced for surface codes in Ref.~\cite{debroy2024:luci}.
%The graph is guaranteed to be acyclic (hence why we call it a \textit{tree});

\subsection{Alternating Normal Syndrome Extraction Circuits are Two-Round Morphing Circuits}\label{subsec:from_normal_circuits}

%We conclude this section with a general construction of morphing circuits \textit{from an existing} normal syndrome extraction circuits. In many ways this construction is ``trivial'' in the sense that we already need to somehow have a normal fault-tolerant circuit for the input code, and moreover does not fit into the ``spirit'' of the morphing circuit design principle because the input code will be an end-cycle code rather than the mid-cycle code of the resulting morphing circuit. 

In this section we show how one can view any \textit{alternating} normal syndrome extraction circuit -- that is, a normal circuit that is run in time-reverse every second measurement round for some code $C$-- as a two-round morphing circuit as in~\cref{fig:overview}(a). This is particularly insightful because it means that the adaptations to reduce leakage in morphing circuits that we discuss in \cref{sec:manipulating} can be applied to alternating normal circuits.

First, let us be slightly more precise with what we mean by a normal (syndrome extraction) circuit. A Clifford syndrome extraction circuit is normal when
\begin{enumerate}
    \item it measures all of the stabiliser generators of the code,
    \item it projects the state into a simultaneous eigenstate of all the stabiliser generators, and
    \item it leaves the logical operators of the code invariant.
\end{enumerate}

Now, consider what happens if we \textit{time-reverse} the operation of the normal circuit. Running the normal circuit in time-reverse order means that the two qubit and single qubit gates $U$ are executed in reverse order and each gate individually is executed as $U^{\dagger}$ instead of $U$. Because we reverse the order, instead of performing a measurement at a given time we apply a reset in the same basis, and vice versa. 

It is not too hard to see that the time-reversed circuit must also be a normal circuit. For example, for normal bare-ancilla circuits in which one measures each Pauli stabiliser $P$ via executing the unitary controlled-$P$, using some two-qubit gates, and reset and measurement of the ancilla qubit in the $X$-basis, it is rather obvious that reversing the order of the two-qubit gates still implements the same circuit. However, the normal circuit can also involve flag qubits or encoded ancilla qubits like in Steane QEC: for completeness we put a short argument in Appendix~\ref{sec:weight_7_details}.
%It turns out that every normal syndrome extraction circuit can be converted by a small 'alternating' modification to a two-round morphing circuit as in Fig.~\ref{fig:overview}(a) without affecting its fault-tolerant properties. This implies that we can apply all techniques in \cref{sec:manipulating}, i.e. reducing leakage and usingh CXSWAP gates, to any normal syndrome extraction circuit with the alternating modification. We state this modification and the upshot concisely as follows:

% We say that the syndrome extraction circuit is ``normal'' if it measures \textit{all} of the stabilisers of $C_{\mathrm{in}}$ in just a single measurement round\footnote{We leave it as an open problem to generalise this to subsystem and/or Floquet codes.}. Because we allow the circuit to use ancilla qubits, it is trivial to rewrite any syndrome extraction circuit that uses more than one measurement round to measure all the stabilisers into a normal circuit simply by using more ancilla qubits and synchronising the resets and measurements to occur at the beginning and end of the circuit. 

Now, we define an alternating normal circuit as one in which we run every second execution of the normal circuit in time-reverse. Here is why such alternating normal circuit is a two-round morphing circuit, i.e.~it has the structure of \cref{fig:overview}(a). We can formally write the normal circuit as a set of resets $R_{\mathrm{norm}}$, followed by a Clifford circuit $F_{\mathrm{norm}}$, followed by a set of measurements $M_{\mathrm{norm}}$. This may require adding extra idling to move any mid-circuit resets and measurements to the start or end of the circuit, but is only necessary to make the connection with morphing circuits clear and does not need to be done physically. With this, we can pick any intermediate layer of the normal circuit and decompose the circuit into $F_{\mathrm{norm}}=F_{\mathrm{norm},2}\circ F_{\mathrm{norm},1}$. Then we can define the mid-cycle code $C^{\mathrm{mid}}$ as the code given by implementing $F_{\mathrm{norm},1}\circ R_{\mathrm{norm}}$ to the stabiliser code $C=C^{\mathrm{end}}$, and the two contracting circuits in~\cref{fig:overview}(a) to be $F_{1}^{\vphantom{\dag}}=F_{\mathrm{norm},1}^{\dag}$ and $F_{2}=F_{\mathrm{norm},2}$. 
%\end{proof}

\begin{figure}
    \centering
    \includegraphics[width=\linewidth]{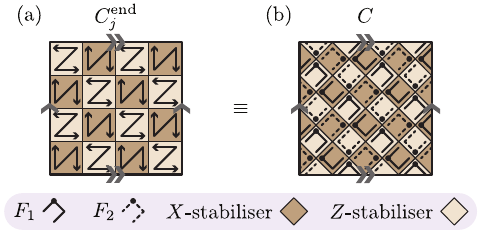}
    \caption{An alternating bare-ancilla circuit for the distance-4 toric code. (a) The alternating (normal) bare-ancilla circuit for the \textit{rotated} toric code, which uses one ancilla qubit in the centre of each stabiliser that is either reset and measured in the $X$-basis for $X$-ancillas or in the $Z$-basis for $Z$-ancillas. Four CNOT gate layers are performed between the ancilla and each of the four data qubits, either from the top-left qubit to the bottom-right qubit or vice-versa, i.e. the arrows are to be followed in one direction in one round, and the other direction in the next round. In either case, the CNOT gates follow the ordering given by the ``N'' or ``Z'' shaped arrows. This circuit is normal because it measures all the stabilisers in a single measurement round. (b) Every alternating normal circuit can be viewed as a two-round morphing circuit with the input code (in this case the \textit{rotated} toric code) as its \textit{end-cycle} code. In the case of the toric code, the morphing circuit is also purely contracting and can be represented with a contraction tree diagram. The mid-cycle code here is the \textit{unrotated} toric code, but requires a square-grid connectivity unlike the hex-grid morphing circuit in~\cref{fig:hex_grid_TC}.}
    \label{fig:toric_back_and_forth}
\end{figure}

%It is worth noting that in principle, given an input qLDPC code $C_{\text{in}}$, the general existence of normal syndrome extraction circuits means that we can always construct a two-round morphing circuit using the alternating modification. 
Let us make a few observations about alternating normal circuits. The input code $C$ to the normal circuit is the \textit{end-cycle} code of the morphing circuit, rather than the \textit{mid-cycle} code, and therefore does not fit within the morphing design principle. The alternating bare-ancilla circuit also, obviously, has the same connectivity requirement as the normal bare-ancilla circuit itself.  
%so may not be beneficial in terms of logical performance. 
Moreover, the construction does not, in general, result in a \textit{purely contracting} morphing circuit and therefore does not always admit a convenient description in terms of contraction trees. Some notable exceptions to this are the toric and surface codes, as such, we have sketched the contraction tree diagram for the toric code alternating bare-ancilla circuits in~\cref{fig:toric_back_and_forth}. Nevertheless, our results in \cref{sec:manipulating} that adapt morphing circuits for leakage reduction will apply also to alternating normal circuits.

Remarkably, any memory experiment for an alternating normal circuit has a circuit-level distance at least as large as its non-alternating counterpart. We consider the same set-up as in Prop.~\ref{prop:two_round_circuit_distance}; that is, we define a $T$-round memory experiment by ideally preparing an arbitrary logical state $\ket{\overline{\psi}}$, performing $T$ rounds of a syndrome extraction circuit, and then ideally measuring the stabilisers of the code. We define the \textit{circuit-level distance} $d_{\mathrm{circ}}(T)$ as the smallest number of circuit-level errors that cause an undetectable logical error in such $T$-round memory experiment. 
\begin{proposition}\label{prop:alternating_d_circ}
    Let $d_{\mathrm{circ}}^{\textnormal{non-alt}}(T)$ be the circuit-level distance of a non-alternating normal circuit, and let $d_{\mathrm{circ}}^{\mathrm{alt}}(T)$ be the circuit-level distance of the corresponding alternating normal circuit. Then, we always have
    \begin{equation}\label{eq:back_and_forth_fault_tolerance}
        d_{\mathrm{circ}}^{\mathrm{alt}}(T)= d_{\mathrm{circ}}^{\mathrm{alt}}(1)=d_{\mathrm{circ}}^{\textnormal{non-alt}}(1)\geq d_{\mathrm{circ}}^{\textnormal{non-alt}}(T).
    \end{equation}
\end{proposition}
\begin{proof}
    This is a straightforward consequence of Prop.~\ref{prop:two_round_circuit_distance}: because every alternating normal circuit is a two-round morphing circuit, we have $d_{\mathrm{circ}}^{\mathrm{alt}}(T)=d_{\mathrm{circ}}^{\mathrm{alt}}(1)$. Moreover, alternating and non-alternating one-round memory experiments are identical because there is only one round, so we have $d_{\mathrm{circ}}^{\mathrm{alt}}(1)=d_{\mathrm{circ}}^{\textnormal{non-alt}}(1)$. And finally, every undetectable circuit-level error that can occur in a one-round memory experiment can also occur in a $T$-round memory experiment, so we therefore have $d_{\mathrm{circ}}^{\textnormal{non-alt}}(1)\geq d_{\mathrm{circ}}^{\textnormal{non-alt}}(T)$, concluding the proof.
\end{proof}

Prop.~\ref{prop:alternating_d_circ} is remarkable because the alternating and non-alternating circuits have identical depth, number of CNOT gates, and connectivity requirements, but the alternating circuit has a circuit-level distance that is both more computationally efficient to calculate \textit{and} can be larger than that of the non-alternating circuit. Indeed, the rotated surface code (with interleaved $X$- and $Z$-stabiliser measurements) can be an example of an alternating circuit having a larger circuit-level distance than a non-alternating circuit, in the situation where the hook errors induced by syndrome extraction circuit aligns with the minimum-weight logical errors, as was previously commented on by Craig Gidney in \href{https://quantumcomputing.stackexchange.com/questions/38963/advantages-and-disadvantages-of-rotated-surface-code?newreg=1ea70e1dc74e46ebb5fef47d45d2039e}{this Stack Exchange comment} and in Refs.~\cite{bluvstein2025architectural,haug2025lattice}. In this case, the hook errors trigger detectors in two different time-steps, so a minimum-weight logical error in the non-alternating circuit spans multiple QEC rounds. Specifically, the circuit-level distance of the non-alternating circuit as a function of the depth of the memory experiment is given by $d_{\mathrm{circ}}^{\textnormal{non-alt}}(T)=\mathrm{max}(d-T,\lceil d/2\rceil)$. Prop.~\ref{prop:alternating_d_circ} implies that the alternating circuit has circuit-level distance $d_{\mathrm{circ}}^{\textnormal{alt}}(T)=d_{\mathrm{circ}}^{\textnormal{non-alt}}(T)=d-1$.

Note, however, that even though the circuit-level distance of an alternating normal circuit is at least as large as that of the non-alternating circuit, this does not immediately guarantee that the logical performance of the alternating circuit is better than that of the non-alternating circuit at all physical error rates, since this performance also depends on the number of minimum-weight and near-minimum-weight circuit-level logical errors that can occur.

Prop.~\ref{prop:two_round_circuit_distance} showed that the circuit-level distance of a two-round morphing circuit is determined by a single syndrome extraction round. Prop.~\ref{prop:alternating_d_circ} similarly shows that the computational problem of determining the circuit-level distance is substantially simpler for alternating normal syndrome extraction circuit as compared to normal non-alternating syndrome extraction circuits (where one may need to consider $\Omega(d)$ rounds of syndrome extraction). Thus these propositions directly bear on the computational task for distance finding algorithms; such algorithms were recently considered in Ref.~\cite{webster2026:distance}. 

\subsection{Beyond Two-Round Morphing for Stabiliser Codes}\label{subsec:general_morphing}

So far we have discussed two-round morphing circuits for stabiliser codes; however, there are many situations when one needs to move beyond this, either by increasing the number of contraction rounds or by moving beyond stabiliser codes. For example, subsystem codes naturally arise when considering qubit dropouts in QEC codes~\cite{auger2017fault,debroy2024:luci}. In this section we describe how morphing can generalise both to subsystem and Floquet codes and to more than two contraction rounds and motivate why these cases can be of interest.

\subsubsection{Morphing for Subsystem and Floquet Codes}

We provide a general formulation of morphing circuits that applies to \eczoo[subsystem codes]{oecc}~\cite{kribs2005unified} and \eczoo[Floquet codes]{da}~\cite{hastings2021dynamically}. To make this formulation clear, it is useful to view morphing circuits merely as an ``ancilla-free'' implementation of a set of multi-qubit Pauli measurements. In particular, consider any set $S_{j}$ of multi-qubit Pauli measurements such that
\begin{itemize}
    \item the elements in the set $S_{j}$ are independent, and
    \item each pair of Pauli operators in the same subset commute, i.e.~$[P_{k},P_{\ell}]=0$ for all $P_{k},P_{\ell}\in S_{j}$.
\end{itemize}
Then, there always exists a contraction circuit $F_{j}$ and a set of measurements $M_{j}$ such that the circuit $F_{j}^{\dag}\circ R_{j}^{\vphantom{\dag}}\circ M_{j}^{\vphantom{\dag}}\circ F_{j}^{\vphantom{\dag}}$ implements all the measurements in $S_{j}$. Conversely, given an arbitrary Clifford circuit $F_{j}$ and a set of single-qubit Pauli measurements $M_{j}$ and corresponding resets $R_{j}$, we are guaranteed that the circuit $F_{j}^{\dag}\circ R_{j}^{\vphantom{\dag}}\circ M_{j}^{\vphantom{\dag}}\circ F_{j}^{\vphantom{\dag}}$ implements a set of commuting and independent multi-qubit Pauli measurements $S_{j}=\{F_{j}^{\dag} m F_{j}^{\vphantom{\dag}}: m\in M_{j}\}$. Therefore, given any sequence $(S_{1},S_{2},\dots)$ of sets $S_j$ of commuting and independent multi-qubit measurements, we can compile this sequence using a morphing circuit for each set $S_{j}$. This works even if there are Pauli operators in two different sets $S_{j}$ and $S_{j'}$ that do not commute.

This formalism can be used to define morphing circuits for subsystem and Floquet codes. Subsystem codes are usually defined in terms of a set of gauge group generators (generating the group $G$) that are to be repeatedly measured. However, these gauge group generators need to be measured in an order such that the stabilisers of the code -- elements in the centre of $G$ -- can be constructed. For our purposes here, let us suppose that we are given a sequence of sets of commuting measurements $(S_{1}^{\mathrm{subsys}},S_{2}^{\mathrm{subsys}},\dots)$ that satisfies this constraint. Then, we can either define each $S_{j}^{\mathrm{subsys}}$ as a contracting subset in the morphing circuit, or we can subdivide each $S_{j}^{\mathrm{subsys}}$ into more than one contracting subset. We can then design the contraction circuits $F_{j}$ in the usual way to construct the full morphing circuit. Designing a morphing circuit for a Floquet code can be done using the same procedure given a specification of the Floquet codes as a (ordered) sequence of sets of commuting multi-qubit Pauli measurements $(S_{1}^{\mathrm{Floq}},S_{2}^{\mathrm{Floq}},\dots)$.

Analysing the mid- and end-cycle codes of morphing circuits for subsystem and Floquet codes requires a generalisation of the concept of mid- and end-cycle codes, as we now explain. In morphing circuits for stabiliser codes, we can define a mid- or end-cycle stabiliser code following Ref.~\cite{McEwen23} by taking each detecting region at that time and defining a stabiliser as the ``time-slice'' of the detecting region. To generalise this to situations where the measurements in a QEC circuit are not necessarily deterministic, we instead define the mid- and end-cycle codes as \textit{subsystem} codes: here, the stabiliser generators are still defined time-slices of detecting regions, but we also define gauge operators as any operator in the mid-/end-cycle code that does not flip any detectors or logical operators in the circuit. This allows us to formally define mid- and end-cycle codes for morphing circuits of subsystem (and Floquet) codes that we will later use in \cref{subsec:adding_rounds}. Interestingly, in some rare cases, a mid-cycle subsystem code can morph into an end-cycle stabiliser code, we will see an example of this later in~\cref{subsec:data_ancilla}.

\subsubsection{The Need for More than Two Contraction Rounds}

Moving beyond two contraction rounds is technically relatively straightforward, and an example of this for a stabiliser code is sketched in \cref{fig:overview}(c). Using more than two contraction rounds, however, increases the number of measurement rounds required for lattice surgery (as we explain in~\cref{sec:optimising_time_like}) \textit{and} means that the computationally-efficient method for calculating the circuit-level distance in Prop.~\ref{prop:two_round_circuit_distance} no longer applies.

However, there are (at least) two reasons we have found that motivate moving beyond two contraction rounds. First, we will see with the example of the Bacon-Shor code in~\cref{subsec:adding_rounds} that subsystem codes typically need more than two contraction rounds because of the additional requirement that the contracting subsets must contain mutually commuting Pauli operators. Moreover, in the Bacon-Shor code, increasing the number contraction rounds also improves the end-cycle code distance, and similar effects may also arise in other codes.

Second, even in stabiliser codes sometimes the most natural approach to designing a morphing circuit uses more than two contraction rounds. A simple approach to designing two-round morphing circuits that applies to, for example, the surface, colour and Abelian 2BGA codes \cite{McEwen23,Gidney2023, ST:morphing}, is to use (approximately) half of the qubits of the mid-cycle code $C$ as ancillas to measure the contracting stabilisers in a given contraction round, leaving the other half behind as data qubits in the end-cycle code. However, some codes -- for example those with favourable single-shot error decoding properties -- use a locally-redundant set of parity checks all of which must be measured to make the syndrome measurement robust.
%For example, for the $Z$-face checks on qubits on the edges of the 3D toric code, there is a meta-check or parity check redundancy for every cube: the face checks around each cube multiply to $I$.
This means that there are significantly more stabilisers than there are qubits in the mid-cycle code. If one uses the usual approach of having half of the qubits as ancillas, it implies that one cannot measure half of the stabilisers in one contraction round, necessitating a $J>2$ morphing circuit as we will explore in~\cref{subsec:3D_TC}.

%In \cref{fig:overview}(a), each of the mid- and end-cycle codes are represented abstractly as vertices in the space of stabiliser codes, connected by the contraction and measurement circuits $M_{j}\circ F_{j}$. While it is possible to draw an analogous diagram to \cref{fig:overview}(a) for arbitrary dynamical protocols, this requires some additional formalism that we will not use elsewhere in the manuscript; \mhs{either:} we therefore leave this as an exercise for the reader \mhs{or:} see Appendix X for more details.

\section{Constructing Morphing Circuits}\label{sec:constructing}

%In~\cref{sec:definition} we articulated the morphing design principle and explored various properties of morphing circuits. 
In this section we present a few methods that can be used in practice to obtain morphing circuits using the morphing design principle that was articulated in \cref{subsec:morphing_design,subsec:morphing_circuits}. We begin in~\cref{subsec:disjoint_circuits} with a simple general result that constructively establishes the existence of ``disjoint'' morphing circuits for arbitrary quantum LDPC code families where the supports of the different Pauli elements in a given set $S_j$ are disjoint. However, disjoint morphing circuits are often impractical because they typically involve many contraction rounds, so in~\cref{subsec:from_contraction_trees} we outline a practical approach to designing morphing circuits by checking whether a given set of contraction trees are ``simultaneously executable''. We apply this practical technique in~\cref{subsec:Abelian_2BGA} to Abelian 2BGA codes and find that morphing circuits can give rise to codes and circuits with improved $[\![n,k,d]\!]$ parameters and lower qubit connectivity requirements than bare-ancilla circuits would allow. 
%This represents the first use of the morphing design principle to our knowledge that doesn't just match but systematically improves the code parameters compared to the corresponding bare-ancilla normal circuits (thus going beyond our previous work \cite{ST:morphing}). 
Finally, in~\cref{subsec:boundaries} we provide a technique how to add boundaries to morphing circuits of topological codes such as surface and colour codes, given a morphing circuit for a boundaryless, infinite code.

\subsection{Existence of Disjoint Morphing Circuits}\label{subsec:disjoint_circuits}
It is not hard to prove the following:
\begin{proposition}\label{prop:existence}
    For any quantum stabiliser code $C$ that is a member of a LDPC code family, there exists a corresponding morphing circuit (called a ``disjoint'' morphing circuit) such that:
    \begin{enumerate}[label=(\alph*)]
        \item $C$ is the mid-cycle code of the morphing circuit,
        \item there are a constant (i.e.~not increasing with the size of the code in the family) number of contraction rounds $J$,
        \item each contracting circuit $F_{j}$, $j=1, \ldots, J$ has constant depth, and
        \item the distance ${d}_{j}^{\mathrm{end}}$ of each end-cycle code ${C}_{j}^{\mathrm{end}}$ scales as $d_{j}^{\mathrm{end}}=\Omega(d)$, where $d$ is the distance of the mid-cycle code.
    \end{enumerate}
\end{proposition}
\begin{proof}[Proof Sketch]
    To construct a disjoint morphing circuit, we first split the stabiliser generators into subsets $S_{j}\subseteq S$ with the property that any two stabilisers $s,s'\in S_{j}$ do not overlap on any qubits. Then, the contraction circuits $F_{j}$ can be trivially constructed using any minimal-depth contraction tree. In Appendix~\ref{sec:proofs} we prove that this construction satisfies the requirements of the proposition.
\end{proof}

\begin{figure}
    \centering
    \includegraphics[width=\linewidth]{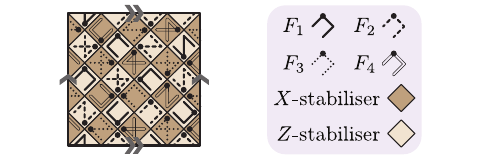}
    \caption{A disjoint morphing circuit for the distance-4 unrotated toric code using $J=4$ contraction rounds. Note that the contraction tree used for each stabiliser can be chosen arbitrarily so long as it represents a CNOT ``fan-out'' circuit that contracts the stabiliser in $\log_{2}(4)=2$ layers of CNOT gates.}
    \label{fig:toric_disjoint_circuit}
\end{figure}

The disjoint construction is typically suboptimal in $J$: for the surface and toric codes it results in $J=4$ as shown in~\cref{fig:toric_disjoint_circuit}, while for the colour code and for the weight-6 bivariate bicycle codes studied in Refs.~\cite{Bravyi24,ST:morphing} it results in $J=6$ due to the 3-colourability of the $X$- and $Z$-stabilisers. It is nevertheless known that constant-depth $J=2$ morphing circuits exist for all of these code families. It is therefore unknown under what conditions a morphing circuit with only two contraction rounds and constant-depth contraction circuits exists for a code family with a growing number of qubits and stabilisers: we leave this as an open problem.

\subsection{Constructing Purely Contracting Morphing Circuits Using Contraction Tree Diagrams}\label{subsec:from_contraction_trees}

In order to get fewer contraction rounds than the disjoint construction in~\cref{subsec:disjoint_circuits}, we can do the following. Construct a contraction tree diagram (as defined in~\cref{subsec:purely_contracting}) by first fixing a choice of contraction tree $T_{s}$ for each stabiliser $s$. Then, assign stabilisers to contraction rounds based on which stabilisers and contraction trees are ``simultaneously contractible''. Using the definition of an $s$-contraction circuit as given in \cref{subsec:purely_contracting} and sketched in \cref{fig:contraction_tree}, we say that two stabilisers $s_{1}$ and $s_{2}$ are simultaneously contractible if their $s_1$- and $s_2$-contraction circuits $F_{s_{1}}$ and $F_{s_{2}}$ are
\begin{enumerate}
    \item Parallelisable: two CNOT gates in $F_{s_{1}}$ and $F_{s_{2}}$ in the same gate layer $\ell$ only overlap if they are identical; and
    \item Proper: the support of the contracting stabiliser $s_{1}$ is not affected in any layer $\ell$ by the contracting circuit $F_{s_{2}}$; i.e.
    \begin{equation}
        \big(F_{s_{1}}^{\rightarrow\ell}\ast F_{s_{2}}^{\rightarrow\ell}\big) s_{1}\big(F_{s_{1}}^{\rightarrow\ell}\ast F_{s_{2}}^{\rightarrow\ell}\big)^{\dag}=F_{s_{1}}^{\rightarrow\ell}s_{1}F_{s_{1}}^{\rightarrow\ell\,\dag}
    \end{equation}
    for all $\ell$; where by $F_{s}^{\rightarrow\ell}$ we mean only the first $\ell$ layers of CNOT gates in the $s$-contraction circuit $F_{s}$, and by $\ast$ we mean to implement the two circuits simultaneously. We also require properness in the other direction, i.e.~that $s_{2}$ is not affected by the contracting circuit $F_{s_{1}}$.
\end{enumerate}
If two stabilisers $s_{1}$ and $s_{2}$ are simultaneously contractible, then we may place them in the same contracting subset $S_{j}$. Once one has found $J$ sets of simultaneously contractible stabilisers $S_{j}$, we have constructed a full contraction tree diagram which fully specifies the morphing circuit.

This is the design philosophy that we have used to construct morphing circuits for Abelian 2BGA codes in Ref.~\cite{ST:morphing}. In the following subsection, we use this design technique as the basis for automating the search for contraction trees for Abelian 2BGA codes.

\subsection{Case Study: Numerical Search for Morphing Abelian 2BGA Codes}\label{subsec:Abelian_2BGA}

\begin{figure*}
    \centering
    \begin{tikzpicture}[node distance=2cm,font = ]

    \node (in1) [io] {Input: $n,w$};
    \node (codesearch) [process, right of=in1, xshift = 2.5 cm] {Search Abelian 2BGA codes $C$ with $|G|=n/2$ and $|A|+|B|=w$};
    \node (codeparams1) [process, right of=codesearch, xshift = 2.5 cm] {Calculate $[\![n,k,d]\!]$};
    \node (bareancillalist) [process, right of=codeparams1, xshift = 2.5 cm] {Store $[\![n,k,d]\!]$ in the ``degree-$w$ bare-ancilla'' parameters list};
    \node (minimal1) [process, below of=bareancillalist] {Find and return minimal codes};
    \node (out1) [io, below of=minimal1] {Output: minimal degree-$w$ Abelian 2BGA codes};
    \node (stop1) [startstop, below of=out1, yshift = -0.5 cm, xshift = -1.75cm] {If run as a bare-ancilla circuit with $C$ as the end-cycle code:
    \begin{itemize}[parsep = -4pt]
        \item uses $2n$ physical qubits
        \item connectivity degree $w$
        \item $d_{\mathrm{circ}}\leq d$
    \end{itemize}};
    \node (morphingsearch) [process, below of=codeparams1] {Search two-round purely contracting homomorphism-based morphing circuits};
    \node (codeparams2) [process, left of=morphingsearch, xshift = -2.5 cm] {Calculate end-cycle code parameters $[\![n^{\mathrm{end}},k,d^{\mathrm{end}}]\!]$ and connectivity};
    \node (morphinglist) [process, left of=codeparams2, xshift = -2.5 cm] {If connectivity degree is $w-1$, store $[\![n^{\mathrm{end}},k,\min(d,d^{\mathrm{end}})]\!]$ in the ``degree-$(w-1)$ morphing'' parameters list};
    \node (minimal2) [process, below of=morphinglist] {Find and return minimal codes};
    \node (out2) [io, right of=minimal2, xshift = 2.5 cm] {Output: minimal degree-$(w-1)$ morphing end-cycle codes};
    \node (stop2) [startstop, below of=out2, yshift = -0.5 cm] {If run as a morphing circuit with $C$ as the mid-cycle code:
    \begin{itemize}[parsep = -4pt]
        \item uses $2n^{\mathrm{end}}$ physical qubits
        \item connectivity degree $w-1$
        \item $d_{\mathrm{circ}}\leq \min(d,d^{\mathrm{end}})$
    \end{itemize}};

    \draw [arrow] (in1) -- (codesearch);
    \draw [arrow] (codesearch) -- (codeparams1);
    \draw [arrow] (codeparams1) -- (bareancillalist);
    \draw [arrow] (bareancillalist) -- (minimal1);
    \draw [arrow] (minimal1) -- (out1);
    \draw [arrow] (out1) -- (stop1);
    \draw [arrow] (codeparams1) -- (morphingsearch);
    \draw [arrow] (morphingsearch) -- (codeparams2);
    \draw [arrow] (codeparams2) -- (morphinglist);
    \draw [arrow] (morphinglist) -- (minimal2);
    \draw [arrow] (minimal2) -- (out2);
    \draw [arrow] (out2) -- (stop2);
    \end{tikzpicture}
    \caption{Flowchart summarising our search over Abelian 2BGA codes and their two-round homomorphism-based morphing circuits, see definitions in~Appendix~\ref{app:Abelian_2BGA_details}.}
    \label{fig:Abelian_2BGA_search_flowchart}
\end{figure*}
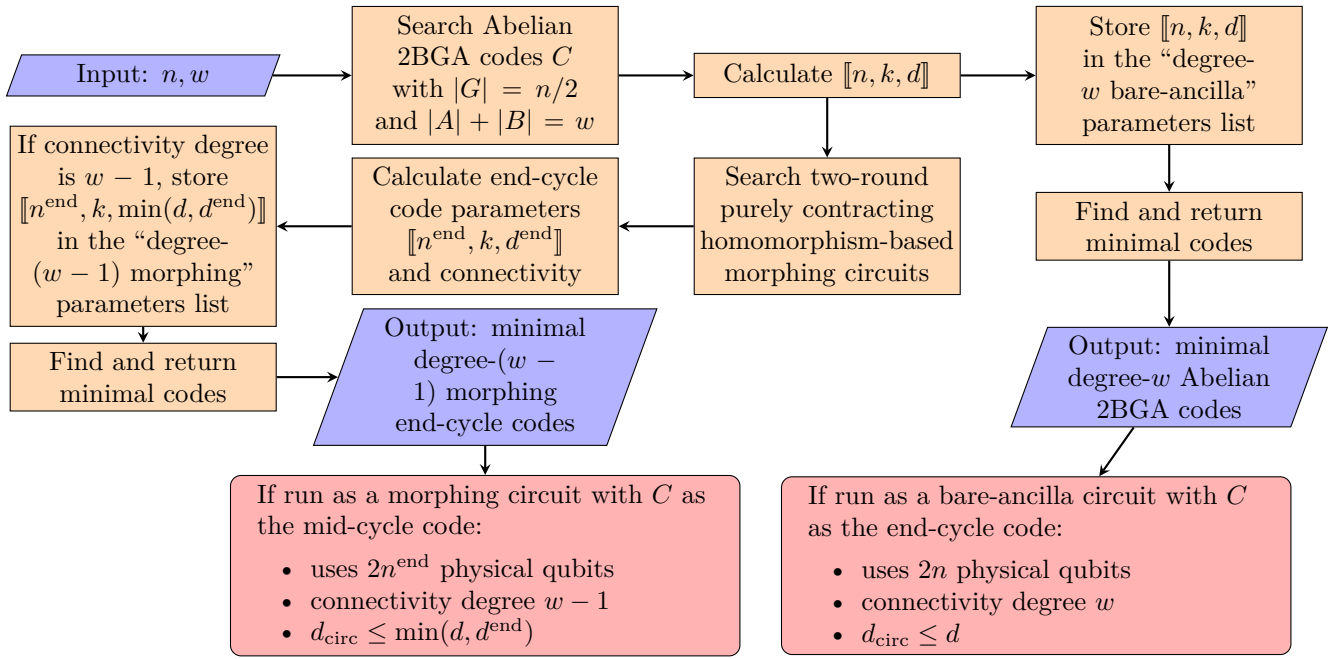

\begin{figure}
    \centering
    \includegraphics[width=\linewidth]{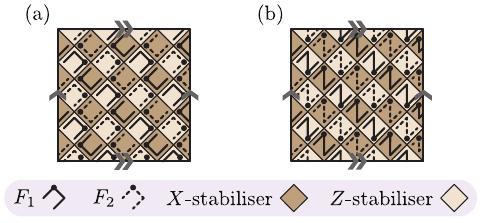}
    \caption{Contraction tree diagrams for two-round purely contracting homomorphism-based morphing circuits for the unrotated toric code. All other two-round purely contracting homomorphism-based morphing circuits for the distance-4 unrotated toric code have the same end-cycle code as either the circuit in (a) or the circuit in (b). (a) One example of a morphing circuit whose end-cycle code is the distance-4 \textit{rotated} toric code with code parameters $[\![16,2,4]\!]$; this corresponds to a toric version of the hex-grid surface code~\cite{McEwen23}. (b) One example of a morphing circuit whose end-cycle code is the \textit{unrotated} toric code on a $4\times 2$ lattice, which therefore has suboptimal code parameters $[\![16,2,2]\!]$.}
    \label{fig:toric_search}
\end{figure}

We perform a study of morphing circuits for \eczoo[Abelian Two-Block Group Algebra (2BGA) codes]{2bga}~\cite{Wang2023Abelian}, which are a broad class of CSS codes that includes toric codes, colour codes on a torus, the bivariate bicycle (BB) codes from Ref.~\cite{Bravyi24}. One feature of Abelian 2BGA codes that is important for us is that qubits and stabilisers are labelled by elements in an Abelian group, see mathematical details in~Appendix~\ref{app:Abelian_2BGA_details}. Our overall method is summarised in~\cref{fig:Abelian_2BGA_search_flowchart} and we explain the key steps here.

First, we set up a numerical search over mid-cycle codes $C$ that are Abelian 2BGA codes with parameters $[\![n,k,d]\!]$ that have stabiliser generators of weight $w=5$ or 6. Because of the large number of equivalences between Abelian 2BGA codes when presented in their usual way in terms of an Abelian group~\cite{wang2022distance,Lin23}, we instead use a presentation of Abelian 2BGA codes in terms of lattices $\Lambda\subseteq \mathbb{Z}^{w-2}$, see~Appendix~\ref{app:lattice}, that more naturally takes into account the equivalences between Abelian 2BGA codes. We use this lattice presentation to perform our numerical search more efficiently, although we note that this perspective may be of independent interest in, for example, defining new families of Abelian 2BGA codes. With this simplification, we can search over the set of Abelian 2BGA codes of a given stabiliser weight and number of qubits. Note that Ref.~\cite{Lin23} performed similar searches for Abelian 2BGA codes with $w\leq 8$, although to the best of our knowledge this search did not utilise the lattice presentation.

Because the space of morphing circuits for a given mid-cycle code is extremely large, we then only search over a subset of morphing circuits that we call \textit{homomorphism-based} circuits, which are defined using the group structure of the code. In particular, we conduct a brute-force search over two-round purely contracting homomorphism-based morphing circuits for a given Abelian 2BGA mid-cycle code $C$, see details in~Appendix~\ref{app:Abelian_2BGA_details}\footnote{Note that many Abelian 2BGA codes do not have \textit{any} two-round homomorphism-based purely contracting morphing circuits.}. The resulting morphing circuit always gives rise to two end-cycle codes $C_{1}^{\mathrm{end}}$ and $C_{2}^{\mathrm{end}}$ that are identical to each other up to a permutation of qubits, so we simply write $C^{\mathrm{end}}$. Moreover, $C^{\mathrm{end}}$ is itself an Abelian 2BGA code: intuitively, this is because two-round homomorphism-based morphing circuits preserve the underlying group-theoretic structure of the Abelian 2BGA code. Note that one can often arrive at the same end-cycle code using different morphing circuits and mid-cycle codes, reducing the number of distinct end-cycle codes that we find. As an illustrative example, our search with the distance-4 unrotated toric code as a mid-cycle code (which is a weight-4 Abelian 2BGA code) yields many morphing circuits but only two distinct end-cycle codes, as shown in~\cref{fig:toric_search}.
We then calculate the end-cycle code parameters $[\![n^{\mathrm{end}},k,{d}^{\mathrm{end}}]\!]$ obtained for each morphing circuit in our search to create a list of end-cycle codes of morphing circuits, see~\cref{fig:d_search_results,fig:k_search_results} in~Appendix~\ref{subsec:numerical_search_appendix} for the full set of results.

%that they have a large number of code symmetries: they are invariant under a large set of translations of qubits, as well as under a $ZX$-duality that swaps $X$- and $Z$-stabilisers.

%Since the set of two-round purely contracting morphing circuits of a given Abelian 2BGA code (or, equivalently, two-round contraction tree diagrams) is intractably large, we therefore limit ourselves to a particular subset of morphing circuits that we call \textit{homomorphism-based} morphing circuits. Not only does this reduce the search space of morphing circuits, but it is also much simpler to find simultaneously contractible subsets of stabilisers, see \cref{sec:Abelian_2BGA_details} for more details. 

The aim of our numerical search is to compare the code parameters and connectivity requirements of morphing circuits and bare-ancilla circuits for Abelian 2BGA codes with the same stabiliser weight $w$. One has to be careful, however, to set up a fair comparison here, because given an Abelian 2BGA code $C$, a morphing circuit will use half the number of physical qubits and will typically have a lower distance than a bare-ancilla circuit for the same code. We therefore compare the \textit{end-cycle} code parameters corresponding to each of the circuits: for a bare-ancilla circuit this is simply the original Abelian 2BGA code $C$, while for a morphing circuit the end-cycle code $C^{\mathrm{end}}$ is determined as described above.

In this work we searched over mid-cycle Abelian 2BGA codes with stabiliser weight 5 and $n\leq 128$, and with stabiliser weight 6 and $n\leq 64$; there is no inherent reason why we stopped at these values of $n$ and in principle one could search over larger values of $n$ as well. This gives rise to four categories of circuits depending on the stabiliser weight $w$ and on the type of circuit (morphing or bare-ancilla). Within each category, we focused only on the ``minimal'' code parameters $[\![n,k,d]\!]$ that we found, i.e.~the parameters such that no other code has parameters $[\![n',k',d']\!]$ with $n'<n$, $k'>k$, and $d'>d$.

\begin{figure}
    \includegraphics[]{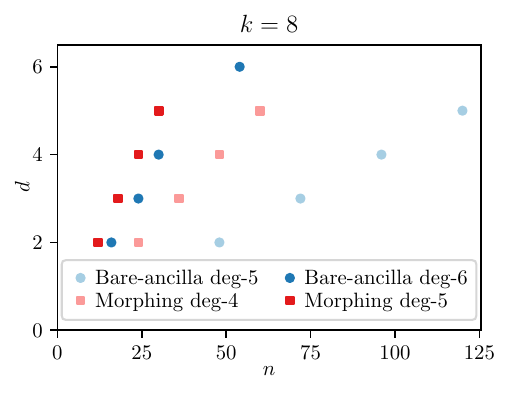}
    \caption{Minimal $k=8$ morphing and bare-ancilla code parameters for weight-5 and 6 Abelian 2BGA codes. Each morphing circuit here has a connectivity degree one lower than the weight of the corresponding Abelian 2BGA code. $n$ here represents the number of qubits in the \textit{end-cycle} code, so for both the morphing and the bare-ancilla circuits the total number of physical qubits required is $2n$. Note also that $d$ here is simply the end-cycle code distance, and therefore the circuit-level distance may be lower than $d$ (for both the morphing and bare-ancilla circuits). The rest of the codes that we found as shown in Appendix~\ref{subsec:numerical_search_appendix}.}\label{fig:k_8_search_results}
\end{figure}

Our search resulted in a large number of novel end-cycle codes and morphing circuits, so here we focus only on the codes we found with $k=8$ as an example in~\cref{fig:k_8_search_results}, see~Appendix~\ref{subsec:numerical_search_appendix} for a full set of results. Excitingly, the minimal morphing code parameters all have strictly better parameters than the corresponding bare-ancilla code parameters of the same weight, demonstrating that searching over morphing circuits can yield improved rate and distance performance compared to codes executed with bare-ancilla circuits. Moreover, the minimal end-cycle codes do not necessarily correspond to minimal mid-cycle codes, suggesting that optimising over end-cycle code parameters may yield different results than optimising over mid-cycle code parameters -- an idea we will return to in~\cref{sec:optimising_end_cycle_distance}. All of the minimal morphing end-cycle code parameters that we found could be implemented with a morphing circuit that has connectivity $w-1$, while all the bare-ancilla circuits require a connectivity of $w$.

An important caveat is that the circuit-level distance for both the bare-ancilla and morphing circuits \textit{may} be lower than the distances shown in~\cref{fig:k_8_search_results}: we leave it to future work to carry out more detailed numerical studies to identify specific codes of interest. To get a tighter bound on the circuit-level distance of the morphing circuits, we additionally check that the mid-cycle code distance is greater than or equal to the end-cycle code distance (as we expect).

Finally, let us consider the depth and CNOT count of the morphing and bare-ancilla circuits. The depth of the end-cycle-to-end-cycle syndrome extraction circuits obtained from the degree-$(w-1)$ morphing circuits may be larger than the corresponding degree-$w$ bare-ancilla circuits. However, this typically does not result in an increased total \textit{number} of CNOT gates required in each measurement round, which has a larger effect on the performance of the circuit under the experimentally-inspired SI1000 noise model. Note however that it is \textit{not} necessarily fair to compare the degree-5 morphing circuits to the degree-5 bare-ancilla circuits, because the morphing circuits are based on weight-6 mid-cycle codes and therefore use more CNOT gates per syndrome extraction cycle than the degree-5 bare-ancilla circuits.

Some specific examples of small codes that may be of interest include a $[\![42,6,6]\!]$ end-cycle code that has a morphing circuit which only requires degree-4 connectivity, and a $[\![30,8,5]\!]$ end-cycle code that only requires degree-5 connectivity; for a full set of minimal codes found, see~Appendix~\ref{subsec:numerical_search_appendix}.

\subsection{Adding Boundaries to Topological Codes}\label{subsec:boundaries}

In this section we describe a method for constructing morphing circuits that is useful for topological codes. In particular, we assume that we already have a CNOT purely contracting morphing circuit for a mid-cycle code $C_{\mathrm{inf}}$ that does \textit{not} have boundaries; we do this by formally considering $C_{\mathrm{inf}}$ to be an infinitely large code, see~Appendix~\ref{sec:inf2BGA} for a description of infinite Abelian 2BGA codes. Then, we wish to construct a morphing circuit for a finite code $C_{\mathrm{fin}}$ that is obtained by adding boundaries to the infinite code $C_{\mathrm{inf}}$. The technique we discuss here has already been introduced in Ref.~\cite{Gidney2023}, we include it here both for completeness and because this can be used to construct morphing circuits for both the surface and colour codes. Although we only use this here for 2D topological codes, the principles could be used to put boundaries on morphing circuits for topological codes in higher dimensions.

\begin{figure*}
    \centering
    \includegraphics[width=\linewidth]{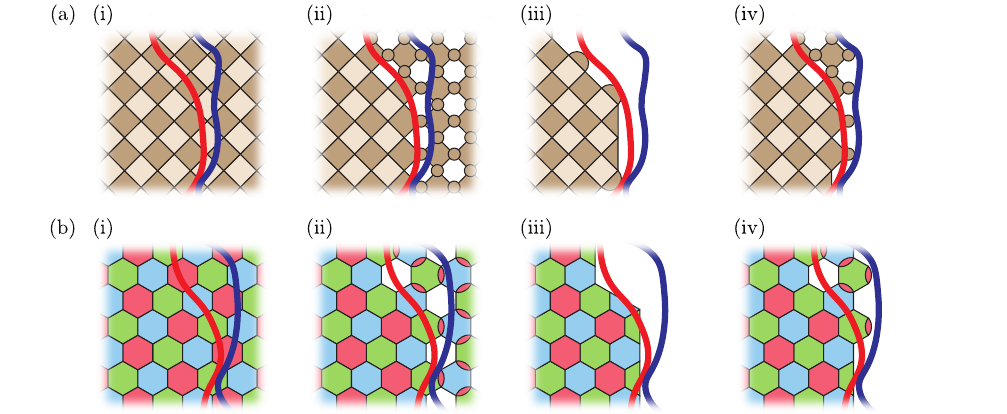}
    \caption{Creating and padding boundaries for (a) an $X$-type boundary in a surface code, and (b) a red boundary in a colour code. The location of the finite code boundary is shown in red, while the location of the padded boundary is shown in navy blue. In each subfigure, we begin in (i) with the code $C_{\mathrm{inf}}$ on an infinite plane. To obtain (ii) we measure all qubits beyond the boundary either (a) in the $X$-basis for the surface code $X$-type boundary, or (b) in Bell pairs for the colour code red boundary. To obtain (iii) we simplify the stabilisers and remove all qubits that are not entangled with the rest of the code, obtaining the usual boundaries of the surface/colour code. To construct morphing circuits for the finite code, we will sometimes need to pad the boundary from (iii). In (iv) we pad the boundary to the (arbitrarily chosen) navy blue line by reinserting some of the trivial stabilisers that are beyond the finite code boundary.}
    \label{fig:padding_boundaries}
\end{figure*}

Our first task is to explain how to create a boundary in the infinite code $C_{\mathrm{inf}}$ for the surface and colour codes as shown in~\cref{fig:padding_boundaries}(i--iii); this procedure is often called ``anyon condensation'' in the literature~\cite{kesselring2018boundaries,kesselring2024anyon}. We do this by fixing a location of the boundary, and then measuring all the qubits beyond this boundary in a basis that is determined by the ``type'' of topological boundary that we impose. These measurements introduce new stabilisers and remove some existing anticommuting stabilisers in accordance with the Gottesman-Knill theorem~\cite{Aaronson04}. Because the measurements of the qubits beyond the boundary break the entanglement between the code and the qubits, we remove them from the code, leaving behind a finite-sized code $C_{\mathrm{fin}}$.

In the surface code, there are two types of boundary that we can choose: $X$-type, in which qubits beyond the boundary are measured in the single-qubit $X$-basis and $X$-logical strings can terminate as shown in \cref{fig:padding_boundaries}(a); and $Z$-type, in which qubits are measured in the $Z$-basis and $Z$-logical strings can terminate. In the colour code, we have more choice: we can impose $X$-, $Y$- or $Z$-type boundaries similar to the surface code (that we collectively refer to as \textit{Pauli} boundaries); or, we can impose red, green, or blue boundaries (that we call \textit{colour} boundaries). These colour boundaries are obtained first by pairing all the qubits beyond the boundary that are connected by an edge that connects two faces of the same corresponding colour. Then, we perform Bell measurements on all these pairs of qubits as shown in \cref{fig:padding_boundaries}(b).

\begin{figure}
    \centering
    \includegraphics[width=\linewidth]{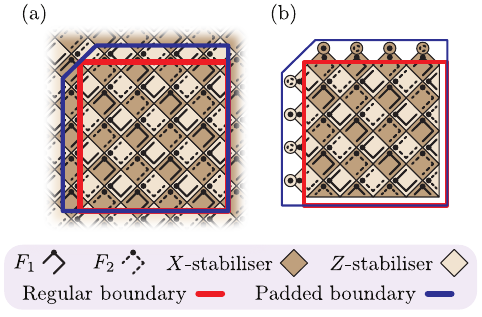}
    \caption{Constructing the hex-grid morphing circuits for the surface code~\cite{McEwen23} from the hex-grid morphing circuit for the surface code on an infinite plane. In (a), the contraction tree diagram for the infinite surface code $C_{\mathrm{inf}}$, along with the desired finite code boundary (red) and the padded boundary (blue) chosen to satisfy~Heur.~\ref{heur:boundaries}, which are used to obtain the contraction tree diagram for the finite surface code $C_{\mathrm{fin}}$ in (b). Note that in (b) some single-qubit stabilisers are measured in the second contraction round, which is indicated with a dashed circle.}
    \label{fig:hex_grid_SC}
\end{figure}

Our next task is to construct a (purely contracting) morphing circuit for the finite code $C_{\mathrm{fin}}$ as its mid-cycle code given a (purely contracting) morphing circuit for the infinite code $C_{\mathrm{inf}}$, as shown in~\cref{fig:hex_grid_SC} for the surface code. For the stabilisers of $C_{\mathrm{fin}}$ that do not overlap on the boundary, this is straight-forward: simply use the same contraction tree that was used in $C_{\mathrm{inf}}$. However, some stabilisers on the boundary of $C_{\mathrm{fin}}$ will have support on fewer qubits than in the bulk and therefore need a different contraction tree. If the qubit(s) that are removed consist of ``leaves'' of the original contraction tree, then we can simply remove these edges from the contraction tree; this is what we do for example in the bottom and right boundaries of the surface code in~\cref{fig:hex_grid_SC}. When this is not the case, then we follow the following heuristic:
\begin{heuristic}\label{heur:boundaries}
    Every contraction tree for a (non-trivial) stabiliser in $C_{\mathrm{fin}}$ should be a \emph{rooted subtree} of the contraction tree for the corresponding stabiliser in $C_{\mathrm{inf}}$.
\end{heuristic}
To satisfy~Heur.~\ref{heur:boundaries}, we can add back some additional qubits to $C_{\mathrm{fin}}$ that were removed when we constructed it from $C_{\mathrm{inf}}$. We call this procedure ``padding'' the boundaries, and is shown in general in~\cref{fig:padding_boundaries}(iv). More specifically, to pad a boundary we simply add back some of the stabilisers from~\cref{fig:padding_boundaries}(ii) until we have enough qubits to satisfy~Heur.~\ref{heur:boundaries}. These stabilisers are typically straightforward to design contraction circuits for since they have weight one or two, as can be seen in the top and left boundaries of the surface code in~\cref{fig:hex_grid_SC}.

\section{Adapting Morphing Circuits For Leakage Reduction}\label{sec:manipulating}

In this section we are interested in adaptations of the morphing circuits that help reduce the effects of leakage in a QEC circuit. In particular, we restrict ourselves here to adaptations that achieve this while preserving the end-cycle code up to qubit permutations -- essentially meaning we can use these adaptations without changing the end-cycle code distance. We consider two such adaptations here: first, the freedom to choose between which qubits are measured in a given contraction round (\cref{subsec:data_ancilla}); and second, rewriting CNOT gate circuits to instead use CXSWAP gates (\cref{subsec:CXSWAP}). Although these adaptations are already known to be generally applicable to morphing circuits~\cite{McEwen23,ST:morphing,yoshida2025low}, we include them here because their careful and general treatment for the surface code in \cref{fig:swap_sc,fig:CXSWAP_sc} (respectively) leads to more qubit-efficient leakage-resistant morphing circuits for the surface code than the circuits that we are aware of in the literature~\cite{McEwen23,eickbusch2024:dynamic}.
Moreover, due to the correspondence between morphing circuits and alternating normal circuits described in~\cref{subsec:from_normal_circuits}, both of these adaptations are also applicable to alternating normal circuits.

\subsection{Choosing between Data and Ancilla Qubits}\label{subsec:data_ancilla}

One of the key features of morphing circuits is their flexibility in being able to assign which qubits are data qubits and which are ancilla qubits in a given contraction round. We make this flexibility rigorous in this section by discussing how to adapt morphing circuits in such a way that every qubit is measured in at least one contraction round, while preserving each of the end-cycle codes $C_{j}^{\mathrm{end}}$. This is particularly desirable for leakage reduction: if every qubit is measured at some point during a morphing circuit, then this measurement may be sufficient to remove ancilla leakage during measurement/reset operations without a dedicated data qubit leakage reduction operation -- an idea which was experimentally demonstrated in the surface code in~\cite{eickbusch2024:dynamic}.

The key adaptation that we are interested in is summarised in the following proposition, which comes from Appendix G1 of Ref.~\cite{ST:morphing}.
\begin{proposition}\label{prop:data_ancilla}
    Consider a contraction circuit $F_{j}$ for a mid-cycle code $C$ compiled into $L$ layers of $\mathrm{CNOT}$ gates $F_{j,\ell}$, $\ell=1,2\ldots, L$ with an end-cycle code $C^{\rm end}_{j}$. Then swapping the direction of any of the $\mathrm{CNOT}$ gates in the final layer of the contraction circuit $F_{j,L}$ gives an end-cycle code ${C}^{\rm end}_{j}$ that is identical to $C^{\rm end}_{j}$ up to a permutation of qubits.
\end{proposition}
\begin{proof}
    Suppose for simplicity that we are contracting an $X$-stabiliser. The statement relies on the following circuit identities (which are taken from our previous published work in Eq.~(G3) of Ref.~\cite{ST:morphing}):
    \begin{equation}\label{eq:final_layer_reverse_circuit_identities}
        \raisebox{-0.69 cm}{\includegraphics{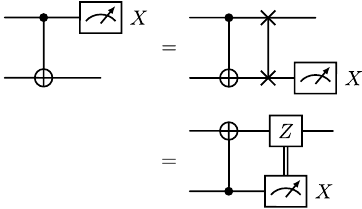}}
    \end{equation}
    Clearly, the $X$-operator is contracted as desired. We therefore only need to argue why the end-cycle code is the same up to qubit permutations. In particular, the conditional $Z$-gate does not change the stabilisers of the end-cycle code, except possibly by applying a $-1$ phase to the stabiliser (which is known from the measurement outcome). This does not change the end-cycle code, but only changes the stabiliser eigenspace (or Pauli frame) that the codespace is in.

    A similar result also trivially holds for $Z$-stabilisers, proving the result.
\end{proof}
It is worth noting that the conditional $Z$ gate in~\cref{eq:final_layer_reverse_circuit_identities} means that the structure of the \textit{detectors} of the morphing circuit is \textit{not} preserved by the transformation in~Prop.~\ref{prop:data_ancilla}.
%, which is something that we will explore further in~\cref{sec:optimising_time_like}.

Prop.~\ref{prop:data_ancilla} provides us with remarkable flexibility to choose which qubits are measured in each contraction round. However, sometimes it is not sufficient to measure \textit{every} qubit in at least one contraction round. For the remaining qubits that are not measured in any of the contraction rounds, we can move it to an extra ancilla qubit in one of the contraction rounds by using one of the two \textit{single-bit teleportation} gadgets~\cite{Zhou_2000}
\begin{subequations}\label{eq:SWAP_gadgets}
\begin{align}
    \begin{quantikz}[align equals at = 1.5]
        &\ctrl{1}&\meter{X}\wire[d][1]{c}&\wireoverride{n}\\
        \lstick{$\ket{0}$}&\targ{}&\gate{Z}&
    \end{quantikz}&,\quad\label{eq:SWAP_gadget_1}\\
    \begin{quantikz}[align equals at = 1.5]
        &\targ{}&\meter{Z}\wire[d][1]{c}&\wireoverride{n}\\
        \lstick{$\ket{+}$}&\ctrl{-1}&\gate{X}&
    \end{quantikz}&.\label{eq:SWAP_gadget_2}
\end{align}
\end{subequations}
More specifically, we add the CNOT gate of such gadget into the final layer of one of the contraction circuits $F_{j}$, and add the reset and measurement to the respective reset and measurement rounds $R_{j-1}$ and $M_{j}$. Again, the conditional Pauli gates can change the detecting regions and the Pauli frame of the state during the computation, but not the end-cycle code itself. The qubit is thus measured out in $M_{j}$, with the end-cycle code remaining the same except that it now lives on the newly introduced ancilla qubit instead of the original one.

\begin{figure}
    \centering
    \includegraphics[width=\linewidth]{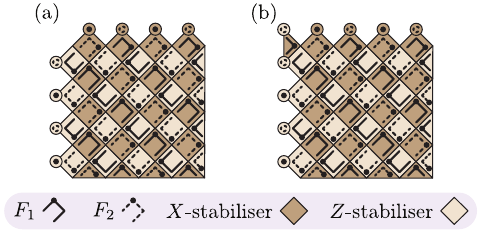}
    \caption{Construction of a hex-grid morphing surface code circuit in which every qubit is measured at least once in the morphing circuit; note that the single-qubit stabilisers that are measured in the second contraction round are indicated with a dashed circle. (a) We first apply~Prop.~\ref{prop:data_ancilla} to the second contraction round of the hex-grid surface code circuit in~\cref{fig:hex_grid_SC}(b). However, because there is one more qubit than there are stabilisers in the surface code, this still leaves one qubit in the top-left corner that is not measured in either measurement round. In (b) we use the one-bit teleportation gadget in~\cref{eq:SWAP_gadget_1} to measure this top-left qubit in the first contraction round. Interestingly, as foreshadowed in~\cref{subsec:general_morphing}, the mid-cycle code is now a subsystem code, however we are still guaranteed that the end-cycle codes are both still rotated surface (so stabiliser) codes up to a permutation of qubits.}
    \label{fig:swap_sc}
\end{figure}

With this, we can always construct a morphing circuit in which each qubit is measured in at least one contraction round: first we attempt to apply~Prop.~\ref{prop:data_ancilla} to measure the qubit, and if this does not work, we include an additional qubit and use a single-bit teleportation gadget from~\cref{eq:SWAP_gadgets}. In many cases, not many additional qubits are needed: for the surface code only one additional qubit is needed as shown in~\cref{fig:swap_sc}. This stands in contrast to the ``walking'' surface code circuit~\cite{McEwen23}, which uses the same modifications described here while using $O(d)$ additional qubits\footnote{Note, however, that the walking circuit~\cite{McEwen23} guarantees that the modified end-cycle code is spatially shifted compared to the standard rotated surface code, while the end-cycle code in our circuit in~\cref{fig:swap_sc} is related to the rotated surface code by a more general qubit permutation. Therefore the walking circuit is necessary if one wishes to string together multiple ``steps'' in the circuit to move the surface code patch a macroscopic distance across a physical device.}.

\subsection{CXSWAP Gates}\label{subsec:CXSWAP}
The second adaptation is to rewrite an arbitrary morphing circuit that uses CNOT gates into a morphing circuit that uses CXSWAP gates, again while preserving the end-cycle code up to qubit permutations. The CXSWAP gate is defined as
\begin{equation}
    \raisebox{-0.5 cm}{\includegraphics{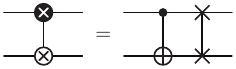}}\;,
\end{equation}
and is an attractive alternative entangling gate as compared to a CZ gate. The CXSWAP gate is equivalent to an ISWAP gate up to single-qubit Clifford gates), which is a native gate in superconducting platforms that may be less prone to leakage than a CZ gate~\cite{McEwen23,eickbusch2024:dynamic}. For this reason, syndrome extraction circuits that exclusively use CXSWAP gates have been designed for morphing surface and colour code circuits~\cite{McEwen23,yoshida2025low} as well as for novel ``directional'' codes~\cite{Geh_r_2024}.

Here we show that any CNOT morphing circuit for any number of rounds $J$ can be written in terms of CXSWAP gates, as shown for example in \cref{fig:CXSWAP_sc} for the surface code. Consider the contraction circuit $F_{j}=F_{j,L}\circ \dots \circ F_{j,1}$ for any $j=1, \ldots, J$, compiled into $L$ layers $F_{j,\ell}$ of simultaneously-executable CNOT gates. Let
\begin{equation}\label{eq:CXSWAP_initial_rewrite}
    F_{j,1}'=V_{1}^{\mathrm{SWAP}} F_{j,1},
\end{equation}
where $V_{1}^{\mathrm{SWAP}}$ is a set of SWAP gates on the pairs of qubits in the CNOTs in $F_{j,1}$, so that $F_{j,1}'$ consists of a single layer of CXSWAP gates. Then, for $\ell\geq2$ we recursively define
\begin{multline}\label{eq:CXSWAP_recursive_rewrite}
    F_{j,\ell}'=V_{\ell}^{\mathrm{SWAP}}(V_{1}^{\mathrm{SWAP}}\dots V_{\ell-1}^{\mathrm{SWAP}})\\
    \times F_{j,\ell}^{\vphantom{S}}(V_{1}^{\mathrm{SWAP}}\dots V_{\ell-1}^{\mathrm{SWAP}})^{\dag},
\end{multline}
where $V_{\ell}^{\mathrm{SWAP}}$ is defined with SWAPs on pairs of qubits such that $F_{j,\ell}'$ consists of a single layer of CXSWAP gates. The definitions in~\cref{eq:CXSWAP_initial_rewrite,eq:CXSWAP_recursive_rewrite} guarantee that $F'$ is equal to $F$ up to a qubit permutation at the end of the circuit, so the end-cycle code is preserved up to qubit permutations (although note that potentially different qubits may be measured).

\begin{figure*}
    \centering
    \includegraphics[width=\linewidth]{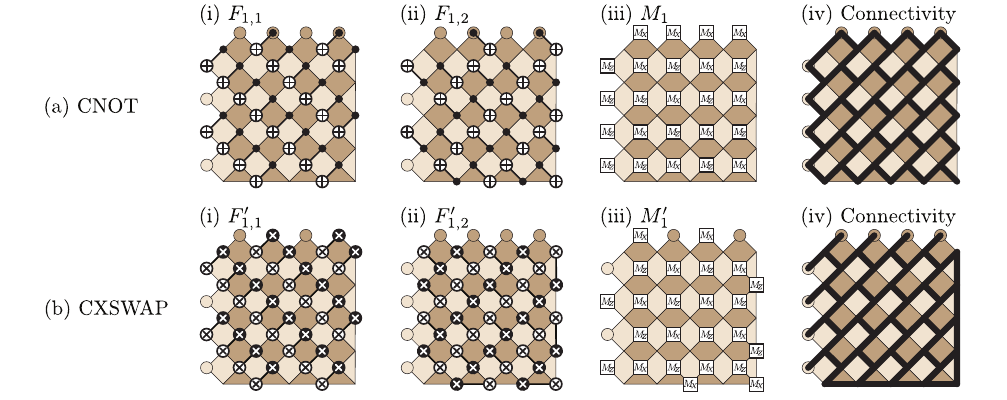}
    \caption{The hex-grid surface code morphing circuits compiled with (a) CNOT gates and (b) CXSWAP gates. The contraction circuit in (a) corresponds exactly to the contraction tree diagram in~\cref{fig:hex_grid_SC}(b), while the circuit in (b) is obtained from (a) by applying the procedure in~\cref{eq:CXSWAP_initial_rewrite,eq:CXSWAP_recursive_rewrite}. We show the mid-cycle stabilisers (instead of the stabilisers at each time in the circuit) throughout for easier comparison between qubit locations. In (i)--(iii) we show the layers of the first contraction circuit $M_{1}\circ F_{1}$ as CNOT and CXSWAP circuits, while in (iv) we show the required connectivity for each compilation taking into account both contraction circuits $F_{1}$ and $F_{2}$.}
    \label{fig:CXSWAP_sc}
\end{figure*}

In~\cref{fig:CXSWAP_sc} we apply this procedure to the hex-grid morphing surface code circuit, giving a CXSWAP circuit that uses the same number of qubits as the CNOT circuit. The CXSWAP circuit that we obtain in~\cref{fig:CXSWAP_sc}(b) uses the same number of qubits as the CNOT circuit in~\cref{fig:hex_grid_SC}, unlike the CXSWAP surface code circuit in Ref.~\cite{McEwen23} which uses $O(d)$ additional qubits. However, this comes at the cost of altered connections required at the boundaries of the surface code, although the overall maximum degree of each qubit is still three. A possible experimental disadvantage of this altered connectivity at the boundary is that it is not consistent with it being a chunk of a larger surface code lattice because it requires different couplers.

This procedure of converting from CNOTs to CXSWAPs was also applied to colour code morphing circuits in Ref.~\cite{yoshida2025low}\footnote{For completeness, we note that the authors of Ref.~\cite{yoshida2025low} also utilised an additional simplification that works for the colour code but not the surface code; namely, before applying the procedure in~\cref{eq:CXSWAP_initial_rewrite,eq:CXSWAP_recursive_rewrite}, each CNOT gate in $F_{1,1}$ can already be combined with a CNOT that occurs immediately before it, in $F_{2,1}^{\dag}$, to give a CXSWAP gate (since a SWAP equals three alternating CNOTs). The procedure in~\cref{eq:CXSWAP_initial_rewrite,eq:CXSWAP_recursive_rewrite} then only needs to be applied to the second layer onwards of each of the contraction circuits.}.

\section{Improving the End-Cycle Code Distance of Morphing Circuits}\label{sec:optimising_end_cycle_distance}

 One of the challenges when designing morphing circuits is to minimise the (typical) reduction in distance when going from mid-cycle code $C$ to end-cycle code $C^{\rm end}_j$, as was seen for the colour and BB codes in Refs.~\cite{Gidney2023,ST:morphing}. In this section we present two concrete strategies that one can use to improve the circuit-level distance for some morphing circuits by improving the distance of the end-cycle code.

In~\cref{subsec:optimising_boundaries} we discuss how to systematically optimise the geometry of the boundaries of topological codes in order to decrease the number of physical qubits required to achieve a given end-cycle distance. We then use this strategy in~\cref{subsec:colour_code} to find improved morphing circuits for the colour code. Our second strategy is to increase the number of contraction rounds $J$ in the morphing circuit: we use the concrete example of the Bacon-Shor code to show how this can also be used to improve the distance of the end-cycle code. This latter strategy comes at the cost of measuring the stabilisers less frequently -- the consequences of which we will investigate in more detail in~\cref{sec:optimising_time_like}.

\subsection{Optimising the Boundary Geometry of Topological Morphing Circuits}\label{subsec:optimising_boundaries}

It is well-known that the boundary conditions of topological codes can have an impact on the code parameters, for example, the rotated surface code uses half the number of physical qubits as the unrotated surface code to achieve the same distance\footnote{Although it is also known that distance is not `everything' for error rates close to threshold, as also {\em the number} of logical operators of low-weight plays a role in the logical performance \cite{Beverland_2019}.}. More recently, it was shown that optimising the periodic boundary conditions of the 4D toric code can lead to greatly improved code parameters~\cite{aasen2025topologically}. The aim of this section is to apply this boundary optimisation to morphing circuits to improve the end-cycle distance of the codes.

The morphing design principle asks us to start with a mid-cycle code $C$, after which we design contracting subsets and contraction circuits that implicitly define the end-cycle codes $C_{j}^{\mathrm{end}}$. It may be tempting to use an optimised topological code as the mid-cycle code $C$, but we would really like the \textit{end-cycle} code to have boundaries that are optimised: even if the mid-cycle code has suboptimal boundaries, it can make for this by having additional qubits. Indeed, we have already seen an example of this in the hex-grid surface code in~\cref{fig:hex_grid_SC}: each end-cycle code is a rotated surface code (which has optimal boundaries), while the mid-cycle code is an unrotated surface code (which has suboptimal boundaries). Nevertheless, the mid- and end-cycle codes have the same distance.

\begin{figure}
    \centering
    \begin{tikzpicture}[node distance=2cm]

    \node (in1) [io] {Infinite mid-cycle code $C_{\mathrm{inf}}$ and contraction circuits $F_{\mathrm{inf},j}$};
    \node (proc1) [process, right of=in1, xshift = 2 cm] {Calculate infinite end-cycle codes $C_{\mathrm{inf},j}^{\mathrm{end}}$};
    \node (proc2) [process, below of=proc1, yshift = -0.5 cm] {Optimise end-cycle boundary conditions};
    \node (proc3) [process, left of=proc2, xshift = -2 cm] {Transfer end-cycle code boundaries to mid-cycle code, padding where necessary (\cref{subsec:boundaries})};
    \node (out1) [io, below of=proc3, yshift = -0.5 cm] {Output: Optimised finite mid-cycle code $C$ and contraction circuits $F_{j}$};

    \draw [arrow] (in1) -- (proc1);
    \draw [arrow] (proc1) -- (proc2);
    \draw [arrow] (proc2) -- (proc3);
    \draw [arrow] (proc3) -- (out1);
    \end{tikzpicture}
    \caption{Summary flowchart of the optimisation procedure for (2D) topological morphing circuits; for more details see Appendix~\ref{sec:code_boundary_optimisation}.}
    \label{fig:optimisation_flowchart}
\end{figure}
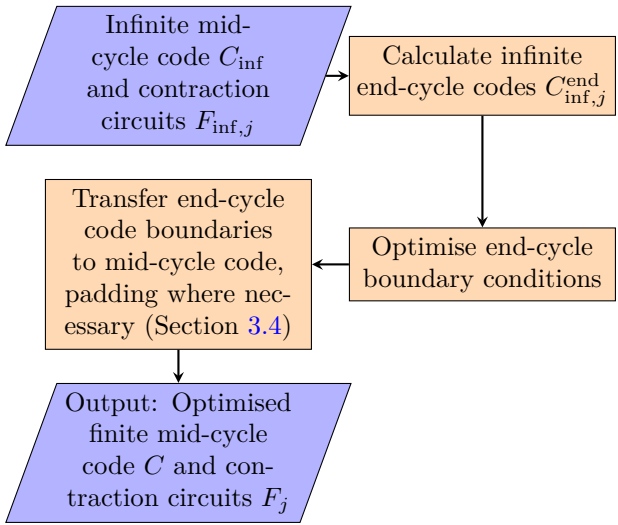

The strategy we follow is summarised in~\cref{fig:optimisation_flowchart}. Similar as in~\cref{subsec:boundaries}, we start with an \textit{infinite} topological code $C_{\mathrm{inf}}$, for example a surface code or a colour code on an infinite 2D plane. Along with this, we assume we have a set of contraction circuits $F_{\mathrm{inf},j}$ and contracting subsets $S_{\mathrm{inf},j}$ for this infinite code. With this, we can calculate the infinite end-cycle codes $C_{\mathrm{inf},j}^{\mathrm{end}}$ whose boundaries we wish to optimise.

%Given that each of the contraction circuit is constant depth, we are guaranteed that the end-cycle codes $C_{\mathrm{inf},j}^{\mathrm{end}}$ are in the same ``topological phase'' as the mid-cycle code, for example, the end-cycle codes of 2D topological mid-cycle codes will have ``string-like'' logical operators. 
By analysing the structure of the logical operators in the end-cycle code, we can try to determine the optimal boundaries to apply to the end-cycle code. We describe this in more mathematical detail in~Appendix~\ref{sec:code_boundary_optimisation}, but intuitively we can do this by judging what geometry optimises the minimum length of a logical string given a fixed number of qubits.
%At this point, however, it is not straightforward to determine exactly what the contracting circuits and the stabilisers of the end-cycle code look like on the boundaries, even if we know what the geometry of the boundaries should be. 
We then ``transfer'' these optimal boundary conditions to the mid-cycle code by keeping track of which expanding stabilisers in the end-cycle code lie on the boundary, and placing the boundaries of the mid-cycle code to match the location of these stabilisers. We can then adapt the contraction circuits to respect these boundary conditions using the methods introduced in~\cref{subsec:boundaries}, which may involve ``padding'' the boundaries.

This strategy can be used for both planar and periodic boundary conditions, and is most straightforward when the logical operators are string-like as is the case for the 2D surface and colour codes and the $Z$-logical operators in the 3D toric code\footnote{For topological codes with non-string-like operators, such as the 4D toric code, one may need to do a numerical search of end-cycle boundary conditions analogous to Ref.~\cite{aasen2025topologically}.}. Because we have optimised the boundaries for the end-cycle code instead of the mid-cycle code, it is also prudent to double check that the distance of the mid-cycle code $C$ is not lower than the desired end-cycle code distance. In the following subsection, we will apply this strategy to a few different morphing circuits for the 2D hexagonal colour code. We also use this strategy implicitly in~\cref{subsec:3D_TC} to design morphing circuits for the 3D toric code. And finally, in~Appendix~\ref{subsec:w4_optimisation_appendix} we derive optimal morphing circuits for weight-4 Abelian 2BGA codes to complement our numerical study of weight-5 and 6 Abelian 2BGA codes in \cref{subsec:Abelian_2BGA}.

\begin{figure}
    \centering
    \includegraphics[width=\linewidth]{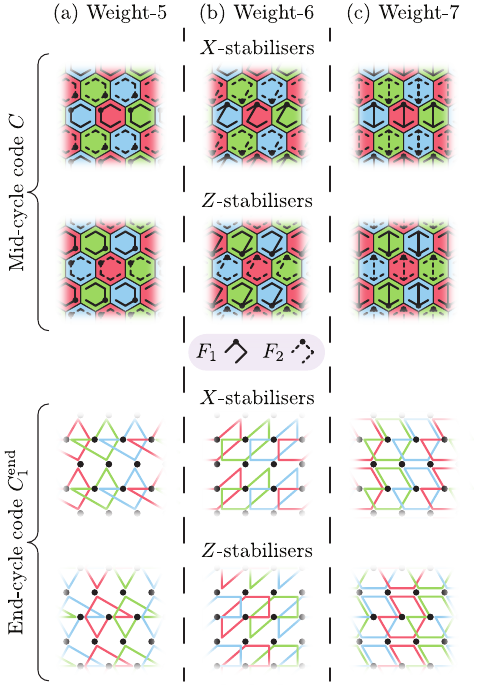}
    \caption{Hex-grid morphing circuits for the colour code on an infinite 2D plane. We consider three distinct end-cycle codes, which we label according to their stabiliser weight. For each end-cycle code, we show the contraction tree diagram of one morphing circuit with the corresponding end-cycle codes. All contraction circuits $F_i$ are depth-3 and can be implemented using purely hex-grid connectivity, and additionally have the property that $F_{2}=H^{\otimes n}F_{1} H^{\otimes n}$. We display the $X$- and $Z$-stabilisers of the end-cycle codes as a shape whose vertices correspond to the qubits that the stabiliser has support on; the colour of each stabiliser is inherited from the colouring of the mid-cycle code. The geometric location of the qubits in the end-cycle code $C_{1}^{\mathrm{end}}$ have been chosen for ease of visualisation: specifically, each end-cycle qubit is placed at the ``average'' location of each pair of qubits that participates in a CNOT in the final layer of the contraction circuit $F_{1}$. The second end-cycle code is related to the first by $C_{2}^{\mathrm{end}}=H^{\otimes n}C_{1}^{\mathrm{end}}H^{\otimes n}$. The weight-5 morphing circuit and end-cycle code have already been presented in Ref.~\cite{Gidney2023}, where the end-cycle code is called the ``pyramid'' code.}
    \label{fig:colour_code_end_cycle}
\end{figure}
\subsection{Case Study: the Hex-grid Colour Code}\label{subsec:colour_code}

\begin{table*}[]
    \caption{Comparison between the optimised morphing circuits for the planar hexagonal lattice colour code and existing colour code and surface code circuits. The triangular weight-6 and diamond weight-7 morphing circuits are novel constructions introduced in this section, with ``triangular'' and ``diamond'' referring to the shape of the colour code boundaries and ``weight-6/7'' referring to the weight of the end-cycle code stabilisers as shown in~\cref{fig:colour_code_end_cycle}. For each circuit, we write the \textit{total} number of data and ancilla qubits $n_{\mathrm{tot}}$ required to achieve an \textit{end-cycle} code distance of $d^{\mathrm{end}}$, the number of logical qubits $k$, the circuit-level distance $d_{\mathrm{circ}}$, and the connectivity requirements. The circuit-level distance in each case is supported by numerical calculations for $d\leq 9$ and is an upper-bound in the large $d$ limit, see~\cref{sec:code_boundary_optimisation,sec:colour_code_numerics} for more details about our methods.}
    \centering
    \begin{tabular}{c c|c|c|c|c}
        \multicolumn{2}{c|}{Colour Code Circuit} & $n_{\mathrm{tot}}$ & $k$ & $d_{\mathrm{circ}}$ & Connectivity \\
        \hline
        \multicolumn{2}{c|}{Superdense~\cite{Gidney2023}} & $3(d^{\mathrm{end}})^{2}/2 - 1/2$ & 1 & $d^{\mathrm{end}}$ & Hex-square grid\\
        \hline
        \multirow{4}{*}{Morphing:} & Gidney \& Jones~\cite{Gidney2023} & $3(d^{\mathrm{end}})^{2}+O(d^{\mathrm{end}})$ & 1 & $d^{\mathrm{end}}$ & Hex grid\\
         & Triangular Weight-6 & $3(d^{\mathrm{end}})^{2}/2 + O(d^{\mathrm{end}})$ & 1 & $3d^{\mathrm{end}}/4+O(1)$ & Hex grid\\
         & Triangular Weight-6 & $2(d^{\mathrm{end}})^{2} + O(d^{\mathrm{end}})$ & 1 & $d^{\mathrm{end}}$ & Hex grid\\
         & Diamond Weight-7 & $3(d^{\mathrm{end}})^{2} + O(d^{\mathrm{end}})$ & 2 & $d^{\mathrm{end}}$? & Hex grid\\
        \hline\hline
        \multicolumn{2}{c|}{Hex-grid surface code~\cite{McEwen23}} & $2(d^{\mathrm{end}})^{2}-1$ & 1 & $d^{\mathrm{end}}$ & Hex grid\\
    \end{tabular}
    \label{tab:colour_code_circuits}
\end{table*}

In this section we apply the optimisation technique from~\cref{subsec:optimising_boundaries} to the 2D hexagonal colour code. A normal bare-ancilla syndrome extraction circuit requires each ancilla qubit to connect to 6 data qubits, but morphing reduces this to 3 which we refer to as hex-grid connectivity. This is a lower connectivity requirement than the superdense colour code circuits~\cite{Gidney2023}, in which data qubits have degree-3 connectivity but ancilla qubits have degree-4. A morphing circuit for the colour code using hex-grid connectivity was first constructed by Gidney and Jones~\cite{Gidney2023}, but since the boundary geometry was optimal for the mid-cycle code instead of the end-cycle code, it requires significantly more physical qubits to achieve a given circuit-level distance than both the surface code and the superdense colour code circuits. We will show that optimising instead for the end-cycle code can lead to significantly more qubit-efficient morphing circuits, see a summary of our results in~\cref{tab:colour_code_circuits}.
%We begin with a brief overview of our results which are summarised in~\cref{tab:colour_code_circuits}. 
Ideally, we would like to find a circuit that (1) uses hex-grid connectivity, (2) uses a minimal number of physical qubits per logical qubit to achieve a given circuit-level distance, and (3) can implement the full Clifford group transversally in the mid-cycle code. Unfortunately, we were not able to find a circuit that achieves all three of these goals simultaneously: the two ``triangular weight-6'' morphing circuits (which we will define shortly) require more physical qubits than desired, and the ``diamond weight-7'' morphing circuit encodes two logical qubits and therefore cannot implement all Clifford gates with mid-cycle transversal gates. Nevertheless, because our newly proposed circuits use a purely hexagonal connectivity, they may still be of interest in various applications. Here, in the main text, we will focus on the overview of how the circuits are designed, and we leave more precise circuit specifications and numerical benchmarks to~\cref{sec:code_boundary_optimisation,sec:colour_code_numerics}.

Following the method in~\cref{fig:optimisation_flowchart}, we start by considering a colour code $C_{\mathrm{inf}}$ on the infinite 2D plane. We consider three morphing circuits that give rise to three distinct end-cycle codes, each of which is shown in~\cref{fig:colour_code_end_cycle}. We call these the weight-5, weight-6, and weight-7 morphing circuits and end-cycle codes, referring to the weight of the end-cycle stabilisers. The weight-5 end-cycle code has already been referred to as the ``pyramid code'' in Ref.~\cite{Gidney2023}.

To illustrate how our optimisation procedure works, consider the weight-6 morphing circuit and its corresponding infinite end-cycle code. The weight-6 end-cycle code is not a colour code -- it is not self-dual and does not have any depth-1 transversal Clifford gates -- but it does have the property that every $X$ or $Z$ error causes three defects (parity check flips), one of which is red, one green, and one blue. This means that the weight-6 end-cycle code has a logical structure that resembles the colour code, with logical strings that can travel along a triangular lattice through the end-cycle code. We can therefore impose optimised \textit{triangular} boundary conditions on the weight-6 end-cycle code, giving a code family encoding one logical qubit, see~\cref{fig:half_shifted_hex}. To obtain the corresponding morphing circuit, we then transfer the boundary conditions found for the end-cycle code onto the mid-cycle code. As is shown in~\cref{fig:half_shifted_hex}, this leads to non-standard and suboptimal boundaries for the mid-cycle colour code, taking the shape of a \textit{right-angled} triangle instead of an equilateral triangle. This is compensated for by the fact that the mid-cycle code has more qubits than the end-cycle code, and we find that the mid-cycle code actually has the same distance as the end-cycle code with this boundary geometry. Finally, we use the techniques from~\cref{subsec:boundaries} to obtain the corresponding morphing circuit. The total number of data and ancilla qubits $n_{\mathrm{tot}}$ required to achieve a given end-cycle code distance $d^{\mathrm{end}}$ scales asymptotically as $n_{\mathrm{tot}}=3d^{\mathrm{end}}/2 + O(d)$.

Unfortunately, the circuit-level distance of this triangular weight-6 morphing circuit is limited by a circuit-level error with weight $3d^{\mathrm{end}}/4+O(1)$, which we show explicitly in~Appendix~\ref{sec:colour_code_numerics}. To overcome this issue, one can once again adjust the triangular boundary conditions to account for the additional error mechanism. This requires additional qubits, as shown in~\cref{fig:weight_6_errors} in~Appendix~\ref{sec:weight_6_details}, but now achieves a circuit-level distance equal to the end-cycle distance. This circuit, while not as qubit efficient as the superdense colour code circuit, uses only a hex-grid connectivity, and can implement transversal logical Clifford gates in the mid-cycle code. These properties may make the weight-6 triangular morphing circuits attractive for magic state cultivation on devices with strictly hex-grid connectivity.

Having optimised the boundary conditions for the weight-6 end-cycle code, we now turn our attention to the weight-5 and weight-7 end-cycle codes. It turns out that the optimisation problem for the weight-5 and weight-7 end-cycle codes is remarkably similar, so to simplify the investigation we focus on the weight-7 end-cycle codes due to the symmetry of the morphing circuit that implements it. Unlike the weight-6 end-cycle code, logical strings travel along a square lattice through the end-cycle code instead of a triangular lattice. This fundamentally changes the optimisation problem for the end-cycle boundary conditions: optimised triangular boundary conditions use \textit{the same} number of physical qubits as optimised \textit{diamond} boundary conditions, despite the latter encoding two logical qubits instead of one. The resulting diamond weight-7 morphing circuit uses $n_{\mathrm{tot}}=3(d^{\mathrm{end}})^2+O(d^{\mathrm{end}})$ to encode \textit{two} logical qubits, which is more qubit-efficient than the fault-tolerant triangular weight-6 morphing circuit. However, because it encodes more than one logical qubit, we cannot implement the full two-qubit Clifford group using mid-cycle transversal gates, and we therefore would need to use lattice surgery to implement logical operations. We conjecture that the circuit-level distance is the same as the end-cycle code distance, but see~Appendix~\ref{sec:weight_7_details} for a detailed discussion of the evidence we have for this.

For all of the circuits discussed here, we also conducted numerical memory experiments to confirm their viability as quantum memories, see~Appendix~\ref{sec:colour_code_numerics} for details.

\begin{figure*}
    \centering
    \includegraphics[width=\linewidth]{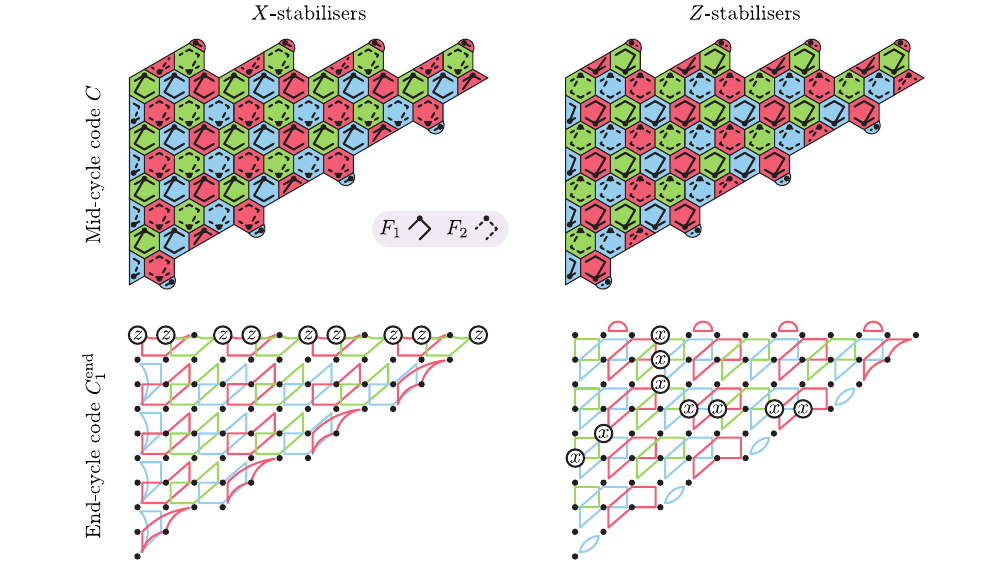}
    \caption{The optimised triangular weight-6 morphing colour code circuit for $d^{\mathrm{end}}=9$. The mid-cycle code is derived by transferring the optimal boundary geometry (see Appendix~\ref{sec:code_boundary_optimisation}) of the weight-6 end-cycle code from~\cref{fig:colour_code_end_cycle}(b), padded using the methods from~\cref{subsec:boundaries}. The result is a mid-cycle colour code on a right-angled-triangular patch instead of the standard equilateral triangle. In the end-cycle code $C_{1}^{\mathrm{end}}$, the qubits are placed on vertices, and the stabilisers are represented by shapes with corners corresponding to the qubits that the stabiliser is supported on. The other end-cycle code is given by $C_{2}^{\mathrm{end}}=H^{\otimes n}C_{1}^{\mathrm{end}}H^{\otimes n}$. The parameters of both end-cycle codes are $[\![73,1,9]\!]$, while the mid-cycle code is $[\![145,1,9]\!]$. We also show example minimum-weight $Z$- and $X$-logical operators in the end-cycle code supported on encircled $z$ and $x$ letters respectively.}
    \label{fig:half_shifted_hex}
\end{figure*}

\subsection{Adding More Contraction Rounds: Bacon-Shor Code Example}\label{subsec:adding_rounds}

We discuss one other technique that can be used to improve the end-cycle distance, namely to increase the number of contraction rounds in the morphing circuit as in~\cref{fig:overview}(c). Intuitively, one can expect this to work because using more contraction rounds means that fewer stabilisers need to be measured in each contraction round, reducing the number of ancilla qubits in each round. As a result, the end-cycle code will have more qubits, allowing it to achieve a larger distance. The disadvantage of this is that stabilisers are measured less frequently, incurring a time-overhead that we investigate in more detail in~\cref{sec:optimising_time_like}.

The \eczoo[Bacon-Shor code]{bacon_shor} provides a simple and clear example of how this effect can manifest. The Bacon-Shor code is a subsystem code defined on a square grid of qubits with $L_{X}$ columns and $L_{Z}$ rows. On each vertically-oriented edge of the square grid there is a weight-2 $X$-gauge generator, and on each horizontal edge there is a weight-2 $Z$-gauge generator. The stabilisers of the code are then complete \textit{rows} of $X$-gauge operators (so double rows of $X$ Pauli operators), and complete \textit{columns} of $Z$-gauge operators (so double columns of $Z$ Pauli operators). $X$-logical operators are then given by any single row of $X$ operators such that the $X$-distance is $d_{X}=L_{X}$, and $Z$-logical operators are any single column of $Z$ operators such that $d_{Z}=L_{Z}$.

Contracting a given gauge operator is trivial: it requires only a single CNOT gate followed by a measurement, and by the arguments in~\cref{subsec:data_ancilla}, the direction of this CNOT gate does not affect the end-cycle code. On the other hand, the choice of contracting subsets is more subtle because we also need to ensure that gauge operators are contracted in an order that ensures that the stabilisers can be inferred, as explained in~\cref{subsec:general_morphing}. One easy way to guarantee this is to contract all the $X$-gauge operators in the first contraction round, and then all of the $Z$-gauge operators in the other contraction round, so $J=2$; but doing this would reduce our $Z$-distance to 1 in the first end-cycle code because there would be no $X$-gauge operators left in the end-cycle code! To achieve a non-trivial end-cycle distance we therefore need to use more than one contraction round for the $X$-gauge operators (and similarly $Z$-gauge operators), such that the morphing circuit has $J> 2$.

Consider the contraction rounds in which we contract $X$-gauge operators. For simplicity, we will contract all the $X$-gauge operators in a given row during the same contraction round, so that now all we have to do is choose which contraction round to contract each row of $X$-gauge operators in. It is not too hard to see that the end-cycle code is another Bacon-Shor code but with fewer rows of qubits: the end-cycle grid of qubits now only has $(L_{X}-N_{\mathrm{rows}})$ rows of qubits, where $N_{\mathrm{rows}}$ is the number of rows of $X$-gauge operators that are contracted in a given contraction round. Assuming $L_{X}=L_{Z}=d$ is the mid-cycle code distance, the end-cycle code distance is now $d_{j}^{\mathrm{end}}=d-N_{\mathrm{rows}}$. We can therefore increase the end-cycle distance by decreasing the number of rows contracted in a given round, which necessitates increasing the number of contraction rounds we use to contract all of the rows. The same story applies to the $Z$-gauge operators as well.

Although increasing the number of contraction rounds here improves the end-cycle distance, it comes at the cost of measuring the stabilisers less frequently, increasing the time-overhead of lattice surgery operations, as we discuss in the following section. However, it is not always straightforward to find examples of codes in which the end-cycle code distance does indeed increase by increasing the number of contraction rounds. The only example besides the Bacon-Shor code that we are aware of is given by the ``diamond'' circuits for the surface code~\cite{debroy2025diamond}, where a four-round morphing circuit for the \eczoo[subsystem surface code]{subsystem_surface} uses fewer physical qubits to achieve a given end-cycle code distance than the standard two-round morphing circuit for the (stabiliser) surface code.

\section{Circuit-Level Distance for Stability Experiments using Morphing Circuits}\label{sec:optimising_time_like}

%So far in this work we have only discussed the error-correction properties of morphing circuits in the context of memory experiments.spatial errors: i.e.~those that can be viewed as an instantaneous flip to a logical observable. However, it is important also to

We consider how frequent the stabilisers of a given code are measured in a morphing circuit compared with a bare-ancilla circuit. The frequency of stabiliser measurements is important in two contexts: in lattice surgery, where multiple stabiliser measurement errors can lead to an error in the logical measurement outcome; and in windowed decoding, where measurement errors can cause sparse spatial errors to accumulate into a logical error across different decoding windows. Because of these error mechanisms, lattice surgery operations must be deep enough and the decoding window must be wide enough for the corresponding circuit-level distance to be $\Omega(d)$.

We already argued intuitively in~\cref{subsec:two_round_morphing_detectors} that (in the absence of single-shot error correction) any morphing circuit with two contraction rounds has \textit{the same} stabiliser measurement frequency as any \textit{normal} circuit in which each stabiliser generator is measured exactly once per measurement round. We make this rigorous here by proving that the circuit-level distance of a stability experiment for a two-round morphing circuit is given exactly by the number of measurement rounds $T$ in the stability experiment. At the same time, we explore how increasing the number of contraction rounds $J$ reduces the circuit-level distance of a $T$-round stability experiment, implying that deeper lattice surgery operations are required to achieve a circuit-level distance $d$. The mathematical results are stated in Prop.~\ref{prop:time_like_lower_bound} and Prop.~\ref{prop:time_like_upper_bound}, with the consequence for $J=2$ in \cref{eq:J=2_time-like_distance} and tighter results for CSS codes in \cref{subsec:tightCSS}. The results for $J>2$ are relevant because morphing circuits with $J>2$ may be desirable both to improve the end-cycle code distance as discussed in~\cref{subsec:adding_rounds}, as well as for schemes that use morphing to combat qubit and/or coupler dropout~\cite{debroy2024:luci,higgott2025handling,wolanski2026automated,anker2025:optimized}. Finally, we consider the interplay of morphing with single-shot error correction, starting in~\cref{subsec:3D_TC} with the concrete example of the 3D toric code and finishing in~\cref{subsec:single_shot} by showing in general that the single-shot properties of mid-cycle codes are preserved in morphing circuits.

\subsection{Additional Conventions}

We introduce some additional conventions that we use throughout this section. We will assume that the experimental platform we are using has access to (unconditional) resets on ancilla qubits, we will comment on the differences that can arise when resets are not present in~\cref{sec:discussion}. In morphing circuits, we will label the time-steps during which the morphing circuit is at the end-cycle code with non-negative integers $t=0,1,\dots,T$, and time-steps during which the morphing circuit is in the mid-cycle code with half-integers $t=1/2,3/2,\dots,T-1/2$. By convention, all of our experiments begin at time $t=0$ in the end-cycle code $C_{J}^{\mathrm{end}}$. After the end-cycle data qubits have been initialised in some chosen state, we will begin by performing the resets $R_{J}$ on the ancilla qubits and then proceed with the circuit $F_{J}^{\dag}$ and so forth. This convention means that at the end-cycle time $t$, we are in the end-cycle code $C_{t\,\mathrm{mod}\,J}^{\mathrm{end}}$ which is notationally convenient. The last operation of the QEC part of the experiment will be the measurements $M_{T\,\mathrm{mod}\,J}$, at which point we will be in the end-cycle code ${C}_{T\,\mathrm{mod}\,J}^{\mathrm{end}}$ and the end-cycle data qubits can be measured in a chosen basis.

\subsection{Detectors in Morphing Circuits: Regular and Meta-Check}\label{subsec:morphing_detectors}

\begin{figure*}
    \centering
    \includegraphics[width=\linewidth]{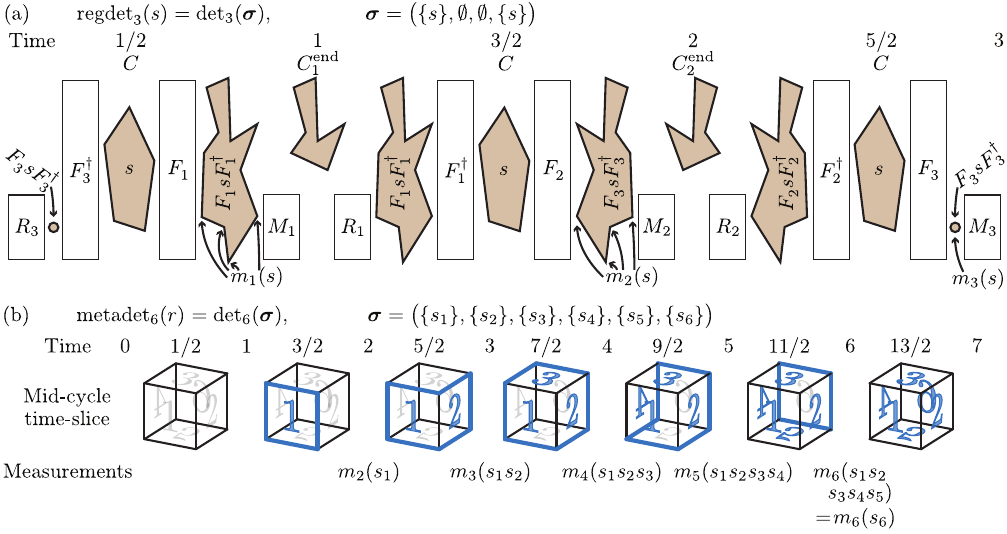}
    \caption{Diagrams representing the two types of detectors in a morphing circuit, time moves from left to right in both of these figures. (a) A regular detector $\mathrm{regdet}_{3}(s)$ corresponding to a stabiliser $s$ that is in the contracting subset $S_{3}$ in a three-round morphing circuit. The detecting region at each of the mid-cycle time-steps corresponds exactly to the stabiliser $s$, which is abstractly represented at each time step by an arbitrary polygon whose vertices represent qubits. The measurements $m_t(s)$ that are included in the detector are simply those that $s$ ends up on while it is expanding at time-steps $t=1$ and $t=2$, and the final measurement that $s$ contracts onto in time-step $t=3$. The duration of this detector is the same as the number of rounds $J$ in the morphing circuit, here $J=3$. (b) A meta-check detector $\mathrm{metadet}_{6}(r)$ derived from a 3D toric code with qubits on the edges of a cube and $Z$-checks on faces, where each of the six $Z$-stabilisers $s_{i}$ that make up the redundancy $r$ is contracted in a different contraction round of a six-round morphing circuit. The mid-cycle time-slice of the meta-check detector at each time step is shown in blue, where the highlighted blue numbers represent the stabilisers that make up the time-slice in~\cref{eq:detector_mid_cycle_time_slice} and the highlighted blue edges represent the qubits the detector time-slice has support on. We also list the end-cycle measurements that make up the meta-check detector. The duration of this detector is $J-1=5$.}
    \label{fig:detector_figures}
\end{figure*}

We describe how to mathematically construct detectors in morphing circuits, going beyond the more intuitive $J=2$ case discussed in~\cref{subsec:two_round_morphing_detectors}. In general, a detector is a set of measurements in a circuit whose product is deterministic in the absence of errors~\cite{Gidney21,McEwen23}. A detector is sensitive to errors in a detection region which, for quantum LDPC codes, we wish to be local in space-time~\cite{delfosse2023}, so that only a constant number of circuit-level error mechanisms can flip a given detector.

We begin by describing an extremely general way of constructing detectors in morphing circuits for stabiliser codes, before explaining how to use this formalism to construct a local generating set of detectors. To construct detectors in morphing circuits, we follow the formalism of Ref.~\cite{McEwen23} by constructing \textit{detecting regions} in the morphing circuit. To summarise, a detecting region is the back-propagation of the Pauli operators corresponding to all of the measurements contained within the detector. Importantly, the ``time-slice'' of a detecting region at a given time in a circuit corresponds to a stabiliser of the codespace at that time in the circuit. We construct detectors from the ``mid-cycle-code'' perspective by constructing a valid set of the mid-cycle time-slices of a detecting region, and then inferring the corresponding detector consistent with these time-slices, as we now explain.

%We begin by deriving an important general property of detectors that we will use to make our definition. 
At some mid-cycle time-step $t-1/2$, let us suppose that a detector has a detecting region time-slice given by the mid-cycle stabiliser $s$. What does this imply about the detector? First, we know that the time-slice of the detecting region immediately before the next measurements $M_{t\,\mathrm{mod}\,J}$ at time $t$ is given by $F_{t\,\mathrm{mod}\,J}^{\vphantom{\dag}}sF_{t\,\mathrm{mod}\,J}^{\dag}$. This is enough to determine which of the measurements in $M_{t\,\mathrm{mod}\,J}$ are contained within the detector: it is simply all of the measurements in $M_{t\,\mathrm{mod}\,J}$ that overlap with the pre-measurement time-slice of the detecting region (see e.g.~\cref{fig:two_round_detector}), which we write as
\begin{equation}
    m_{t}(s)=M_{t\,\mathrm{mod}\,J}\cap (F_{t\,\mathrm{mod}\,J}^{\vphantom{\dag}}sF_{t\,\mathrm{mod}\,J}^{\dag}),
    \label{eq:mts}
\end{equation}
where it is understood that $m_t(s)$ is a set of measurements.
Moreover, in the end-cycle code $C_{t\,\mathrm{mod}\,J}^{\mathrm{end}}$, the detecting region time-slice is given by the Gottesman-Knill update of $F_{t\,\mathrm{mod}\,J}^{\vphantom{\dag}}sF_{t\,\mathrm{mod}\,J}^{\dag}$ given by projecting it onto the data qubits of $C_{t\,\mathrm{mod}\,J}^{\mathrm{end}}$.

Now, let us suppose that the support of this same detector at the \textit{next} mid-cycle time-step $t+1/2$ is given by another mid-cycle stabiliser $s'$. Following the same logic as above but working backwards, we know that the end-cycle time-slice of the detecting region at time $t$ is given by the projection of $F_{t\,\mathrm{mod}\,J}^{\vphantom{\dag}}s'F_{t\,\mathrm{mod}\,J}^{\dag}$ onto the data qubits of $C_{t\,\mathrm{mod}\,J}^{\mathrm{end}}$. Because we want this to be the same detector as in the previous paragraph this projection must coincide with the projection of $F_{t\,\mathrm{mod}\,J}^{\vphantom{\dag}}sF_{t\,\mathrm{mod}\,J}^{\dag}$ onto the same data qubits in $C_{t\,\mathrm{mod}\,J}^{\mathrm{end}}$. Of course, this is trivially satisfied if $s=s'$, but we actually have more flexibility than this: $s$ and $s'$ can differ by any number of stabilisers that are \textit{contracting} at time $t$, because these stabilisers are necessarily measured and do not affect the end-cycle detecting region time-slice at time $t$.

The above considerations motivate the following general definition of a detector in a morphing circuit with the ideas visually illustrated in~\cref{fig:detector_figures}.
We begin with a ``detection sequence'' $\vect{\sigma}=(\sigma_{t_1},\sigma_{t_1+1},\dots,\sigma_{t_2})$, where each element $\sigma_{t}$ is a (possibly empty) subset of the mid-cycle contracting stabilisers $S_{t\,\mathrm{mod}\,J}$ in that contraction round, such that the product of all the stabilisers in the sequence $\vect{\sigma}$ is the identity, i.e.
\begin{equation}\label{eq:detector_identity_product}
    \prod_{t=t_{1}}^{t_{2}}\prod_{s\in \sigma_{t}}s=I.
\end{equation}
To define detectors of finite duration, $t_2-t_1$ in \cref{eq:detector_identity_product} will be constant. For example, we can construct a regular detector by setting $\sigma_{t_{1}}=\sigma_{t_{2}}=\{s\}$, where $t_{1}$ and $t_{2}$ are two time-steps in which $s$ is contracting, and $\sigma_{t}=\emptyset$ for all other $t$. The detection sequence trivially satisfies~\cref{eq:detector_identity_product}, because it contains the same stabiliser $s$ twice.

With a detection sequence satisfying~\cref{eq:detector_identity_product}, we construct the detector $\mathrm{det}_{t_2}(\vect{\sigma})$ such that, at any mid-cycle time-step $t-1/2$, the mid-cycle detecting region time-slice is given by the product of all the stabilisers in $\vect{\sigma}$ up to time $t$, i.e.
\begin{equation}\label{eq:detector_mid_cycle_time_slice}
    \prod_{t'=t_{1}}^{t-1}\prod_{s\in \sigma_{t'}}s.
\end{equation}
With this, the measurements that make up the detector $\mathrm{det}_{t_2}(\vect{\sigma})$ are given by forward-propagating each of the mid-cycle time-slices from~\cref{eq:detector_mid_cycle_time_slice}, i.e.
\begin{equation}\label{eq:detector_definition}
    \mathrm{det}_{t_2}(\vect{\sigma})=\bigcup_{t=t_{1}+1}^{t_{2}}m_{t}\bigg(\prod_{t'=t_{1}}^{t-1}\prod_{s\in \sigma_{t'}}s\bigg).
\end{equation}
with $m_t(s)$ defined in~\cref{eq:mts}. Such a definition works because any pair of consecutive mid-cycle time-slices of the detecting region differ only by a product of \textit{contracting} stabilisers in the set $\sigma_{t}$. 

We use this general definition to construct two types of detectors: {\em regular} detectors $\mathrm{regdet}_{t_{2}}(s)$ that are defined for every mid-cycle stabiliser generator $s$ (shown in~\cref{fig:detector_figures}(a)), and {\em meta-check} detectors $\mathrm{metadet}_{t}(r)$ that are defined for each mid-cycle redundancy $r$ (shown in~\cref{fig:detector_figures}(b)). Given a mid-cycle stabiliser generator $s$ and a time-step $t_{2}$ in which $s$ is contracting, we can define a regular detector $\mathrm{regdet}_{t_{2}}(s)$ whose detecting region spans from $t_{1}$ to $t_{2}$, where $t_{1}$ is the most-recent time-step before $t_{2}$ in which $s$ was also contracting. The corresponding detection sequence $\vect{\sigma}$ is given by setting $\sigma_{t_{1}}=\sigma_{t_{2}}=\{s\}$ and $\sigma_{t'}=\emptyset$ for all other $t'$. The detection sequence trivially satisfies~\cref{eq:detector_identity_product} because it contains the same stabiliser $s$ twice. This means that the mid-cycle detecting region time-slice is $s$ for all mid-cycle time-steps between $t_{1}$ and $t_{2}$, and otherwise is the identity. In a $J$-round morphing circuit in which each stabiliser generator is measured exactly once, then the duration of any regular detector is exactly $J$ measurement rounds.

Moreover, given a mid-cycle redundancy $r$ -- which is a set of mid-cycle stabiliser generators whose product is the identity -- and a time-step $t_{2}$ in which at least one stabiliser in $r$ is contracting, we can define a meta-check detector $\mathrm{metadet}_{t_{2}}(r)$ as follows. Unlike the regular detectors, here we include each stabiliser $s\in r$ only once in the detection sequence $\vect{\sigma}$, because we are already guaranteed that $\prod_{s\in r} s=I$. In particular, we include each stabiliser $s\in r$ in $\vect{\sigma}$ in the most-recent time-step $t'\leq t_{2}$ in which it is contracting. This means that the mid-cycle detecting region time-slice is no longer the same in every time-slice, as shown in~\cref{fig:detector_figures}(b). In a $J$-round morphing circuit in which each stabiliser generator is measured exactly once, the duration of any meta-check detector is exactly $J-1$ measurement rounds.

Finally, we note a fact that will be useful for later: that meta-check detectors corresponding to the same redundancy $r$ in different time-steps are related by regular detectors via
\begin{equation}\label{eq:meta_check_detector_relation}
    \mathrm{metadet}_{t}(r)\oplus \mathrm{metadet}_{t+1}(r)=\bigoplus_{s\in r_{t+1}}\mathrm{regdet}_{t+1}(s),
\end{equation}
where here we use the $\oplus$ symbol to denote that a measurement is included in the detector $\bigoplus_{d\in D}d$ if and only if it is included an odd number of times in the detectors in $D$.

Although our above description only applies to stabiliser codes, it is straightforward to generalise this to subsystem and dynamical codes by adding the restriction that the mid-cycle detecting region time-slice at each mid-cycle time-step $t-1/2$ commutes with all of the contracting operators in both $S_{(t-1)\,\mathrm{mod}\,J}$ and $S_{t\,\mathrm{mod}\,J}$. This guarantees that the resulting detecting region commutes with all the measurements and resets that it overlaps with, and therefore defines a deterministic product of measurements.

\subsection{Circuit-level Distance of Stability Experiments}\label{subsec:time-like_distance}

Now that we know how to define detectors in morphing circuits, we can bound the circuit-level distance of stability experiments in morphing circuits. Stability experiments were introduced in Ref.~\cite{gidney2022stability}, and the most well-known example of a stability experiment is a modified surface code patch where all four boundaries have the same Pauli type. This removes the logical qubit of the surface code, but introduces a redundancy $r$ which can be used to define the stability experiment. However, any code $C$ with at least one redundancy can define a stability experiment, including toric codes and more general Abelian 2BGA codes.

In a stability experiment one begins by initialising each data qubit in an arbitrary single-qubit state; it does not matter what state as long as we do not use any information about this state in the decoding protocol. Then, we implement $T$ measurement rounds of some syndrome extraction circuit. Afterwards, it does not matter what we do with the qubits, so long as the protocol also does not use any information about it. In this circuit, regardless of the initialisation, there are a number of regular and meta-check detectors that are deterministic in the absence of errors.
%because their detecting regions do not come into contact with the initial resets or final measurements. 

To perform the stability experiment, we promote each meta-check detector to the role of a logical \textit{observable}: that is, we do not pass information about its ideal value to the decoder and instead use it as a ``test'' to see whether the decoder can correctly predict its deterministic value. To make a clear distinction we write $\mathrm{obs}_{t}(r)$ instead of $\mathrm{metadet}_{t}(r)$ for each meta-check observable, and we refer to the ``observing region'' of $\mathrm{obs}_{t}(r)$ instead of the detecting region. Moreover, note from~\cref{eq:meta_check_detector_relation} that any pair of meta-check observables $\mathrm{obs}_{t}(r)$ and $\mathrm{obs}_{t'}(r)$ defined from the same redundancy $r$ are related by some regular detectors. Any decoder will therefore predict the same outcome for $\mathrm{obs}_{t}(r)$ and $\mathrm{obs}_{t'}(r)$, so they must represent two representatives of the same logical observable in the stability experiment.

We are interested in the circuit-level distance of such a depth-$T$ stability experiment. This circuit-level distance, written as $d_{\mathrm{stab}}(T)$, is defined as the smallest number of circuit-level errors that flip at least one of meta-check observables without flipping any regular detectors. 
For example, in the surface code stability experiment with a normal bare-ancilla syndrome extraction circuit, we simply have $d_{\mathrm{stab}}(T)=T$. With this, the number of measurement rounds $T$ in a lattice surgery experiment should be such that the corresponding stability experiment has $d_{\mathrm{stab}}(T)\geq d_{\mathrm{circ}}$, where $d_{\mathrm{circ}}$ is the usual \textit{spatial} circuit-level distance of the syndrome extraction circuit for the code $C$.

With these definitions, we can present our upper- and lower-bounds on $d_{\mathrm{stab}}(T)$. Our results are written in fairly general terms to allow for cases in which stabilisers are measured more than once (in the cycle of length $J$) in the morphing circuit, and to allow for subsystem codes in which stabilisers are composed of a product of gauge operators.
\begin{proposition}\label{prop:time_like_lower_bound}
    Consider a depth-$T$ stability experiment for a morphing circuit in which the maximum duration of its regular detectors, i.e.~the temporal length of its corresponding detection regions, is $\delta$ measurement rounds. Then, the circuit-level distance of this experiment satisfies the lower-bound
    \begin{equation}\label{eq:time_like_lower_bound}
        d_{\mathrm{stab}}(T)\geq \bigg\lfloor\frac{T}{\delta-1}\bigg\rfloor.
    \end{equation}
\end{proposition}
\begin{proof}[Proof Sketch]
    The proof uses the fact that the maximum duration of an observing region in a stability experiment is $\delta-1$ to show that for any logical observable in a stability experiment, there exists a set of (at least) $\lfloor T/(\delta-1)\rfloor$ non-overlapping observing regions that each represent the same logical observable. Since any circuit-level logical error must flip all of the representatives of a given logical observable, such an error must have weight at least $\lfloor T/(\delta-1)\rfloor$, giving the result (see Appendix~\ref{sec:proofs} for the full proof).
\end{proof}

\begin{proposition}\label{prop:time_like_upper_bound}
    Consider a depth-$T$ stability experiment for a morphing circuit with mid-cycle code $C$. Then, let $N(s)$ be the number of times that a mid-cycle stabiliser generator $s$ of $C$ is contracted in an end-cycle time-step $t=0,\dots,T$, with the first ($t=0$) and last ($t=T$) contraction rounds only counting as half in $N(s)$. Then, the circuit-level distance of this experiment satisfies the upper-bound
    \begin{equation}\label{eq:time_like_upper_bound}
        d_{\mathrm{stab}}(T)\leq 2\min_{s\in R}N(s),
    \end{equation}
    where $R$ is the set of stabilisers that are contained within at least one redundancy of the mid-cycle code.
\end{proposition}
\begin{proof}[Proof Sketch]
    The proof works by constructing an undetectable logical error. In particular, for any stabiliser generator $s$ that participates in at least one redundancy, we construct an error by placing a measurement error on every measurement upon which $s$ is contracted, and a reset error on every following reset upon which $s$ is contracted. In Appendix~\ref{sec:proofs} we prove that this is an undetectable logical error and has weight $2N(s)$, which gives the upper-bound in~\cref{eq:time_like_upper_bound} when considered for all stabiliser generators that participate in at least one redundancy.
\end{proof}

Because Prop.~\ref{prop:time_like_lower_bound} and \ref{prop:time_like_upper_bound} summarise the bounds we have found for a wide variety of morphing circuits, we will unpack its consequences for morphing circuits in the simpler case where every stabiliser generator is measured exactly once in the morphing circuit. First, we consider the simplest case where the morphing circuit has two rounds. In this case, the maximum duration of the regular detectors is $\delta=2$, and moreover $N(s)=T/2$ for all stabilisers $s$. Therefore, the upper- and lower-bounds coincide, giving
\begin{equation}\label{eq:J=2_time-like_distance}
    d_{\mathrm{stab}}(T)=T.
\end{equation}
This is significant because it makes formal the time-like equivalence between two-round morphing circuits and complete bare-ancilla circuits that we already discussed intuitively in~\cref{subsec:two_round_morphing_detectors}. \cref{eq:J=2_time-like_distance} applies to almost all of the examples of morphing circuits we have explicitly discussed so far, including the hex-grid surface code of Ref.~\cite{McEwen23}, the colour code circuits in~\cref{subsec:colour_code}, and the Abelian 2BGA code circuits in~\cref{subsec:Abelian_2BGA}.

\begin{figure}
    \centering
    \includegraphics[width=\linewidth]{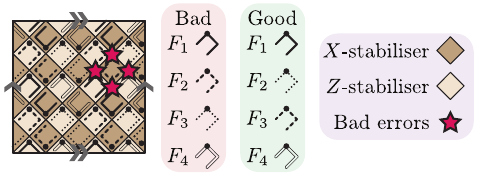}
    \caption{Two four-round toric code morphing circuits with different circuit-level stability distances $d_{\rm stab}(T)$. There are two redundancies in the toric code given by the product of all the $X$- and all the $Z$-stabilisers, each of which defines one observable in the stability experiment. The only difference between the two circuits is the order in which the contracting subsets are implemented. The ``bad'' ordering has circuit-level stability distance $\lfloor T/3\rfloor$, with a minimum-weight string of ``bad'' errors that can be constructed from the four mid-cycle $X$-errors indicated in the figure. In particular, the bad error string consists of placing one mid-cycle $X$-error every 3 time-steps, starting with one error at the right spatial error location, then the top, then the left, and then the bottom, and so on. Each single-qubit mid-cycle error triggers two detectors $\mathrm{regdet}_{t_{1}}(s_{1})$ and $\mathrm{regdet}_{t_{2}}(s_{2})$ corresponding to the two toric code $X$-stabilisers that it overlaps with. However, due to the time-steps in which the errors occur, these two detectors that are triggered are separated by 3 measurement rounds in time, i.e.~$t_{2}=t_{1}+3$. Moreover, the error string flips any observable corresponding to the $Z$-stabiliser redundancy, and therefore constitutes a logical error. By reordering the contracting subsets into the ``good'' ordering, this error mechanism is no longer available, and we have numerically found that the ``good'' circuit has time-like distance $\lfloor T/2\rfloor-1$.}
    \label{fig:time-like_errors}
\end{figure}

A more subtle picture emerges when we consider morphing circuits with $J>2$ contraction rounds, even while still assuming each stabiliser is measured exactly once per contraction round. Here, the upper- and lower-bounds no longer coincide, and even scale at different rates with respect to $T$:
\begin{equation}
    \frac{T}{J-1}-O(1)\leq d_{\mathrm{stab}}(T)\leq\frac{2T}{J}+O(1),
\end{equation}
where $O(1)$ is meant with respect to increasing $T$. To demonstrate that the upper- and lower-bounds need not coincide we sketch the contracting circuits for two $J=4$ toric code morphing circuits in~\cref{fig:time-like_errors}. The first of these circuits has $d_{\mathrm{stab}}(T)=\lfloor T/3\rfloor$ and the second has $d_{\mathrm{stab}}(T)=\lfloor T/2\rfloor -1$. The difference between these two circuits is extremely subtle: it only involves permuting the order in which the contracting subsets are contracted, without changing the contraction tree diagram or the stabilisers in each subset. Thus when designing morphing circuits with $J>2$, one has to be careful to ensure that the circuit-level stability distance is as close as possible to the upper bound.

\subsubsection{A Tighter Bound for CSS Morphing Circuits}
\label{subsec:tightCSS}

We can improve the lower-bound in~Prop.~\ref{prop:time_like_lower_bound} for CSS morphing circuits by separately considering the $X$-type and the $Z$-type circuit-level distance of a stability experiment. Here, we define the $Z$-type circuit-level distance as the minimum number of circuit-level $Z$-errors that cause an undetectable flip to a meta-check observable in the stability experiment corresponding to a redundancy in the $X$-stabilisers of the code. Now, suppose we have a depth-$T$ stability experiment but where we omit all the $Z$-measurements, such that there are only $T_{X}$ measurement rounds remaining. Skipping over these $Z$-basis measurements and resets in the stability experiment would not change the $Z$-type circuit-level distance because $Z$-type errors on $Z$-basis measurements and resets are trivial. We can therefore lower-bound the $Z$-type circuit-level stability distance by
\begin{equation}\label{eq:CSS_time-like_distance_bounds}
    d_{\mathrm{stab},Z}(T)\geq\bigg\lfloor\frac{T_{X}}{\delta_{X}-1}\bigg\rfloor,
\end{equation}
where now $\delta_{X}$ is the maximum $X$-duration of a regular $X$-type detector -- that is, the number of measurement rounds that the $X$-type detecting region spans \textit{excluding those measurement rounds that do not measure any $X$-stabilisers}.

\cref{eq:CSS_time-like_distance_bounds} can often lead to a tighter lower-bound than~Prop.~\ref{prop:time_like_lower_bound}. In particular, whenever a morphing circuit has only two contraction rounds in which $X$-stabilisers are contracting, the upper- and lower-bounds match up to a constant. Identical results of course also hold for the $X$-type circuit-level distance of the stability experiment.

\subsubsection{Consequences for Combating Qubit and Coupler Dropouts}

There is one application of morphing circuits that has been widely discussed in the literature that frequently utilises circuits with $J>2$: combating qubit and coupler dropouts~\cite{debroy2024:luci,higgott2025handling,wolanski2026automated,anker2025:optimized}. These typically have mid-cycle codes that are subsystem codes instead of stabiliser codes (due to qubit dropouts), and have morphing circuits in which some gauge operators are measured more than once in the morphing circuit. Nevertheless, we can still apply the lower-bound in~Prop.~\ref{prop:time_like_lower_bound} and, if the morphing circuit is CSS, the tighter lower-bound in~\cref{eq:CSS_time-like_distance_bounds}.

As an example, the original formulation of LUCI~\cite{debroy2024:luci} uses four contraction rounds to implement a subsystem code for the surface code with dropouts. All of the $X$-gauge-checks are contained in the first two contracting subsets, $F_{1}$ and $F_{2}$, while all of the $Z$-gauge-checks are contained in the last two, $F_{3}$ and $F_{4}$. If stabilisers are formed by products of gauge operators in different contracting subsets, then the maximum duration of a detector is $\delta=5$ and the circuit-level stability distance is lower-bounded by $d_{\mathrm{stab}}(T)\geq T/4-O(1)$. However, this suggests a possible strategy to improve $d_{\mathrm{stab}}$ by ensuring that all stabilisers are formed by products of gauge operators in the \textit{same} contracting subsets: in this case, we can use \cref{subsec:tightCSS} to improve the 
lower-bound to $d_{\mathrm{stab}}(T)\geq T/2-O(1)$. This lower-bound implies that lattice surgery would require \textit{at most} $2d+O(1)$ measurement rounds to achieve $d_{\mathrm{stab}}(T)=d$ instead of $4d+O(1)$.\footnote{Note that in an earlier version of this manuscript we made an incorrect claim about the circuit-level stability distance of LUCI, we apologise for any confusion caused.}

\subsection{Single-Shot Error Correction and the 3D Toric Code}\label{subsec:3D_TC}

To conclude this section, we examine the interplay between morphing circuits and codes with \textit{single-shot} error-correction \cite{Campbell_2019}, illustrated with the example of the \eczoo[3D toric code]{3d_surface}. Single-shot codes contain an extensively-large number of redundancies between the measured stabiliser or gauge generators that are designed to be used by a decoder to correct for measurement errors in only a constant number of measurement rounds, not scaling with the size of the code.

\begin{figure}
    \centering
    \includegraphics[width=\linewidth]{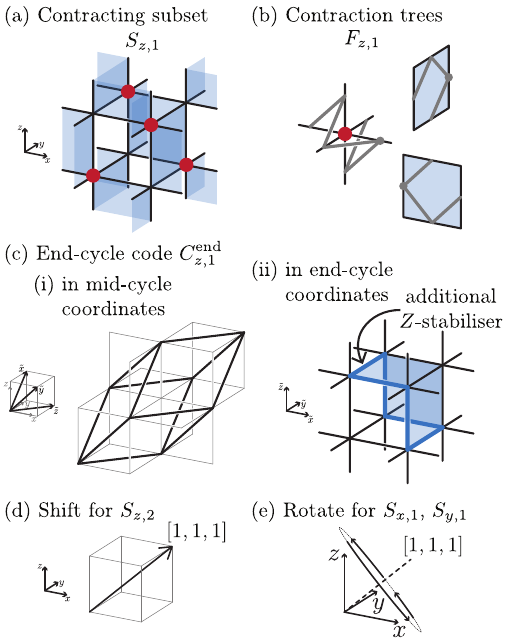}
    \caption{Morphing the 3D toric code, with $X$-checks on vertices, qubits on edges, $Z$-checks on faces, and $Z$-redundancies around each cube. (a) The contraction subset $S_{z,1}$, with contracting $X$-stabilisers highlighted in red, and contracting $Z$-stabilisers shaded in blue. Half of all $X$-stabilisers and half of the $Z$-stabilisers along $x$-oriented and $y$-oriented faces are contracted; none of the $z$-oriented $Z$-stabilisers are contracted. (b) The contraction trees that make up the depth-3 contraction circuit $F_{z,1}$, corresponding to the contracting stabilisers in $S_{z,1}$. (c) The end-cycle code $C^{\mathrm{end}}_{z,1}$ for this contraction circuit is again a 3D toric code on a cubic lattice but the lattice is transformed as compared to the mid-cycle lattice. In (i) we show the relationship between the mid-cycle cubic lattice (grey) and the end-cycle cubic lattice (black). Although the codespace of $C^{\mathrm{end}}_{z,1}$ is that of the 3D toric code, the end-cycle code does have one key difference: each cube now contains \textit{two} redundancies due to the presence of an additional weight-six $Z$-stabiliser that traverses each cube. This additional stabiliser is indicated in (ii) with six bold-blue edges representing the qubits it has support on. One of the two redundancies in this cube are given by the three shaded faces in (ii) along with the additional $Z$-stabiliser; the second redundancy contains the additional $Z$-stabiliser along with the other three unshaded faces of the cube. The full morphing circuit consists of six contraction circuits executed in the order $S_{x,1},S_{y,2},S_{z,1},S_{x,2},S_{y,1},S_{z,2}$, each of which can either be obtained by shifting the contraction subset/contracting circuit by the vector $[1,1,1]$ as shown in (d) and/or by performing a $2\pi/3$ rotation around the vector pointing in the $[1,1,1]$ direction that permutes the $x$, $y$ and $z$ coordinates as shown in (e). Both of these transformations leave the end-cycle code invariant (up to a shift in the qubits), so that all the $J=6$ end-cycle codes are related by a shift of qubits. With this, every $X$-check is measured every two contraction rounds, and every $Z$-check is measured two times every six contraction rounds.}
    \label{fig:3D_TC}
\end{figure}

 To define the 3D toric code, consider any 3D cubic lattice -- the lattice could have some periodic boundary conditions or it could formally exist on the infinite 3D plane $\mathbb{R}^{3}$ -- and place a qubit on every \textit{edge} of this lattice. 
 %Edges can be oriented along one of the $x$, $y$, or $z$ coordinate in the lattice. 
 $X$-stabilisers are located at every vertex of the lattice, with the stabiliser having support on the six qubits on the edges that touch the vertex. Likewise, $Z$-stabilisers are located at every face of the lattice, with each face oriented perpendicular to either the $x$, $y$ or $z$ axis. The six $Z$-stabilisers corresponding to the faces around any cube in the lattice form a redundancy, since their product is the identity, as was also shown in~\cref{fig:detector_figures}(b). This means that measurement errors on the $Z$-stabilisers can be detected by checking the product of $Z$-stabilisers around a cube in a single measurement round. In other words, the 3D toric code is single-shot with respect to measurement errors on the $Z$-stabilisers, although it is not single-shot with respect to measurement errors on the $X$-stabilisers.

The first challenge that arises in designing a morphing circuit for the 3D toric code is that there are significantly more stabiliser generators than there are qubits in the code (which is necessary to have the redundancies that allow for single-shot error correction). In codes without single-shot properties such as the 2D toric code, the number of stabiliser generators is approximately equal to the number of qubits. This makes it relatively natural to design a two-round morphing circuit by dividing the stabiliser generators into two subsets and by using half of the qubits of the mid-cycle code as ancilla qubits in each contraction round. Since the 3D toric code has more stabiliser generators than qubits, the usual approach of having half of the qubits as ancillas in each contraction round implies that we cannot contract half of the stabilisers in each round, necessitating a morphing circuit with $J>2$. On the other hand, if one wishes to design a morphing circuit with $J=2$, then this requires using more than half of the mid-cycle data qubits as ancillas, reducing the number of end-cycle data qubits and -- one would expect -- the end-cycle distance.

In~\cref{fig:3D_TC}(a--b) we show a natural example of a contraction circuit for the 3D toric code that follows the usual approach in which half of the qubits in the mid-cycle code are ancilla qubits in the end-cycle code. The contracting subset that we label $S_{z,1}$ contains exactly half of the $X$-stabilisers but only a third of the $Z$-stabilisers. In particular, none of the $Z$-stabilisers that are on faces oriented in the $z$-direction are contracted, and half of the $x$- and $y$-oriented $Z$-stabilisers are contracted. We can contract the remaining $X$-stabilisers and $x$- and $y$-oriented $Z$-stabilisers using a shifted contraction circuit $S_{z,2}$ as shown in~\cref{fig:3D_TC}(d), but this leaves all of the $z$-oriented $Z$-stabilisers still to be contracted. To obtain contracting circuits that contract the $z$-oriented $Z$-stabilisers, we perform a rotation of the contraction circuit as shown in~\cref{fig:3D_TC}(e), giving the contracting subsets $S_{x,1}$ and $S_{y,1}$ along with their shifted counterparts $S_{x,2}$ and $S_{y,2}$. Here, the $x$, $y$, and $z$ subscripts refer to the orientation of the $Z$-stabilisers that are \textit{not} contracted. This gives rise to a \textit{six}-round ($J=6$) morphing circuit, but where every stabiliser is contracted more than once in the morphing circuit: each $X$-stabiliser is measured three times (either in $S_{x,1}$, $S_{y,1}$ and $S_{z,1}$; or in $S_{x,2}$, $S_{y,2}$ and $S_{z,2}$) and each $Z$-stabiliser is measured two times (for example, an $x$-oriented $Z$-stabiliser will be measured either in both $S_{y,1}$ and $S_{z,1}$, or in both $S_{y,2}$ and $S_{z,2}$).

One of the consequences of this choice of contraction subsets is that the duration of the regular and meta-check detectors varies for different detectors in the circuit. For the regular $X$-detectors, it is straightforward: every $X$-stabiliser is contracted every second contraction round and therefore the duration of each regular $X$-detector is two. On the other hand, the regular $Z$-detectors $\mathrm{regdet}_{t}(s)$ of a given $Z$-stabiliser $s$ alternate between duration two and four due to the order in which the $Z$-stabilisers are contracted.

Unfortunately, our choice of morphing circuit here actually \textit{increases} the connectivity requirements compared to a bare-ancilla circuit: each qubit on a given edge needs a connection to every adjacent qubit on an edge with a different orientation, meaning the connectivity graph has constant degree eight. This contrasts to the bare-ancilla circuit, in which each data qubit and $X$-ancilla has degree six, while each $Z$-ancilla has degree four. This particular morphing circuit is therefore unlikely to be of practical use in superconducting devices, although it may still be of interest in mobile qubit platforms.

In~\cref{fig:3D_TC}(c) we show the end-cycle code that corresponds to this morphing circuit, which is itself a 3D toric code on a cubic lattice which is different from the mid-cycle cubic lattice. Importantly, the rotation in~\cref{fig:3D_TC}(e) is a symmetry of the end-cycle lattice, which means that all six end-cycle codes are the same up to a shift in qubits. One can therefore use the methodology developed in~\cref{subsec:optimising_boundaries} to optimise the periodic boundaries of the 3D toric code for the end-cycle distance; for example, one could choose the periodic boundary conditions to be given by taking the quotient of $\mathbb{R}^{3}$ by a face-centred cubic lattice to improve the code parameters.

One perhaps counterintuitive property of the 3D toric code morphing circuit is that each redundancy in the mid-cycle code gives rise to one redundancy of the end-cycle code, even though the end-cycle code has only half the number of physical qubits of the mid-cycle code. This necessarily means that the ``density'' of redundancies in the end-cycle codes must be double that of the regular 3D toric code. This manifests itself in the code $C^{\mathrm{end}}$ by the presence of an additional $Z$-stabiliser generator within each cube, as shown in~\cref{fig:3D_TC}(c)(ii). This means now that each cube hosts \textit{two} redundancies rather than one. This behaviour is typical for morphing circuits as we explained in~\cref{subsec:two_round_morphing_redundancies}.

\subsection{Morphing Circuits Preserve Single-Shot Error Correction}\label{subsec:single_shot}

Given the morphing circuit for the 3D toric code in \cref{subsec:3D_TC}, it would be useful to explicitly show that it retains its single-shot error correction properties for the $Z$-stabiliser measurements, in particular show that $d_{\mathrm{stab}}(T)=\Omega(d)$ even when $T=O(1)$ for a relevant stability experiment. For this to make sense, we need a slight generalisation of the definition of a stability experiment that we gave in~\cref{subsec:time-like_distance} to allow for the ``local'' redundancies of the 3D toric code (around each cube) to remain as detectors so they can be used in decoding, while using the ``non-local'' redundancies of the code (sheets of faces of $Z$-stabilisers) as logical observables.

We attempted to demonstrate this numerically, however executing a numerical optimisation to determine $d_{\mathrm{stab}}(T)$ exactly for non-trivial lattice sizes took an inconclusively-long time. Using BP-OSD to upper-bound $d_{\mathrm{stab}}(T)$ following Ref.~\cite{Bravyi24} gave upper-bounds higher than an explicit circuit-level logical error that we were able to construct by hand.

Despite our lack of numerical evidence, we can state a general result which shows that single-shot codes remain single-shot under a morphing circuit. In particular, let us consider a phenomenological noise model as described in \cref{subsubsec:morphing_decoding}. Our aim in the proposition below is to show that the circuit-level stability distance of a morphing circuit $d_{\mathrm{stab}}^{\mathrm{morph}}(T)$ is the same as that of the phenomenological noise model $d_{\mathrm{stab}}^{\mathrm{phen}}(T)$ up to a constant factor. Since single-shotness means that this circuit-level stability distance scales with the code distance $d$ when $T=O(1)$, Prop.~\ref{prop:single_shot_time_like_distance} guarantees that if a code is single-shot under a phenomenological noise model, it is also single-shot in a morphing circuit.
\begin{proposition}\label{prop:single_shot_time_like_distance}
    For a quantum LDPC mid-cycle code and corresponding morphing circuit, if $J$ is constant and each contraction circuit has constant depth, then $d_{\mathrm{stab}}^{\mathrm{morph}}(T)/d_{\mathrm{stab}}^{\mathrm{phen}}(T)=\Omega(1)$, where the constant $\Omega(1)$ is with respect to $T$ and the distance $d$ of the code $C$, and for simplicity $T$ is an integer multiple of $J$. 
\end{proposition}
\begin{proof}[Proof Sketch]
    The result follows relatively straightforwardly from Prop.~\ref{prop:decode}, since this shows that every error in a morphing circuit can be written as a constant number of errors in a phenomenological circuit; we prove the details in Appendix~\ref{sec:proofs}.
\end{proof}

% As an explicit example of how single-shot error-correction manifests itself in morphing circuits, we briefly walk through the example of the 3D toric code to highlight the opportunities and challenges that can arise. A natural approach to designing a morphing circuit that we follow here is to use half of the qubits as ancillas to measure the contracting stabilisers, leaving the other half as data qubits in the end-cycle code. However, because of the presence of redundancies between the $Z$-checks of the 3D toric code, there are significantly more stabilisers in the standard over-complete generating set than there are qubits. This means that following this natural approach necessarily cannot measure half of the stabilisers in one contraction round, necessitating a $J>2$ morphing circuit.

\section{Morphing Circuits under Biased Measurement versus Reset Noise}\label{sec:measurement_reset_bias}

In~\cref{subsec:two_round_morphing_detectors} and~\cref{sec:optimising_time_like} we explained how measurement errors can trigger more than two detectors. In this section we investigate what effect this has on the performance of stability experiments, in the simple case where there is no single-shot error-correction and where the morphing circuit has two rounds. In this setting, we already know that the circuit-level stability distance satisfies $d_{\rm stab}(T)=T$ from~\cref{eq:J=2_time-like_distance}, the same as for normal bare-ancilla circuits.

However, we will show that there is a practical benefit of morphing circuits that can arise when there is a large noise bias between the measurement and reset error rates in an experiment. We show this numerically for the toric and toric colour codes.
In the colour code, this advantage can be achieved without increasing the depth of the syndrome extraction circuit and is therefore a win--win, while in the toric code a deeper circuit is required meaning that it is only advantageous at relatively large noise bias. The idea is that if these circuits provide more protection against errors affecting stabiliser measurement outcomes, then switching to these circuits during lattice surgery operations could be advantageous.

\begin{table}[]
    \centering
    \begin{tabular}{c|c|c}
        Noise location & Noise Type & Error Prob. \\
        \hline
        Idling & 1-qubit Depol. & $p/10$\\
        1-qubit Gates & 1-qubit Depol. & $p/10$\\
        2-qubit Gates & 2-qubit Depol. & $p$\\
        Measurements & Meas. Error & $5p\times\eta$\\
        Resets & Reset Error & $2p$\\
        Idling during M/R & 1-qubit Depol. & $2p$
    \end{tabular}
    \caption{The noise model used in our numerics. Here, the error probability of a depolarising channel refers to the probability that \textit{any} Pauli error occurs, while the error probability of a measurement or reset refers to the probability that the wrong measurement outcome is recorded, or that the wrong basis state is initialised (for example, preparing $\ket{1}$ instead of $\ket{0}$). Note that noise from ``idling'' and from ``idling during measurements/resets'' stack, so that qubits that are idling while at least one other qubit is being measured or reset experience two single-qubit depolarising channels. This is the same as the SI1000 noise model from Ref.~\cite{gidney2021fault}, except for two things. First, for convenience, we use measurements/resets in both the $X$ and $Z$ basis along with CNOT gates, rather than compiling to $Z$-basis measurements/resets and CZ gates. And second, we have added an additional noise parameter $\eta$ that allows us to independently scale the measurement-error rate. Note that setting $\eta=1$ recovers the original SI1000 noise model, while $\eta=0.4$ corresponds to the measurement and reset error rates being equal to each other.}
    \label{tab:noise_model}
\end{table}

We are interested here in an improvement in the logical error rate of a stability experiment under a specific noise model. Namely, the noise model that we use is the SI1000 noise model~\cite{McEwen23}, but with a modified measurement error rate to allow for variable levels of measurement versus reset noise bias, as shown in~\cref{tab:noise_model}. In particular, the measurement error rate in the SI1000 noise model is multiplied by a factor $\eta$, where $\eta=0.4$ represents the point where the measurement and reset error rates are equal.

In the SI1000 noise model, measurement and reset operations are significantly noisier than CNOT gates, so here we look for ways to make the circuit more resilient against measurement/reset errors even if it does not improve its resilience against CNOT gate errors. In particular, we design circuits that have a high measurement-error-only stability distance, which we define as the circuit-level distance of a stability experiment under a noise model that \textit{only} contains measurement errors. Such a circuit will therefore perform well under the full circuit-level noise model in~\cref{tab:noise_model} in the regime of large $\eta$. To achieve this, we want measurement errors to trigger more than two defects each, so that there are no purely-measurement string-like errors that can cause a logical error in a stability experiment -- in this case, we say that the circuit has a ``non-string-like measurement structure''. Note that in bare-ancilla circuits, because each detector only contains two measurements, each measurement error only triggers two detectors and therefore the circuit has a string-like measurement structure.

To see why this might be possible, recall that the circuit-level stability distance of a morphing circuit was upper-bounded in the proof of~Prop.~\ref{prop:time_like_upper_bound} by an error string that consists of {\em alternating} measurement and reset errors on the same qubit. This bound therefore does not apply when there are no reset errors, because measurement and reset errors can cause different syndromes in morphing circuits. Moreover, we explained in~\cref{subsec:two_round_morphing_detectors} that measurement and reset errors trigger different syndromes whenever a measurement/reset error triggers more than two defects.
%This is important because in morphing circuits, measurement and reset errors can cause different syndromes, and therefore there may not be a corresponding minimum-weight stability error that contains only measurement errors. 
In contrast, in normal bare-ancilla syndrome extraction circuits every measurement error necessarily has a corresponding reset error (at an earlier time) that causes the same syndrome. 

\begin{figure*}
    \centering
    \includegraphics[width=\linewidth]{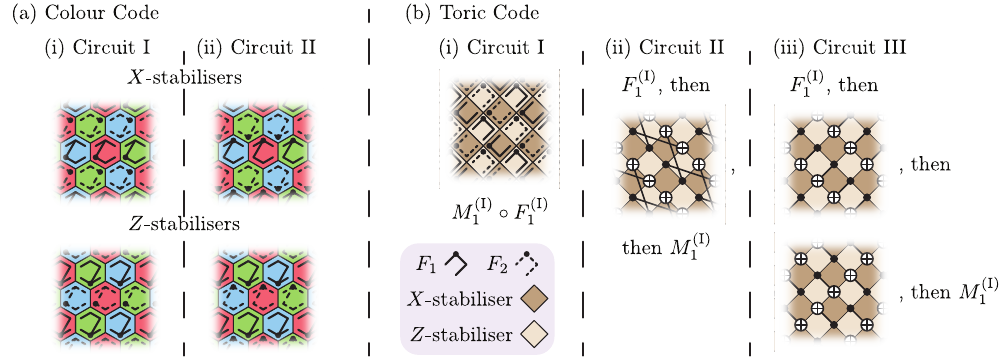}
    \caption{The circuits used to investigate the effect of string-like measurement structure on stability experiments. (a) The two purely contracting homomorphism-based morphing circuits for the colour code represented as contraction trees, where circuit I has a string-like measurement structure, but circuit II does not because each measurement error triggers three different detectors: one in a given time-step and two in the following time-step. Because circuits I and II differ only in the direction of their final CNOT gate, the resulting end-cycle codes are identical (up to a permutation of qubits), in this case they are both the weight-6 end-cycle code from~\cref{fig:colour_code_end_cycle}(b). (b) In the toric code, circuit I is the standard hex-grid toric code, and we explicitly show the contraction tree for $F_{1}^{\mathrm{(I)}}$ and the corresponding measurements $M_{1}^{\mathrm{(I)}}$ for clarity. Circuits II and III are not purely contracting (see~\cref{subsec:purely_contracting}) and therefore cannot be represented by contraction tree diagrams: their respective contraction circuits $F_{1}^{\mathrm{(II)}}$ and $F_{1}^{\mathrm{(III)}}$ \textit{begin} with $F_{1}^{\mathrm{(I)}}$, but are then followed by additional CNOT gates that cause the measurement structure to be non-string-like. In particular, measurement errors now trigger one detector in one time-step and three in the following. This allows the measurement-error-only stability distance to be much larger than the circuit-level stability distance, improving the performance in the regime of large $\eta$.}
    \label{fig:scrambled_circuits}
\end{figure*}

\begin{figure}
    \centering
    \includegraphics[width = \linewidth]{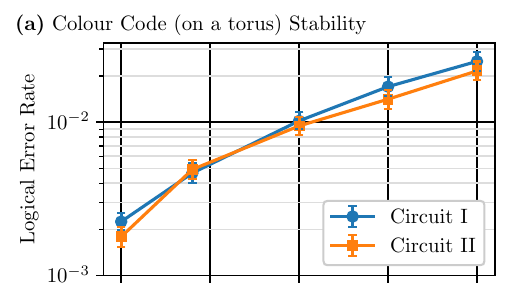}\\[- 0.35 cm]
    \includegraphics[width = \linewidth]{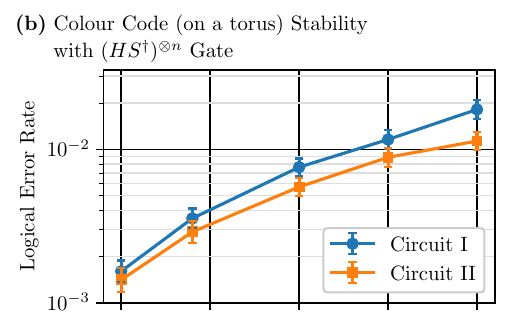}\\[- 0.35 cm]
    \includegraphics[width = \linewidth]{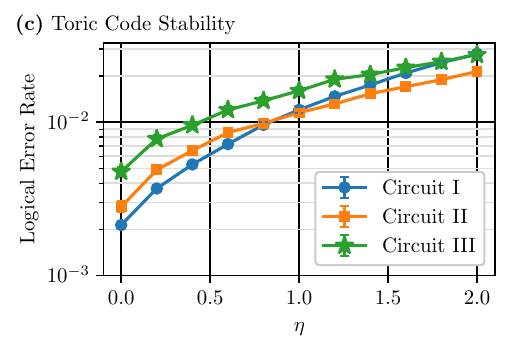}
    \caption{Effect of string-like and non-string-like measurement structure on the logical performance of stability morphing experiments in the colour (a--b) and toric (c) codes. All plots share the same $x$-axis and are evaluated with the noise model in~\cref{tab:noise_model} with a fixed $p=0.002$. Note that $\eta=0.4$ corresponds to equal measurement and reset error rates. The error bars represent 99\% confidence intervals and are sometimes smaller than the corresponding data point markers. In (a--b) we use a $d=4$ colour code on a torus where all four redundancies correspond to observables in a $T=5$ stability experiment, decoded by a minimum-weight decoder using Gurobi~\cite{Gurobi}. The morphing circuits are shown in~\cref{fig:scrambled_circuits}(a). In (b), a transversal $(HS^{\dag})^{\otimes n}$ gate is performed at each mid-cycle time-step to again improve the performance of the stability experiment. In (c), we perform the same experiment for the $d=4$ toric code in a $T=5$ stability experiment which now has two logical observables corresponding to the two redundancies in the toric code, decoded using Belief Matching~\cite{higgott2023improved}. The morphing circuits are given in~\cref{fig:scrambled_circuits}(b).}
    \label{fig:scrambled_plots}
\end{figure}

Our goal is to investigate what, if any, effect having a non-string-like measurement structure has on the performance of stability experiments in two-round morphing circuits. The simplest example we found to investigate this effect is in the weight-6 morphing circuit for the hexagonal lattice colour code on a torus, as shown in~\cref{fig:colour_code_end_cycle}(b). This is because the time-like structure of the circuit can be changed simply by reversing the direction of the final CNOT gate, as shown in~\cref{fig:scrambled_circuits}(a). We use the colour code on a torus instead of the triangular colour code both because it has four redundancies for use as observables in the stability experiment, and to avoid dealing with the complications of adding planar boundaries. In~\cref{fig:scrambled_plots}(a) we numerically compare the logical error rate of a stability experiment for a $d=4$ colour code on a torus using the two circuits in~\cref{fig:scrambled_circuits}(a), using a minimum-weight decoder. Note that for all circuits, the error rate increases with $\eta$ because $\eta$ increases the amount of measurement noise without changing the probability of any other types of noise. We see that for $\eta=0.4$, where the measurement and reset error rates are equal, the logical error rate of both circuits is almost equal. For $\eta\neq 0.4$, the two circuits have logical error rates that fall within the 99\% confidence intervals, but in each case the non-string-like variant slightly outperforms the string-like variant. Note, however, that because this modification fixes the direction of the last CNOT gate in each contraction circuit, this idea cannot be executed at the same time as the swapping of data and ancilla qubits discussed in \cref{subsec:data_ancilla} in the weight-6 colour code morphing circuit\footnote{Although, note that in the weight-7 morphing circuit in \cref{fig:colour_code_end_cycle}(c), both orientations of the final CNOT gate have a non-string-like measurement structure, allowing one to swap data and ancilla qubits while retaining the advantages of the non-string-like measurement structure.}.

We observed a more convincing advantage in~\cref{fig:scrambled_plots}(b), where a transversal $(HS^{\dag})^{\otimes n}$ gate is performed during each mid-cycle code\footnote{Note that the $(HS^{\dag})^{\otimes n}$ gate has the logical effect of permuting the Pauli operators $\overline{X}_{i}\mapsto\overline{X}_{i}\overline{Z}_{i\pm 2}\mapsto\overline{Z}_{i\pm 2}\mapsto\overline{X}_{i}$, for $i=1,2,3,4$.}. The net effect of this modification is that the $X$, $Z$ \textit{and} $Y$ stabilisers are measured in consecutive rounds, rather than alternating between just $X$ and $Z$ stabiliser measurements. 
This improves the performance of both circuits I and II, because detectors can be constructed by comparing, for example, the measured outcome of a $Y$ stabiliser with the parity of the $X$ and $Z$ stabilisers that were measured in the two previous contracting circuits. Therefore, measurement errors appear in more detectors, improving the performance against measurement errors. As expected, both circuits with the transversal gate outperform both circuits without the transversal gate, but here the non-string-like circuit has a lower logical error rate than the string-like circuit falling well outside the 99\% confidence intervals of the experiment.

The situation is more subtle in the toric code because it's difficult to achieve a non-string-like measurement structure with a purely contracting circuit. Instead, we use additional CNOT gates after the standard contraction circuit that commute with the contracting stabilisers, but not with the expanding stabilisers, thereby achieving a non-string-like measurement structure. We show two possible circuits that achieve this in~\cref{fig:scrambled_circuits}(b), where Circuit I is the standard hex-grid toric code with a string-like measurement structure, Circuit II has a non-string-like measurement structure but a depth-3 contraction circuit with non-nearest-neighbour connectivity, and Circuit III has a non-string-like measurement structure \textit{and} hex-grid connectivity, but a depth-4 contraction circuit. The total number of CNOT gate layers to go from one end-cycle code to the next is therefore 4 for Circuit I, 6 for Circuit II, and 8 for Circuit III.

One can understand circuits II and III as containing additional gates that commute with the contracting stabilisers but ``spread out'' the effect of measurement errors so that they effectively look like a measurement error \textit{plus} an end-cycle data qubit error. The increase in depth of the contraction circuits in circuits II and III adds extra noise compared to Circuit I, but in~\cref{fig:scrambled_plots}(c) we find numerically that for sufficiently large $\eta$ the advantage of the non-string-like measurement structure overcomes the disadvantage of having a deeper circuit. Remarkably, the cross-over point for Circuits I and II already occurs at $\eta=1$ which is the standard SI1000 noise model, although Circuit II requires a non-planar connectivity. The crossover between Circuits I and III -- both of which only require hex-grid connectivity -- occurs at roughly $\eta=2$. Both of these circuits may therefore be advantageous to run during lattice surgery operations on devices with large enough measurement--reset bias.

\section{Discussion}\label{sec:discussion}

%In this work, we have performed an extensive exploration of morphing circuits to demonstrate that they provide an attractive approach to performing quantum error correction for a wide range of QEC codes and across various physical platforms. We did this in multiple ways ways:\begin{itemize}   \item by generalising known advantages of morphing circuits beyond the specific examples that exist in the literature,  \item by designing techniques to improve the fault-tolerance of morphing circuits, and    \item by uncovering previously-unexplored advantages of morphing circuits.\end{itemize}

Here are a number of open problems relating to morphing circuits that could make for interesting future work. Throughout this work we only proved our formal results (propositions) for stabiliser codes; however, many of them likely generalise to subsystem and Floquet codes with appropriate generalisations. Moreover, while we provided definitions for the mid- and end-cycle codes of morphing circuits for subsystem and Floquet codes that are sufficient for our analysis of the Bacon-Shor code, we did not explore the consequences of this definition more broadly -- which could lead to morphing circuits with a more complicated structure than is sketched in~\cref{fig:overview}.

In \cref{sec:optimising_time_like}, we assumed that the morphing circuit was executed with unconditional resets, but it may be advantageous in some experimental platforms not to reset after a projective measurement as was considered for example in Ref.~\cite{geher2025reset}. In this situation, so-called ``classical'' measurement errors halve the circuit-level stability distance upper-bound of the morphing circuit in~Prop.~\ref{prop:time_like_upper_bound}, because a single classical measurement error has the same effect as a measurement \textit{and} a reset error combined. Would this trade-off be worth it in morphing circuits?

In \cref{sec:measurement_reset_bias}, we explored how a ``non-string-like measurement structure'' can improve the performance of a morphing circuit in stability experiments with a biased reset and measurement error rate. However, morphing circuits are not the only circuits that can have a non-string-like measurement structure: for example, some measurement errors in superdense circuits also trigger more than two detectors, and could possibly also have a similar advantage under biased noise models.

While we showed in Ref.~\cite{ST:morphing} how one can measure Pauli operations for BB codes using ancillary system using the morphing design principle following Ref.~\cite{Cohen22}, it remains an open question how to adapt the more ancilla-overhead efficient methods (see Ref.~\cite{he2025} and references therein) underpinning the Pauli-based computation in \cite{yoder2025, webster2026} with the morphing design principle. The issue is that in the mid-cycle codes which occur in these logical Pauli measurements, the qubits and stabilisers are no longer labelled by group elements, making a numerical search for a fault-tolerant reduced-connectivity two-round morphing circuit more difficult.

Related to this, a ``holy grail'' of morphing circuit design would be to develop an automated algorithm that searches through morphing circuits for a given mid-cycle code. Ref.~\cite{wolanski2026automated} already presents an algorithm to do this that takes the connectivity graph of a device as input, but ultimately we really would like to search through morphing circuits with different connectivity to determine what device connectivity is necessary. The formalism to do this already exists -- the problem simply amounts to searching through valid contraction tree diagrams -- but the challenge here is that the set of possible contraction tree diagrams is extremely large. One therefore would need to define a search space that is not too large, and yet also yields enough valid contraction tree diagrams over which to search.

Finally, we spent a lot of time optimising colour code boundaries by hand in \cref{subsec:colour_code} with the aim of improving the end-cycle and the circuit-level distance. One therefore wonders if there is an automated way of designing such boundaries for morphing circuits which could be more efficient than our heuristic approaches here. 
Recent examples of works using automation in the design of syndrome extraction circuits are \cite{strikis2026h, liu2026alphasyndrome,viszlai2026}: however the works miss the design principle of morphing circuits of which we have demonstrated its fruitful use in this paper. In the final stages of preparing this manuscript, we had a new idea to improve the circuit-level distance of the triangular weight-6 morphing circuit that we are now actively investigating -- please contact the authors for further information.

%\section{Data and Software Availability} XXX

\section{Acknowledgements}
This work is supported by QuTech NWO funding 2020-2028 – Part I “Fundamental Research”, project number 601.QT.001-1, financed by the Dutch Research Council (NWO). We also thank the OpenSuperQPlus100 project (no. 101113946) of the EU Flagship on Quantum Technology (HORIZON- CL4-2022-QUANTUM-01-SGA) for support.
We acknowledge the use of the DelftBlue supercomputer for running the decoding.
We thank Ted Yoder, Stasiu Wolanski, Mark Turner, Marc Serra Peralta, Matt McEwen, Andrey Khesin, Oscar Higgott, Ioana-Lisandra Draganescu, Ophelia Crawford, Sean Camps, Joan Camps (no relation), Benjamin Brown, and Asmae Benhemou for various insightful and interesting discussions throughout this project. No LLMs happen to be used in the research, the numerics nor in the writing of this paper.

\appendix

\section{Time-Reversal of Normal Circuits}\label{app:time_reversal}

Here we briefly explain why the time-reversed version of a normal syndrome extraction circuit is itself also a normal circuit. The most convenient way to see this is to consider the circuit stabilisers of the normal circuit: these are the Pauli operators placed on the input and/or output qubits that have no effect on the circuit (except for possibly an overall $-1$ phase).

The definition of normal circuits in~\cref{subsec:from_normal_circuits} is actually enough to specify the circuit stabilisers:
\begin{enumerate}
    \item the stabilisers of the code $C$ on the inputs of the normal circuit (these do not affect the circuit except for a $\pm 1$ phase because these operators are measured by the normal circuit),
    \item the stabilisers of the code $C$ on the outputs of the normal circuit (these do not affect the circuit except for a $\pm 1$ phase because the output of a normal circuit is in a $\pm 1$ eigenstate of each of these operators), and
    \item the logical operators of the code $C$ placed \textit{both} on the inputs \textit{and} the outputs of the circuit.
\end{enumerate}
Importantly, note that this set of circuit stabilisers is invariant under swapping the input and outputs of the normal circuit. Therefore, if we run the normal circuit in reverse order, it does not change the circuit stabilisers, and the resulting reversed circuit is also a normal circuit.

\section{Propositions and Their Proofs}\label{sec:proofs}

In this appendix we collect the propositions from the main text that have not yet been proved and provide these proofs that have been deferred. The order in which we prove the propositions here is not the same as in the main text, since some of the propositions are closely related to each other. For convenience, we restate each proposition here and provide a hyperlink to the original location in the main text. 
%\mhs{add a minimal amount of linking text between these propositions.}

\begin{taggedproposition}{\ref{prop:two_round_circuit_distance}}
    For any two-round morphing circuit for a stabiliser code, we have
    \begin{equation}
        d_{\mathrm{circ}}(T)= d_{\mathrm{circ}}(1).\tag{\ref{eq:back_and_forth_fault_tolerance}}
    \end{equation}
\end{taggedproposition}
\begin{figure}
    \centering
    \includegraphics[width=\linewidth]{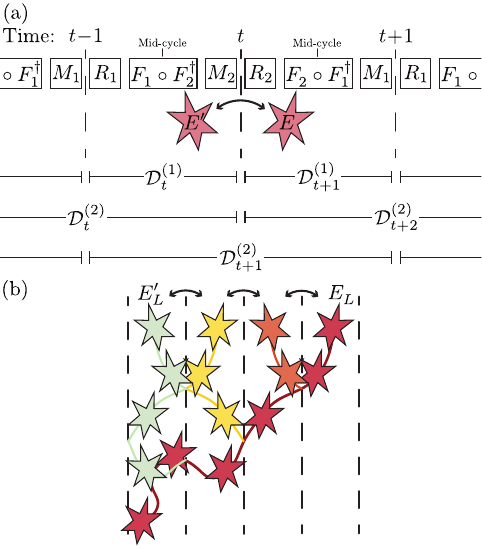}
    \caption{(a) The time-reversal symmetry in a two-round morphing circuit. An error $E$ in a morphing circuit between times $t$ and $t+1$ can be reflected around the axis of symmetry at time $t$ to give an error $E'$ that triggers the same regular detectors in the set $\mathcal{D}_{t+1}^{(2)}$.(b) An arbitrary undetectable circuit-level logical error $E_{L}$ (dark red) spanning some number of measurement rounds (in this example, four) can always be ``folded up'' into an undetectable circuit-level logical error $E_{L}'$ (light green) that only spans one measurement round by repeatedly applying the time-reversal symmetry shown in (a). The resulting one-measurement-round logical error $E_{L}'$ necessarily has weight less than or equal to the original logical error.}
    \label{fig:time_reversal}
\end{figure}
\begin{proof}
    Every undetectable circuit-level logical error that can occur in a one-round memory experiment can also occur in a $T$-round memory experiment, so $d_{\mathrm{circ}}(1)\geq d_{\mathrm{circ}}(T)$. It therefore remains to show that $d_{\mathrm{circ}}(T)\geq d_{\mathrm{circ}}(1)$, which we will do by showing that, given a weight-$w$ undetectable circuit-level logical error $E_L$ in the depth-$T$ memory experiment, we can construct an undetectable circuit-level logical error $E_{L}'$ in the depth-one memory experiment with weight $w'\leq w$.
    
    Before we do that, however, let us establish some preliminary facts. As explained in \cref{subsec:morphing_detectors}, in any morphing circuit we can define a generating set of detectors that take one of the two following forms: either they are \textit{regular} detectors $\mathrm{regdet}_{t}(s)$ for some mid-cycle stabiliser $s$, or they are \textit{meta-check} detectors $\mathrm{metadet}_{t}(r)$ for some mid-cycle redundancy $r$. Let us sort these detectors by their duration and the time in which they occur in the two-round morphing circuit, writing $\mathcal{D}_{t}^{(2)}$ for detectors with a duration of two measurement rounds and containing measurements at times $t$ and $t-1$, and $\mathcal{D}_{t}^{(1)}$ for detectors with a duration of one measurement round and containing measurements at time $t$. $\mathcal{D}_{t}^{(1)}$ contains all the meta-check detectors $\mathrm{metadet}_{t}(r)$ as well as any regular detectors $\mathrm{regdet}_{t}(s)$ for stabilisers $s$ that are contracting in both contraction rounds (this can arise, for example, in alternating normal circuits with flag qubits), while $\mathcal{D}_{t}^{(2)}$ contains the rest of the regular detectors $\mathrm{regdet}_{t}(s)$.
    
    Now, consider a weight-one circuit-level error $E$ that occurs during the morphing circuit between times $t$ and $t+1$ that triggers some set of detectors, as shown in \cref{fig:time_reversal}(a). The detectors that can be triggered by $E$ only come from three subsets: $\mathcal{D}_{t+1}^{(1)}$, $\mathcal{D}_{t+2}^{(2)}$, and $\mathcal{D}_{t+1}^{(2)}$.

    One important property of a two-round morphing circuit is that it has ``time-reversal'' symmetry, as shown in \cref{fig:time_reversal}(a). In particular, we can always find an error $E'$ that occurs at the ``same'' place in the morphing circuit but between times $t-1$ and $t$. For example, a measurement error occurring in $M_1$ just before time $t+1$, is flipped over to a reset error in $R_1$ right after time $t-1$. Because the morphing circuit between times $t-1$ and $t$ is reversed compared to the circuit from $t$ to $t+1$, the set of detectors triggered by $E'$ are a ``time-reversal'' of the detectors triggered by $E$, where most importantly, the regular detectors $\mathcal{D}_{t+1}^{(2)}$ that are triggered by $E$ and by $E'$ are \textit{exactly} the same. In addition to this, the other detectors that are triggered by $E'$ are the same as those triggered by $E$ \textit{except} that they occur at a different time\footnote{Note that this time-reversal symmetry only works for a memory experiment as described here and not for stability experiments, because in the first and last rounds of stability experiments some detectors are missing.}. Finally, the errors $E$ and $E'$ flip the same observables: propagating each error to the nearest mid-cycle code $C$ gives rise to the same (possibly multi-qubit) mid-cycle error and therefore flips the same logical operators.

    Now we can prove that, given any undetectable circuit-level logical error $E_L$ that contains errors in between times $t_{1}$ and $t_{2}+1$, we can always construct another undetectable circuit-level logical error with equal or lower weight that only contains errors between $t_{1}$ and $t_{2}$. We do this by taking each of the errors in $E_L$ that occur between time $t_{2}$ and $t_2+1$ and replacing it with its time-reversed counterpart before time $t_{2}$ as shown in \cref{fig:time_reversal}(b). This time-reversal does not change which detectors in the set $\mathcal{D}_{t_2+1}^{(2)}$ are triggered, and if the error $E_L$ triggered no detector in $\mathcal{D}_{t_2+1}^{(1)}$ and $\mathcal{D}_{t_2+2}^{(2)}$, then the time-reversed error triggers also no detector in $\mathcal{D}_{t_2}^{(1)}$ and $\mathcal{D}_{t_2+1}^{(2)}$, hence the new error is necessarily still undetectable. Furthermore, its logical effect is the same as $E_L$. Moreover, its weight will be the same as $E_L$, unless some of the time-reversed errors cancel out with an existing error in $E_L$, in which case its weight is reduced.

    We can repeat the above argument to finally construct an undetectable logical error $E_L'$ that only contains errors between times $t_{1}$ and $t_{1}+1$ as shown in \cref{fig:time_reversal}(b), and whose weight is less than or equal to that of $E_L$, giving $d_{\mathrm{circ}}(T)\geq d_{\mathrm{circ}}(1)$. 
\end{proof}

\begin{taggedproposition}{\ref{prop:alternating_d_circ}}
    Let $d_{\mathrm{circ}}^{\textnormal{non-alt}}(T)$ be the circuit-level distance of a non-alternating normal circuit, and let $d_{\mathrm{circ}}^{\mathrm{alt}}(T)$ be the circuit-level distance of the corresponding alternating normal circuit. Then, we always have
    \begin{equation}\tag{\ref{eq:back_and_forth_fault_tolerance}}
        d_{\mathrm{circ}}^{\mathrm{alt}}(T)= d_{\mathrm{circ}}^{\mathrm{alt}}(1)=d_{\mathrm{circ}}^{\textnormal{non-alt}}(1)\geq d_{\mathrm{circ}}^{\textnormal{non-alt}}(T).
    \end{equation}
\end{taggedproposition}
\begin{proof}
    See main text.
\end{proof}

Our next proposition, Prop.~\ref{prop:decode}, requires a few definitions before we can make it precise. Recall that we write $G_{\mathrm{circ}}^{C,\mathrm{morph}}$ and $G_{\mathrm{phen}}^{C}$ for the decoding hypergraphs of the morphing circuit under circuit-level noise and the phenomenological noise model, respectively. We write the vertices of these two hypergraphs -- that represent detectors -- as $V_{\mathrm{circ}}^{C,\mathrm{morph}}$ and $V_{\mathrm{phen}}^{C}$, respectively. We wish to establish a correspondence between the detectors (or, equivalently, the vertices in the decoding hypergraph) in a given morphing circuit and the corresponding phenomenological circuit for the mid-cycle code $C$, as defined in \cref{subsubsec:morphing_decoding}. We do this by defining a map $f\colon V_{\mathrm{circ}}^{C,\mathrm{morph}}\rightarrow V_{\mathrm{phen}}^{C}$. As explained in \cref{subsec:morphing_detectors}, a generating set of detectors in a morphing circuit is given by the regular detectors $\mathrm{regdet}_{t}(s)$ for all stabiliser generators $s$ of the mid-cycle code $C$ and all times $t=j+kJ$ ($k\in\mathbb{Z}$) at which $s$ is contracting, and the meta-check detectors $\mathrm{metadet}_{t}(r)$ for all redundancies $r$ between the stabiliser generators and for the times $t=kJ$ ($k\in\mathbb{Z}$). Note that this generating set is overcomplete, for example, it is only \textit{necessary} to include meta-check detectors at a single time-step $t$, because the meta-check detectors at any other time can be obtained by taking their product some regular detectors, but for later convenience we include it at multiple time-steps here. Meanwhile, in a phenomenological noise model, a similarly overcomplete generating set of detectors is made up of regular detectors $\mathrm{regdet}_{t}^{\mathrm{phen}}(s)$ for all $s$ and all times $t$ that consist of the product of two measurements of the same stabiliser $s$ in consecutive time-steps, and meta-check detectors $\mathrm{metadet}_{t}^{\mathrm{phen}}(r)$ that consist of the product of the measurements corresponding to each stabiliser in the redundancy set $r$ at a single time-step $t$. With this, and writing $j(s)$ as the index of the contraction round $S_{j(s)}$ in which a given stabiliser $s$ is contracting, we define the map
\begin{subequations}\label{eq:bijection}
    \begin{align}
        f&:&V_{\mathrm{circ}}^{C,\mathrm{morph}}&\rightarrow V_{\mathrm{phen}}^{C},\\
        &&\mathrm{regdet}_{j(s)+kJ}(s)&\mapsto \mathrm{regdet}_{k}^{\mathrm{phen}}(s),\\
        &&\mathrm{metadet}_{kJ}(r)&\mapsto\mathrm{metadet}_{k}^{\mathrm{phen}}(r),
    \end{align}
\end{subequations}
for $k\in\mathbb{Z}$. Finally, we say that a morphing circuit and a phenomenological circuit are \textit{equivalent under} $f$ if $f$ is a bijection when applied between the morphing and phenomenological circuits.

We are finally ready to state and prove the formal version of Prop.~\ref{prop:decode}.
\begin{taggedproposition}{\ref{prop:decode}}[Formal]
    Consider a quantum LDPC stabiliser code $C$ with a constant-depth $J$-round morphing circuit where every stabiliser generator is measured exactly once in the $J$ rounds of the morphing circuit, and a phenomenological noise model for $C$ that is equivalent under the map $f$ (\cref{eq:bijection}) to the morphing circuit. Then, every hyperedge of the morphing decoding hypergraph $G_{\rm circ}^{C, \rm morph}$ can be written as a product of a constant number of hyperedges in $G_{\rm phen}^C$.
\end{taggedproposition}
\begin{proof}
    Consider any circuit-level error $E_{\mathrm{circ}}$ in a morphing circuit that corresponds to a single hyperedge $h_{\mathrm{circ}}\in H_{\mathrm{circ}}^{C,\mathrm{morph}}$. Our first task is to determine the detectors that are triggered by $E_{\mathrm{circ}}$. To do this, we forward- or backward-propagate the error to the mid-cycle code that is in the middle of the same measurement round as the error. We label the time-step of this mid-cycle code with a half-integer $t_{\mathrm{err}}=j_{\mathrm{err}}+1/2+k_{\mathrm{err}}J$; the resulting (possibly multi-qubit) mid-cycle error $E_{\mathrm{mid}}$ will by definition trigger the same detectors as $E_{\mathrm{circ}}$. Because each contraction circuit is constant-depth, $E_{\mathrm{mid}}$ is a constant-weight Pauli operator. Let us denote by $\mathrm{syn}(E_{\mathrm{mid}})$ the syndrome of $E_{\mathrm{mid}}$; that is, the set of mid-cycle stabilisers that are flipped by $E_{\mathrm{mid}}$. Because $C$ is LDPC, we know that the size of the set $\mathrm{syn}(E_{\mathrm{mid}})$ is also constant. Each stabiliser $s\in\mathrm{syn}(E_{\mathrm{mid}})$ gives rise to one regular detector that is triggered by $E_{\mathrm{circ}}$, but importantly the time at which the detector is triggered is \textit{after} $t_{\mathrm{err}}$. On top of this, a meta-check detector $\mathrm{metadet}_{(k_{\mathrm{err}}+1)J}(r)$ associated with the redundancy $r$ will be triggered if the mid-cycle time-slice of the detecting region at $t_{\mathrm{err}}$ (as illustrated in \cref{fig:detector_figures}(b)) contains an odd number of stabilisers in the set $\mathrm{syn}(E_{\mathrm{mid}}$). Explicitly, $E_{\mathrm{circ}}$ triggers the detectors
    \begin{subequations}\label{eq:morphing_triggered_detectors}
    \begin{align}
        &\mathrm{regdet}_{j(s)+k_{\mathrm{err}}J}(s)
        \text{ for all }s\in\mathrm{syn}(E_{\mathrm{mid}})_{\mathrm{late}},\\
        &\mathrm{regdet}_{j(s)+(k_{\mathrm{err}}+1)J}(s)
        \text{ for all }s\in\mathrm{syn}(E_{\mathrm{mid}})_{\mathrm{early}},\\
        &\mathrm{metadet}_{(k_{\mathrm{err}}+1)J}(r)\text{ for all }r\text{ such that }\nonumber\\
        &\hspace{3 cm}\big|r\cap\mathrm{syn}(E_{\mathrm{mid}})_{\mathrm{early}}\big|\text{ is odd},
    \end{align}
    \end{subequations}
    where $\mathrm{syn}(E_{\mathrm{mid}})_{\mathrm{late}}$ is the set of flipped stabilisers $s$ that are contracting in a contraction round $j>j_{\mathrm{err}}$, and $\mathrm{syn}(E_{\mathrm{mid}})_{\mathrm{early}}$ contains all the other flipped stabilisers.
    
    What are the equivalent detectors that are triggered in the phenomenological noise model? Under the bijection $f$, \cref{eq:morphing_triggered_detectors} corresponds to
    \begin{subequations}\label{eq:phenom_triggered_detectors}
    \begin{align}
        &\mathrm{regdet}_{k_{\mathrm{err}}}^{\mathrm{phen}}(s)
        \text{ for all }s\in\mathrm{syn}(E_{\mathrm{mid}})_{\mathrm{late}},\label{eq:phenom_triggered_detectors_reg_late}\\
        &\mathrm{regdet}_{k_{\mathrm{err}}+1}^{\mathrm{phen}}(s)
        \text{ for all }s\in\mathrm{syn}(E_{\mathrm{mid}})_{\mathrm{early}},\label{eq:phenom_triggered_detectors_reg_early}\\
        &\mathrm{metadet}_{(k_{\mathrm{err}}+1)}^{\mathrm{phen}}(r)\text{ for all }r\text{ such that }\nonumber\\
        &\hspace{3 cm}\big|r\cap\mathrm{syn}(E_{\mathrm{mid}})_{\mathrm{early}}\big|\text{ is odd},\label{eq:phenom_triggered_detectors_meta}
    \end{align}
    \end{subequations}

    Therefore, we have to find a constant number of phenomenological errors that trigger the detectors in \cref{eq:phenom_triggered_detectors}. We claim that the following set of phenomenological hyperedges satisfies this:
    \begin{itemize}
        \item data qubit errors before the measurements at time $t=k_{\mathrm{err}}$ for each error in $E_{\mathrm{mid}}$, and
        \item measurement errors at time $t=k_{\mathrm{err}}+1$ for each stabiliser $s\in\mathrm{syn}(E_{\mathrm{mid}})_{\mathrm{early}}$.
    \end{itemize}
    Note that this set of errors has constant size, because $E_{\mathrm{mid}}$ and $\mathrm{syn}(E_{\mathrm{mid}})$ are both also constant in size.
    To see why this works, note that the data qubit errors trigger regular detectors $\mathrm{regdet}_{k_{\mathrm{err}}}^{\mathrm{phen}}(s)$ for all $s\in\mathrm{syn}(E_{\mathrm{mid}})$, however, this doesn't match the triggered regular detectors in \cref{eq:morphing_triggered_detectors} because the detectors corresponding to stabilisers in $\mathrm{syn}(E_{\mathrm{mid}})_{\mathrm{early}}$ are triggered one time-step later than those in $\mathrm{syn}(E_{\mathrm{mid}})_{\mathrm{late}}$. To compensate for this, the measurement errors trigger pairs of regular detectors $\mathrm{regdet}_{k_{\mathrm{err}}}^{\mathrm{phen}}(s)$ and $\mathrm{regdet}_{k_{\mathrm{err}}+1}^{\mathrm{phen}}(s)$, satisfying \cref{eq:phenom_triggered_detectors_reg_late,eq:phenom_triggered_detectors_reg_early}. It is also not hard to directly see that the measurement errors also trigger exactly the same meta-check detectors as in \cref{eq:phenom_triggered_detectors_meta}, concluding the proof.
\end{proof}

\begin{taggedproposition}{\ref{prop:single_shot_time_like_distance}}
    For a quantum LDPC mid-cycle code and corresponding morphing circuit, if $J$ is constant and each contraction circuit has constant depth, then $d_{\mathrm{stab}}^{\mathrm{morph}}(T)/d_{\mathrm{stab}}^{\mathrm{phen}}(T)=\Omega(1)$, where the constant $\Omega(1)$ is with respect to $T$ and the distance $d$ of the code $C$, and for simplicity $T$ is an integer multiple of $J$.
\end{taggedproposition}
\begin{proof} 
    The result follows relatively straightforwardly from Prop.~\ref{prop:decode}, since this shows that every error in a morphing circuit can be written as a constant number of errors in a phenomenological circuit, under the assumption that each stabiliser was contracted once per $J$ rounds in the morphing circuit. But because we don't make this assumption here, some more arguments are needed. Let the circuit-level logical error $E_{L}^{\mathrm{morph}}$ in the $T$-round morphing stability experiment $\mathcal{C}_{\mathrm{morph}}$ be of minimum weight, i.e.\ $w(E_{L}^{\mathrm{morph}})=d_{\mathrm{stab}}^{\mathrm{morph}}(T)$. Given such an error, we wish to show that there exists a logical error $E_{L}^{\mathrm{phen}}$ in the corresponding $T$-round phenomenological stability experiment $\mathcal{C}_{\mathrm{phen}}$ with weight satisfying $w(E_{L}^{\mathrm{morph}})\geq w(E_{L}^{\mathrm{phen}})/c$ for some constant $c$, which is enough to prove the result since $w(E_{L}^{\mathrm{phen}})\geq d_{\mathrm{stab}}^{\mathrm{phen}}(T)$.

    First, in order to apply Prop.~\ref{prop:decode}, we need to have a morphing circuit and a phenomenological circuit whose detectors are equivalent under the bijection $f$ in~\cref{eq:bijection}. In a morphing circuit where some stabilisers are contracting more than once per $J$ contraction rounds $f$ will not be injective, because multiple regular detectors in the morphing circuit will be mapped to the same phenomenological regular detector. To remedy this, we construct a morphing circuit $\tilde{\mathcal{C}}_{\mathrm{morph}}$ that is identical to the original morphing circuit $\mathcal{C}_{\mathrm{morph}}$ except that some measurements and resets are removed to make it so that every mid-cycle stabiliser is only measured once in the $J$ contraction rounds. $\tilde{\mathcal{C}}_{\mathrm{morph}}$ has the same set of error locations as $\mathcal{C}_{\mathrm{morph}}$, with the only difference being that $\tilde{C}_{\mathrm{morph}}$ now has fewer detectors so that $f$ is injective. The meta-check detectors and/or observables are set by the mid-cycle code and therefore are still present in $\tilde{C}_{\mathrm{morph}}$.
    
    To make $f$ surjective (and hence bijective) we need to define a phenomenological depth-$T/J$ stability experiment $\tilde{\mathcal{C}}_{\mathrm{phen}}$. With this, $\tilde{\mathcal{C}}_{\mathrm{morph}}$ and $\tilde{\mathcal{C}}_{\mathrm{phen}}$ are equivalent under $f$.
    
    %Specifically, the output phenomenological circuit $\mathcal{\tilde{C}}_{\mathrm{phen}}$ will only be a depth-$T/J$ stability experiment -- which is of lower depth than the phenomenological circuit in the proposition. However, even with this, we may need to remove some detectors from the original morphing circuit $\mathcal{C}_{\mathrm{morph}}$ in order for it to be equivalent to $\mathcal{\tilde{C}}_{\mathrm{phen}}$: this is particularly the case if $\mathcal{C}_{\mathrm{morph}}$ has some stabilisers that contract more than once per $J$ rounds of the morphing circuit. Specifically, we construct a morphing circuit $\mathcal{\tilde{C}}_{\mathrm{morph}}$ that is identical to the original morphing circuit $\mathcal{C}_{\mathrm{morph}}$ except that some measurements and resets are removed. $\mathcal{\tilde{C}}_{\mathrm{morph}}$ has the same set of error locations as $\mathcal{C}_{\mathrm{morph}}$, with the only difference being that $\mathcal{\tilde{C}}_{\mathrm{morph}}$ now has fewer detectors such that it is equivalent to $\mathcal{\tilde{C}}_{\mathrm{phen}}$ under $f$.
    
    Starting with $E_{L}^{\mathrm{morph}}$, we can trivially construct an identical circuit-level logical error $\tilde{E}_{L}^{\mathrm{morph}}$ in $\mathcal{\tilde{C}}_{\mathrm{morph}}$. Then, we can apply Prop.~\ref{prop:decode} to each error mechanism in $\tilde{E}_{L}^{\mathrm{morph}}$ to obtain a circuit-level logical error $\tilde{E}_{L}^{\mathrm{phen}}$ in the phenomenological circuit $\mathcal{\tilde{C}}_{\mathrm{phen}}$ that has weight no more than a constant factor larger than $\tilde{E}_{L}^{\mathrm{morph}}$. Finally, we need to turn $\tilde{E}_{L}^{\mathrm{phen}}$ into a logical error for the depth-$T$ phenomenological stability experiment $\mathcal{C}_{\mathrm{phen}}$, which we can do by copying the measurement errors in the last measurement round of the shortened circuit $\mathcal{\tilde{C}}_{\mathrm{phen}}$ and pasting them in each subsequent measurement round of the original depth-$T$ stability experiment $\mathcal{C}_{\mathrm{phen}}$. It is not too hard to see that the resulting set of errors commutes with all the regular and meta-check detectors and anticommutes with the same meta-check observables and thus constitutes a logical error. Moreover, because $J$ is a constant, the weight of this error $E_{L}^{\mathrm{phen}}$ is at most a constant factor larger than the weight of $\tilde{E}_{L}^{\mathrm{phen}}$, proving the result.
\end{proof}

\begin{taggedproposition}{\ref{prop:time_like_lower_bound}}
    Consider a depth-$T$ stability experiment for a morphing circuit in which the maximum duration of its regular detectors, i.e.~the temporal length of its corresponding detection regions, is $\delta$ measurement rounds. Then, the circuit-level distance of this experiment satisfies the lower-bound
        \begin{equation}\tag{\ref{eq:time_like_lower_bound}}
            d_{\mathrm{stab}}(T)\geq \bigg\lfloor\frac{T}{\delta-1}\bigg\rfloor.
        \end{equation}
\end{taggedproposition}
\begin{proof}
    Any undetectable logical error $E_L^{\mathrm{morph}}$ must anticommute with at least one logical observable. Recall from~\cref{eq:meta_check_detector_relation} that for any redundancy $r$, the corresponding meta-check observable has multiple representatives $\mathrm{obs}_{t}(r)$ at different times $t$ that are each related by a regular detector. So, if $E_L^{\mathrm{morph}}$ anticommutes with one meta-check observable representative $\mathrm{obs}_{t}(r)$ for some $t$, it must anticommute with all the representatives in order for it to be undetectable by all regular detectors.
    
    Next, since the maximum duration of a regular detector is $\delta$, every stabiliser is measured at least once in $\delta$ consecutive contraction rounds. Using this, and following the definition of $\mathrm{obs}_{t}(r)$ in~\cref{subsec:morphing_detectors} that says that we include each stabiliser $s\in r$ in the corresponding detection sequence $\vect{\sigma}$ in the most recent time $t'\geq t$ that $s$ is contracting, the earliest time that $s$ could be included in $\vect{\sigma}$ is $t-(\delta-1)$. Therefore the maximum duration of any observable region of any meta-check detector $\mathrm{obs}_{t}(r)$ is $\delta-1$.
    
    This means that the observing regions of the meta-check observable $\mathrm{obs}_{t}(r)$ with $t=k(\delta-1)$ for $k\in \mathbb{Z}$ do not overlap, while still representing the same underlying logical observable. Since $E_L^{\mathrm{morph}}$ must anticommute with all of these representatives, it must have at least one error contained within each of the non-overlapping observing regions corresponding to $\mathrm{obs}_{t}(r)$ with $t=k(\delta-1)$. Therefore, the weight of this error is at least $\lfloor T/(\delta-1)\rfloor$, giving the lower-bound in~\cref{eq:time_like_lower_bound}.
\end{proof}

\begin{taggedproposition}{\ref{prop:time_like_upper_bound}}
    Consider a depth-$T$ stability experiment for a morphing circuit with mid-cycle code $C$. Then, let $N(s)$ be the number of times that a mid-cycle stabiliser generator $s$ of $C$ is contracted in an end-cycle time-step $t=0,\dots,T$, with the first ($t=0$) and last ($t=T$) contraction rounds only counting as half in $N(s)$. Then, the circuit-level distance of this experiment satisfies the upper-bound
    \begin{equation}\tag{\ref{eq:time_like_upper_bound}}
        d_{\mathrm{stab}}(T)\leq 2\min_{s\in R}N(s),
    \end{equation}
    where $R$ is the set of stabilisers that are contained within at least one redundancy of the mid-cycle code.
\end{taggedproposition}
\begin{proof}
    To prove this upper-bound, we construct an undetectable time-like error $E(s)$ consisting of a string of alternating reset and measurement errors for each stabiliser generator $s$ that participates in at least one redundancy. In particular, at every time-step $t$ in which $s$ is contracting, we place a measurement error on the single measurement in $m_{t}(s)$, and we place a reset error on the corresponding reset in the following round of resets.
    Note that the weight of $E(s)$ can be calculated by counting the number of times $s$ contracts during the stability experiment: each time-step in which $s$ contracts contributes two errors (one reset and one measurement error) to $E(s)$ \textit{except} the first ($t=0$) and last ($t=T$) rounds which only contribute one error (a reset or a measurement error respectively) and thus the weight of $E(s)$ is equal to $2N(s)$.

    Working towards explaining why $E(s)$ is an undetectable time-like error, note that $E(s)$ commutes with all of the regular detectors $\mathrm{regdet}_{t}(s)$ corresponding to the stabiliser $s$ because the detector overlaps with both the measurement error at time $t$ and the reset error at the start of the detecting region. Moreover, any \textit{other} regular detector $\mathrm{regdet}_{t}(s')$ that overlaps with a measurement error in $E(s)$ must also overlap with the reset error in $E(s)$ in the same contraction round -- see for example how the regular detector in~\cref{fig:detector_figures}(a) overlaps with the same set of measurements and resets in, for example, $t=1$. Therefore $E(s)$ is undetectable.

    On the other hand, by assumption we know that the stabiliser $s$ is contained in at least one redundancy, let us call this redundancy $r$. We will now show that the meta-check observable $\mathrm{obs}_{t}(r)$ anticommutes with the error $E(s)$ for any $t$ in which $s$ is contracting, and therefore $E(s)$ is an undetectable time-like error. As a reminder, by definition $\mathrm{obs}_{t}(r)$ has a sequence $\vect{\sigma}$ that contains $s$ in the final time-step $t$, and the other stabilisers in $r$ are included in $\vect{\sigma}$ at the most-recent time-step $t'\leq t$ in which they are contracting. Clearly, $\mathrm{obs}_{t}(r)$ anti-commutes with the measurement error on $m_{t}(s)$ that is contained in $E(s)$. What we now need to show is that $\mathrm{obs}_{t}(r)$ commutes with the remaining errors in $E(s)$. The only situation in which $\mathrm{obs}_{t}(r)$ could overlap with another error in $E(s)$ is if there is an earlier time-step $t'$ in which $s$ is once again contracting, but $\mathrm{obs}_{t}(r)$ is not yet terminated -- this is an edge-case that can only occur if the stabiliser $s$ is measured more frequently than some other stabiliser(s) in the redundancy $r$. If this happens, then $s\notin\sigma_{t'}$ because we already have $s\in\sigma_{t}$ and there is no need to include $s$ twice. Now, we want to know the XOR between (1) the qubits corresponding to measurements $M_{t'}$ that the observing region $\mathrm{obs}_{t}(r)$ overlaps on, and (2) the qubits corresponding to resets $M_{t'}$ that the observing region $\mathrm{obs}_{t}(r)$ overlaps on. By definition, this corresponds precisely to the contracted stabilisers in $\sigma_{t'}$, which does not include $s$. Therefore, $\mathrm{obs}_{t}(r)$ either overlaps with both the measurement $m_{t'}(s)$ \textit{and} the corresponding reset, or with neither of them. Combining this with the argument at the later time $t$, this means that the overlap of $\mathrm{obs}_{t}(r)$ with $E(s)$ will always be odd, and therefore $E(s)$ flips the observable.
\end{proof}

\begin{taggedproposition}{\ref{prop:existence}}
    For any quantum stabiliser code $C$ that is a member of a LDPC code family, there exists a corresponding morphing circuit (called a ``disjoint'' morphing circuit) such that:
    \begin{enumerate}[label=(\alph*)]
        \item $C$ is the mid-cycle code of the morphing circuit,
        \item there are a constant (i.e.~not increasing with the size of the code in the family) number of contraction rounds $J$,
        \item each contracting circuit $F_{j}$, $j=1, \ldots, J$ has constant depth, and
        \item the distance ${d}_{j}^{\mathrm{end}}$ of each end-cycle code ${C}_{j}^{\mathrm{end}}$ scales as $d_{j}^{\mathrm{end}}=\Omega(d)$, where $d$ is the distance of the mid-cycle code.
    \end{enumerate}
\end{taggedproposition}
\begin{proof}
    Label the set of stabiliser generators of $C$ as $S$. First, we partition $S$ into a number $J$ of subsets $S_{j}\subseteq S$ with the property that any two stabilisers $s,s'\in S_{j}$ do not overlap on any qubits. To do this, we construct the adjacency graph of the code, which has a vertex for each stabiliser generator $s\in S$ and an edge between pairs of stabilisers that overlap on at least one qubit. Since $C$ is LDPC, every stabiliser $s\in S$ is of constant weight, and every qubit participates in a constant number of stabilisers in $S$. Therefore, the adjacency graph must have constant maximum degree $\Delta$, and by Brooks' Theorem~\cite{brooks1941colouring}, there exists a vertex colouring of the graph that uses a constant number of colours (at most $\Delta+1$). This vertex colouring defines the subsets $S_j$ with $J\leq\Delta+1$, which guarantees that $J$ is constant. 

    With this colouring, the contraction circuits $F_{j}$ can be trivially constructed by first using a layer of single-qubit gates to convert each contracting stabiliser $s\in S_{j}$ into a purely $X$-type operator, and then contracting each $s$ independently with a CNOT ``scrambling'' or ``fan-out'' circuit; this is a circuit that uses CNOT gates to contract the stabiliser in as few gate layers as possible. The depth of the circuit $F_{j}$ in general is therefore $\lceil \log_{2}(w_{\text{max}})\rceil$ where $w_{\text{max}}$ is the maximum weight of the stabilisers in $S_{j}$ -- which is also constant for LDPC codes. The constant depth of this circuit also guarantees that the end-cycle distances are reduced by at most a constant factor by the arguments in Appendix~B.2 of Ref.~\cite{ST:morphing}.
\end{proof}

\section{Morphing Abelian 2BGA Codes}\label{app:Abelian_2BGA_details}

In this appendix we summarise the mathematical details of Abelian 2BGA codes and present the full results from our numerical search.

\subsection{Abelian 2BGA Codes}

\begin{table*}[]
    \caption{Some examples of Abelian 2BGA codes. It is known from Ref.~\cite{wang2022distance} that (copies of) the rotated toric codes always give optimal parameters for weight-4 Abelian 2BGA codes. The hexagonal-lattice colour codes on a torus are self-dual and therefore have $B=A^{-1}=\{a^{-1}:a\in A\}$. In Ref.~\cite{Bravyi24}, many Bivariate Bicycle (BB) codes are discussed, some of which do not have the form listed in the table. However, the examples that are of most interest to us are given in the table by setting $(\ell,m)=(6,6)$, $(9,6)$, $(12,6)$, and $(12,12)$, yielding codes with parameters $[\![72,12,6]\!]$, $[\![108,8,10]\!]$, $[\![144,12,12]\!]$, and $[\![288,12,18]\!]$ respectively.}
    \centering
    \begin{tabular}{c|c|c|c|c}
        Code & $G$ & $A$ & $B$ & $[\![n,k,d]\!]$\\\hline
        Unrotated Toric & $\mathbb{Z}_{d}{\times}\mathbb{Z}_{d}$ & $\{1,x\}$ & $\{1,y\}$ & $[\![2d^2,2,d]\!]$\\
        Rotated Toric (even $d$) & $\mathbb{Z}_{d}{\times}\mathbb{Z}_{d/2}$ & $\{1,x\}$ & $\{1,xy\}$ & $[\![d^2,2,d]\!]$\\
        Rotated Toric (odd $d$) & $\mathbb{Z}_{(d^{2}+1)/2}$ & $\{1,x\}$ & $\{1,x^d\}$ & $[\![d^2+1,2,d]\!]$\\
        Hex Colour on a Torus ($d\equiv0$ mod 4) & $\mathbb{Z}_{3d/4}{\times}\mathbb{Z}_{3d/4}$ & $\{1,x,y\}$ & $A^{-1}$ & $[\![9d^2/8,4,d]\!]$\\
        Hex Colour on a Torus ($d\equiv2$ mod 4) & $\mathbb{Z}_{(9d^{2}{+}12)/16}$ & $\{1,x,x^{(3d+2)/4}\}$ & $A^{-1}$ & $[\![(9d^2{+}12)/8,4,d]\!]$\\
        Some BB from Ref.~\cite{Bravyi24} ($\ell,m$ div.\ by 3) & $\mathbb{Z}_{\ell}{\times}\mathbb{Z}_{m}$ & $\{x^{3},y^{7},y^{2}\}$ & $\{y^{3},x,x^{2}\}$ & Various
    \end{tabular}
    \label{tab:Abelian_2BGA_code_examples}
\end{table*}

An Abelian 2BGA code is defined by an Abelian group $G$ and two subsets $A,B\subseteq G$; we will sometimes compactly write this as the tuple $(G,A,B)$ and we refer to this as the ``group presentation'' of the Abelian 2BGA code. The code then has a total of $2|G|$ physical qubits, each of which is labelled with a label $L$ (``left'') or $R$ (``right'') and by a group element $g\in G$; we use the notation $q(L,g)$ and $q(R,g)$. The code has $|G|$ $X$-stabilisers labelled $s(X,g)$ for $g\in G$ and $|G|$ $Z$-stabilisers labelled $s(Z,g)$. Note that these stabiliser generators in general are not independent which allows for the existence of logical qubits. Specifically, the stabilisers are defined by
\begin{subequations}
    \begin{align}
        s(X,g)&=\prod_{i=1}^{|A|}X_{L,a_{i}g}\prod_{j=1}^{|B|}X_{R,b_{j}g},\\
        s(Z,g)&=\prod_{j=1}^{|B|}Z_{L,b_{j}^{-1}g}\prod_{i=1}^{|A|}Z_{R,a_{i}^{-1}g},
    \end{align}
\end{subequations}
where $P_{T,g}$ represents the Pauli operator $P=X,Z$ acting on the qubit $q(T,g)$ for $T=L,R$ and $g\in G$. This construction guarantees that all the stabilisers commute: if an $X$-stabiliser $s(X,g)$ overlaps with a $Z$-stabiliser $s(Z,g')$ on a left-qubit $q(L,a_{i}g)$, then this necessarily means that $g'=a_{i}b_{j}g$ for some $b_{j}\in B$ and therefore the $X$- and $Z$-stabilisers also overlap on the right-qubit $q(R,b_{j}g)$. 

Some examples of Abelian 2BGA codes are listed in~\cref{tab:Abelian_2BGA_code_examples}: all the examples listed in the table are ``bivariate'' where $G$ is the product of two cyclic groups or \eczoo[Generalised bicycle codes]{generalized_bicycle} where $G$ is $G=\mathbb{Z}_m$, although some Abelian 2BGA codes cannot be written as either a bivariate or Generalised bicycle code.

Abelian 2BGA codes have a number of symmetries that preserve the structure of the code (see for example Statement 3 of Ref.~\cite{wang2022distance} and Theorem 6 of Ref.~\cite{Lin23}). Of particular relevance for us are the following:
\begin{itemize}
    \item Translational symmetry: The Abelian 2BGA codes defined by $(G,A,B)$ and $(G,g_{A}A,g_{B}B)$ are identical for any $g_{A},g_{B}\in G$; here $gA$ denotes the set $\{ga:a\in A\}$.
    \item $ZX$-duality: The Abelian 2BGA codes defined by $(G,A,B)$ and $(G,B^{-1},A^{-1})$ are identical; here $A^{-1}$ denotes the set $\{a^{-1}:a\in A\}$.
    \item Automorphism symmetry: The Abelian 2BGA codes defined by $(G,A,B)$ and $(G,\phi(A),\phi(B))$ are identical for any group automorphism $\phi$; here $\phi(A)$ denotes the set $\{\phi(a):a\in A\}$.
\end{itemize}
These equivalences drastically reduce the number of distinct Abelian 2BGA codes that exist, reducing the computational cost of searching over Abelian 2BGA codes. However, exploiting this simplification is not straightforward using the $(G,A,B)$ presentation given here, not in the least because the number of symmetries grows with the size of the group $G$. Therefore in the following subsection we provide an alternative presentation of Abelian 2BGA codes that more conveniently takes these symmetries into account.

\subsection{The Lattice Presentation of Abelian 2BGA Codes}
\label{app:lattice}

\begin{table*}[]
    \caption{The lattice presentation of the Abelian 2BGA codes from \cref{tab:Abelian_2BGA_code_examples}. Here $\mathrm{HNF}(\Lambda)$ is the (lower-triangular) Hermite Normal Form of the generating vectors of the lattice $\Lambda$, where the generating vectors are rows of $\mathrm{HNF}(\Lambda)$. $n_{\mathrm{copies}}$ is one for all of these codes.}
    \centering
    \begin{tabular}{Sc|Sc|Sc|Sc|Sc}
        Code & $m_{A}$ & $m_{B}$ & $\mathrm{HNF}(\Lambda)$ & $[\![n,k,d]\!]$\\\hline
        Unrotated Toric & 1 & 1 & $\begin{bmatrix}d&0\\0&d\end{bmatrix}$ & $[\![2d^2,2,d]\!]$\\
        Rotated Toric (even $d$) & 1 & 1 & $\begin{bmatrix}d&0\\d/2&d/2\end{bmatrix}$& $[\![d^2,2,d]\!]$\\
        Rotated Toric (odd $d$) & 1 & 1 & $\begin{bmatrix}(d^{2}{+}1)/2&0\\(d{-}1)^2/2&1\end{bmatrix}$ & $[\![d^2+1,2,d]\!]$\\
        Hex Colour on a Torus ($d\equiv0$ mod 4) & 2 & 2 & $\begin{bmatrix}3d/4&0&0&0\\0&3d/4&0&0\\1&0&1&0\\0&1&0&1\end{bmatrix}$ & $[\![9d^2/8,4,d]\!]$\\
        Hex Colour on a Torus ($d\equiv2$ mod 4) & 2 & 2 & $\begin{bmatrix}(9d^2{+}12)/8&0&0&0\\(3d{-}2)^{2}/16&1&0&0\\1&0&1&0\\(3d{+}2)/4&0&0&1\end{bmatrix}$ & $[\![(9d^2{+}12)/8,4,d]\!]$\\
        ``Half-Gross'' BB Code from Ref.~\cite{Bravyi24} & 2 & 2 & $\begin{bmatrix}6&0&0&0\\2&2&0&0\\3&0&3&0\\4&1&1&1\end{bmatrix}$ & [\![72,12,6]\!]\\
        ``\textthreequarters-Gross'' BB Code from Ref.~\cite{Bravyi24} & 2 & 2 & $\begin{bmatrix}6&0&0&0\\0&3&0&0\\5&2&3&0\\2&2&1&1\end{bmatrix}$ & [\![108,8,10]\!]\\
        ``Gross'' BB Code from Ref.~\cite{Bravyi24} & 2 & 2 & $\begin{bmatrix}12&0&0&0\\2&2&0&0\\9&0&3&0\\4&1&1&1\end{bmatrix}$ & [\![144,12,12]\!]\\
        ``Two-Gross'' BB Code from Ref.~\cite{Bravyi24}& 2 & 2 & $\begin{bmatrix}12&0&0&0\\4&4&0&0\\11&2&3&0\\4&1&1&1\end{bmatrix}$ & [\![288,12,18]\!]\\
    \end{tabular}
    \label{tab:Abelian_2BGA_lattice_presentation}
\end{table*}

Now we provide an entirely new definition of Abelian 2BGA Codes in terms of lattices. In this perspective, an Abelian 2BGA code is defined by a ``lattice signature'' $(m_{A},m_{B})$ where $m_{A}$ and $m_{B}$ are positive integers, a full-rank lattice $\Lambda\subseteq\mathbb{Z}^{m}$ where $m=m_{A}+m_{B}$, and a number of copies of the code $n_{\mathrm{copies}}$. We call this tuple of information $(m_{A},m_{B},\Lambda,n_{\mathrm{copies}})$ the ``lattice presentation'' of the Abelian 2BGA code; as an example we have rewritten the codes in \cref{tab:Abelian_2BGA_code_examples} in this lattice presentation in \cref{tab:Abelian_2BGA_lattice_presentation}.

The simplest way to define the Abelian 2BGA code is to describe the mapping from the lattice presentation to the group presentation. In the simple case where $n_{\mathrm{copies}}=1$, we have
\begin{subequations}\label{eq:lattice_to_group_presentation}
\begin{align}
    G&=\mathbb{Z}^{m}/\Lambda,\\
    A&=\{[\vect{0}]\}\cup\big\{[\vect{e}_{i}]:i=1,\dots,m_{A}\big\},\\
    B&=\{[\vect{0}]\}\cup\big\{[\vect{e}_{i}]:i=m_{A}+1,\dots,m\big\},
\end{align}
\end{subequations}
where by $[\vect{v}]$ we mean the coset $\{\vect{v}+\vect{\lambda}:\vect{\lambda}\in\Lambda\}$, and we write $\vect{e}_{i}$ for the $i$th canonical basis vector in $\mathbb{Z}^{m}$. Since the lattice $\Lambda$ is full rank, the group $G=\mathbb{Z}^{m}/\Lambda$ is a finite Abelian group and hence a product of cyclic groups. Clearly, the sets $A$ and $B$ are sets of elements in $G$. Note that the vectors $\vect{\lambda}$ spanning the lattice $\Lambda$ are not necessarily orthogonal, see~\cref{tab:Abelian_2BGA_lattice_presentation}. Also, note that a coset $[\vect{e}_i]$ may be equal to $[\vect{e}_j]$ for $i\neq j$ if $\vect{e}_{i}-\vect{e}_{j}\in\Lambda$.

As the name suggests, when $n_{\mathrm{copies}}$ is greater than 1 it simply means that the code consists of multiple copies of the $n_{\mathrm{copies}}=1$ code. In the group presentation, this simply means making the modification
\begin{subequations}\label{eq:multiple_copies}
\begin{align}
    G&\mapsto G\times \mathbb{Z}_{n_{\mathrm{copies}}},\\
    a&\mapsto (a,1),\text{ for all }a\in A,\\
    b&\mapsto (b,1),\text{ for all }b\in B,
\end{align}
\end{subequations}
where we remind the reader that 1 is the identity element in our notation for the cyclic group $\mathbb{Z}_{n}$. From \cref{eq:lattice_to_group_presentation,eq:multiple_copies} the number of qubits in the code $n$ is equal to $2|G|=2\det(\Lambda)$, where here $\mathrm{det}(\Lambda)$ refers to the determinant of a matrix of generating vectors for the lattice $\Lambda$.

It is clear from~\cref{eq:lattice_to_group_presentation,eq:multiple_copies} that the lattice presentation defines an Abelian 2BGA code. What we now wish to show is that every Abelian 2BGA code given by its group presentation can be written in terms of a lattice presentation. To do this, we first fix an ordering of the elements in the sets $A$ and $B$. Then, we use the translational symmetry of the Abelian 2BGA code to write $A$ and $B$ with their first elements equal to the identity element in the group, so that $A=\{1,a_{2},a_{3},\dots\}$ and $B=\{1,b_{2},b_{3},\dots\}$. Next, consider the subgroup $H\subseteq G$ that is generated by the elements in both $A$ and $B$. Because $G$ and $H$ are both Abelian, we can define the Abelian quotient group $C=G/H$ or for every element $g\in G$ we can write $g= (h,c)$ with $h\in H$ and $c\in C$, hence we can set $n_{\mathrm{copies}}=|C|$. Then, the Abelian group $H$ can always be written in terms of its ``generator presentation'' as $H=H_{\mathrm{free}}/H_{\mathrm{relations}}$. Here, $H_{\mathrm{free}}$ refers to the free Abelian group generated by the symbols $a_{2},\dots,a_{|A|}$ and $b_{2},\dots,b_{|B|}$, which is isomorphic to the group $\mathbb{Z}^{|A|+|B|-2}$, while $H_{\mathrm{relations}}$ refers to the relations between each of the elements $a_{i}$ and $b_{i}$ in the group $H$. $H_{\mathrm{relations}}$ is a subgroup of $H_{\mathrm{free}}$ and can therefore be viewed as a lattice $\Lambda\subseteq \mathbb{Z}^{m}$, where $m=|A|+|B|-2$. The lattice presentation of the Abelian 2BGA code is therefore given by $m_{A}=|A|-1$, $m_{B}=|B|-1$, $\Lambda=H_{\mathrm{free}}$, and $n_{\mathrm{copies}}=|C|$.

The lattice presentation is convenient because any lattice $\Lambda$ can be uniquely specified by the \href{https://en.wikipedia.org/wiki/Hermite_normal_form}{Hermite Normal Form} of its generators; and, conversely, one can iterate over lattices by systematically generating matrices that are in the Hermite Normal Form. Moreover, if one begins with two different group presentations of the same Abelian 2BGA code where the two presentations differ only by translational or automorphism symmetries, then the corresponding lattice presentation will be the same. There are still some symmetries of the lattice presentation, namely the $ZX$-duality and \textit{permutations} of the elements in the sets $A$ and $B$ since the lattice presentation requires a fixed ordering of the elements. However, the number of symmetries no longer scales with the size of the group (or, equivalently, the number of physical qubits in the code), instead only scaling with the number of elements in $A$ and $B$ (or, equivalently, the weight of the stabilisers). This makes the lattice presentation more attractive for a numerical enumeration of Abelian 2BGA codes with fixed stabiliser weight than the group presentation.

As for some closing remarks, the lattice presentation makes it clear that an Abelian 2BGA code with stabiliser weight $w$ is \textit{nearest-neighbour} in $w-2$ dimensional space (with appropriate periodic boundary conditions given by $\Lambda$). The qubits of the code labelled by elements in $G$ can thus be embedded in a $w-2$ dimensional space with stabilisers of the code acting locally, while with $n_{\rm copies}> 1$ one puts $n_{\rm copies}$ qubits per location. The lattice presentation makes the connection from Ref.~\cite{wang2022distance} between weight-4 Abelian 2BGA codes and toric codes explicit, since the lattice presentation directly corresponds to the stabilisers of a 2D toric code with periodic boundary conditions given by the lattice $\Lambda$.

Moreover, the lattice presentation provides a way of defining families of Abelian 2BGA codes, namely, by multiplying the lattice $\Lambda$ by a scaling factor. 

\subsection{Homomorphism-Based Morphing Circuits}
\label{subsec:homo-morph}
With the above simplifications we can search over the set of Abelian 2BGA codes of a given stabiliser weight and number of qubits. However, the set of two-round morphing circuits of a given Abelian 2BGA code is intractably large. We therefore limit ourselves to a particular subset of morphing circuits that we call \textit{homomorphism-based} morphing circuits, which aims to generalise the morphing circuit we found in Ref.~\cite{ST:morphing}. Note that the following definition does not guarantee the existence of a (purely contracting CNOT) morphing circuit, however it makes it very simple to search for such morphing circuits. Note also that in our previous work for $J=2$ we did explicitly prove that the existence of a particular homomorphism ensures the existence of a purely contracting homomorphism-based morphing circuit.

\begin{definition}\label{def:homo-morph}
     A morphing circuit with contracting circuits $\{F_{j}\}_{j=1}^J$ for an Abelian 2BGA code $(G,A,B)$ is \emph{homomorphism-based} if:
\begin{enumerate}[label=(\alph*)]
    \item for any unitary translation $U(g)$ of the qubits that maps $q(L,g')\mapsto q(L,g'g)$ and $q(R,g')\mapsto q(R,g'g)$, and for each contracting circuit $F_{j}$, the conjugated circuit $U(g)F_{j}U(g)^{\dag}$ is also a contracting circuit $F_{j'}$ in the family $\{F_{j}\}_{j=1}^J$ (where $j$ may be equal to $j'$),
    \label{item:morphing_translations}
    \item for each contracting circuit $F_{j}$, the translated circuit $H^{\otimes n}\mathrm{SWAP}_{ZX}F_{j}H^{\otimes n}\mathrm{SWAP}_{ZX}$ is also a contracting circuit $F_{j'}$ in in the family $\{F_{j}\}_{j=1}^J$ where $\mathrm{SWAP}_{ZX}$ swaps the qubits $q(L,g)\leftrightarrow q(R,g^{-1})$: this guarantees that the stabilisers are mapped $s(X,g)\leftrightarrow s(Z,g^{-1})$, and
    \label{item:morphing_ZX}
    \item there exists a group homomorphism $\phi:G\rightarrow H$ for an Abelian group $H$ such that the set of contracting $X$-stabilisers in any given round $F_j$ satisfy $\{g:s(X,g)\in S_{j}\}=\phi^{-1}(h_{X,j})$ for some $h_{X,j}\in H$, and the set of contracting $Z$-stabilisers in any given round $F_j$ satisfy $\{g:s(Z,g)\in S_{j}\}=\phi^{-1}(h_{Z,j})$ for some $h_{Z,j}\in H$.\label{item:homomorphism_condition}
\end{enumerate}
\end{definition}

A visual example is given by the two-round purely contracting homomorphism-based morphing circuits for the unrotated toric code shown in~\cref{fig:toric_search}, corresponding at this size to $G=\mathbb{Z}_4 \times \mathbb{Z}_4$. Here $H=\mathbb{Z}_{2}=\{1,u|u^2=1\}$ (as we chose in \cite{ST:morphing}), and $J=2=|H|$. These circuits both satisfy~\cref{item:morphing_translations} in Def.~\ref{def:homo-morph} due to the translational symmetry of the circuits. To see that they also satisfy \cref{item:morphing_ZX} note that the SWAP${}_{ZX}$ gate implements a $\pi$-rotation about a point such that the $X$- and $Z$-stabilisers of the toric code are swapped. The homomorphism that satisfies~\cref{item:homomorphism_condition} in both circuits is given by $\phi:\mathbb{Z}_{4}\times\mathbb{Z}_{4}\rightarrow\mathbb{Z}_{2}$, $\phi(x)=u$, $\phi(y)=u$. The circuits in~\cref{fig:toric_search} are two examples of proper and parallelisable contraction trees (as defined in \cref{subsec:from_contraction_trees}) that can be found after fixing the homomorphism. Note, moreover, that a morphing circuit being homomorphism-based does \textit{not} guarantee that the connectivity requirements are reduced: the circuit in \cref{fig:toric_search}(a) indeed has only degree-3 connectivity, but the circuit in \cref{fig:toric_search}(b) has degree-4 connectivity.

Another example of a two-round purely contracting homomorphism-based morphing circuit is given by the circuits introduced for BB codes in Ref.~\cite{ST:morphing}. As noted in the main text, many Abelian 2BGA codes do not have any two-round purely contracting morphing circuits: this happens for example whenever the group $G$ is odd, meaning that~\cref{item:homomorphism_condition} is impossible to satisfy.

Intuitively, \cref{item:morphing_translations,item:morphing_ZX} in~Def.~\ref{def:homo-morph} ensure that the morphing circuits ``respect'' both the translational symmetry and the $ZX$-duality of the mid-cycle Abelian 2BGA code. In particular, a homomorphism-based morphing circuit can be specified just with a \textit{single} contraction tree for the stabiliser $s(X,1)$: the contraction tree for every other $X$-stabiliser $s(X,g)$ can be inferred from the translational symmetry of the circuit, and the contraction tree for the $Z$-stabilisers can be inferred from the $ZX$-duality. \cref{item:homomorphism_condition} is convenient because it simplifies the possible contracting subsets $S_{j}$, making it simple to check whether a given contraction tree and homomorphism condition is parallelisable and proper (although note that in general, homomorphism-based morphing circuits do not need to be purely contracting).

\subsubsection{Two-round Morphing}

Two-round homomorphism-based morphing circuits of Abelian 2BGA codes have the additional property that the two end-cycle codes are equal, up to a permutation of qubits, and that this end-cycle code can be written as an Abelian 2BGA code. The fact that the two end-cycle codes $C_1^{\rm end}$ and $C_2^{\rm end}$ are equal is guaranteed by the translational symmetry of the morphing circuit, i.e.~\cref{item:morphing_translations} for $J=2$. The translational and $ZX$-duality symmetries of the homomorphism-based circuit also guarantee that any expanding $X$- and $Z$-stabilisers whose group elements belong to the same coset $\phi^{-1}(h)$ will expand in such a way that they still respect the translational and $ZX$-duality symmetries. In a two-round morphing circuit, all the expanding $X$-stabilisers in a given contraction round have group elements belonging to the same coset and therefore the resulting end-cycle code will respect the translational and $ZX$-duality symmetries. \cref{item:morphing_translations,item:morphing_ZX} guarantee that in a given morphing circuit, half of the left and half of the right qubits are measured, leaving the other half as data qubits in the end-cycle code. The labels of these left-over qubits correspond to one of the two cosets $\phi^{-1}(h)$ for $h\in\mathbb{Z}_{2}$, which correspond to a new ``end-cycle'' group $G^{\mathrm{end}}$ that is isomorphic to $\mathrm{ker}(\phi)=\phi^{-1}(1)$.

\subsection{Numerical Search Results for Weight-5 and 6 Abelian 2BGA Codes}\label{subsec:numerical_search_appendix}

\begin{figure*}
    \centering
    \includegraphics[width=\linewidth]{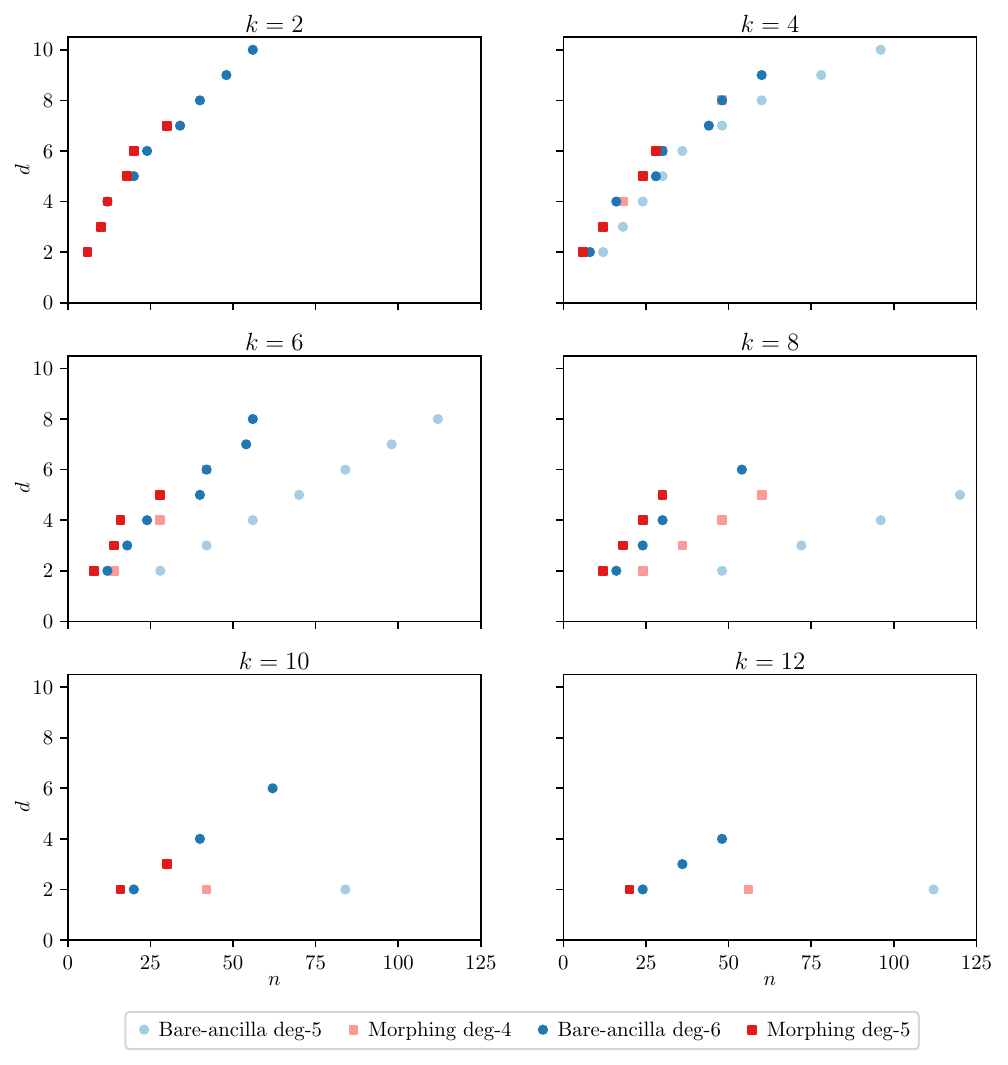}
    \caption{The full set of minimal codes found (up to $k=12$) using the methods outlined in~\cref{subsec:Abelian_2BGA} and Appendix~\ref{app:Abelian_2BGA_details}. Note that a few more minimal codes with larger $k$ can be found in~\cref{fig:d_search_results}.}
    \label{fig:k_search_results}
\end{figure*}

\begin{figure*}
    \centering
    \includegraphics[width=\linewidth]{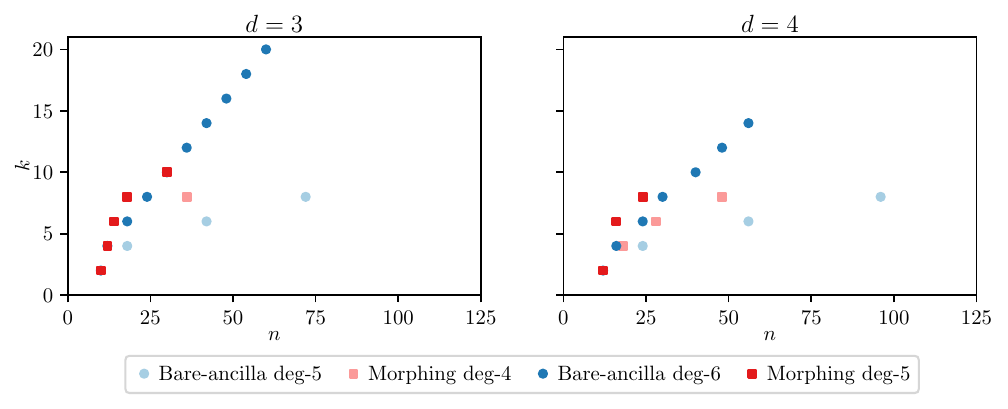}
    \caption{Some minimal codes found by our numerical search using the methods outlined in~\cref{subsec:Abelian_2BGA,app:Abelian_2BGA_details}. Here, we show only the minimal codes with $d=3$ or 4, which includes all the minimal codes with $k>12$ which are not plotted anywhere in~\cref{fig:k_search_results}.}
    \label{fig:d_search_results}
\end{figure*}

Following the flowchart in \cref{fig:Abelian_2BGA_search_flowchart} and using the methods outlined above, we were able to search over two-round homomorphism-based purely contracting morphing circuits for Abelian 2BGA codes of a fixed weight and number of qubits. Here, we show our complete set of numerical results in~\cref{fig:k_search_results,fig:d_search_results}. One can see that the morphing circuits consistently beat, or at least match, the code parameters of bare-ancilla circuits, even though they require a lower degree connectivity graph. Note that in our plots, we only show morphing circuits with a connectivity degree of $w-1$, although the minimal morphing circuits we found with a higher connectivity degree do not have significantly better code parameters. As discussed in the main text, we did not check the circuit-level distance of the bare-ancilla or the morphing circuits. We leave it to future work to investigate particular codes and/or circuits of interest, or to continue the search to larger code sizes.

\subsection{Infinite Abelian 2BGA Codes}
\label{sec:inf2BGA}

We introduce the formal construction of infinite Abelian 2BGA codes. We use this to define morphing circuits for the surface and colour codes on infinite planes, after which we can add boundaries as described in~\cref{subsec:boundaries,subsec:optimising_boundaries}. In general, an infinite Abelian 2BGA code is simply any Abelian 2BGA code where the group $G$ is no longer finite. For example, 
for \eczoo[Generalised Bicycle codes]{generalized_bicycle} $G$ becomes $\mathbb{Z}$.
In terms of the lattice presentation of Abelian 2BGA codes, the group $G=\mathbb{Z}^{m}/\Lambda$ in ~\cref{eq:lattice_to_group_presentation} being infinite is equivalent to the lattice $\Lambda$ not being full rank. We are particularly interested in the infinite surface code
\begin{subequations}
\begin{align}\label{eq:infinite_surface_code}
	G&=\mathbb{Z}\times\mathbb{Z},&A&=\{1,x\},&B&=\{1,y\},
\end{align}
and the infinite hexagonal-lattice colour code
\begin{align}\label{eq:infinite_colour_code}
	G&=\mathbb{Z}\times\mathbb{Z},&A&=\{1,x,y\},&B&=A^{-1},
\end{align}
which have lattice presentations
\begin{align}\label{eq:infinite_surface_code_lattice}
	m_{A}&=1,&m_{B}&=1,&\mathrm{HNF}(\Lambda)&=\big[\,\big],&n_{\mathrm{copies}}&=1,
\end{align}
and
\begin{equation}\label{eq:infinite_colour_code_lattice}
\begin{aligned}
	m_{A}&=2,&\mathrm{HNF}(\Lambda)&=\begin{bmatrix}1&0&1&0\\0&1&0&1\end{bmatrix},\\
    m_{B}&=2,&n_{\mathrm{copies}}&=1,
\end{aligned}
\end{equation}
respectively.
\end{subequations}

Importantly, our definition of homomorphism-based morphing circuits given in~Def.~\ref{def:homo-morph} is still valid for infinite Abelian 2BGA codes: \cref{item:morphing_translations,item:morphing_ZX} are well-defined for infinite groups $G$, and to ensure that the morphing circuit has a finite number of contraction rounds, we add the requirement that $H$ is a finite Abelian group in~\cref{item:homomorphism_condition}.  Similar to the finite case, for any two-round homomorphism-based morphing circuit of an infinite Abelian 2BGA code, the resulting end-cycle codes are equal and can themselves be written as infinite Abelian 2BGA codes. These properties are also particularly useful for the boundary optimisations in Appendix~\ref{sec:code_boundary_optimisation}, because the optimal boundary conditions for both end-cycle codes will be the same.

\section{Optimising Distance versus Number of Qubits for the End-Cycle Codes of Morphing Colour Code Circuits}\label{sec:code_boundary_optimisation}

\begin{figure*}
    \centering
    \includegraphics[width=\linewidth]{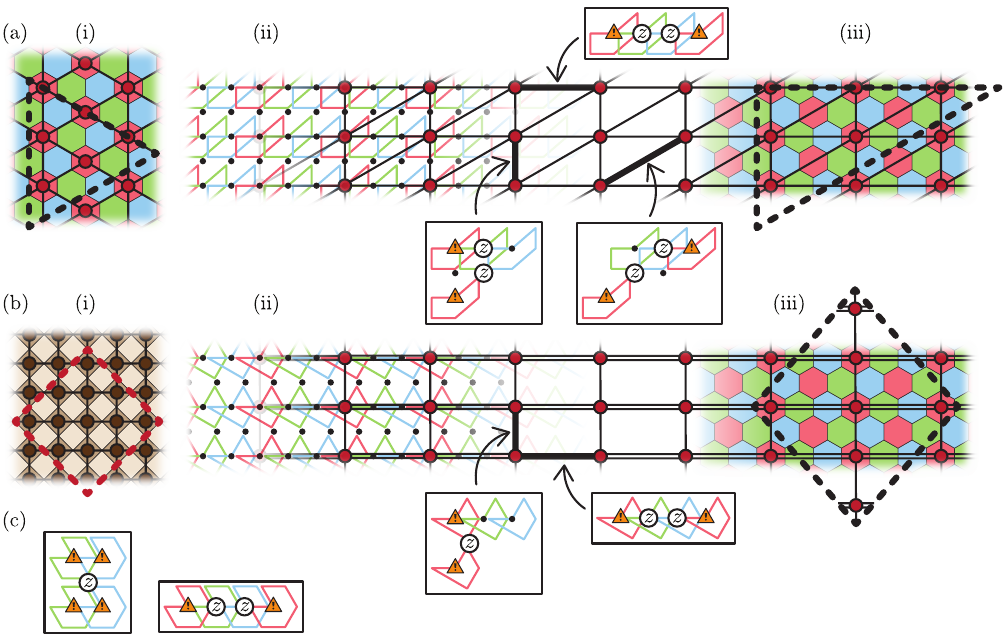}
    \caption{Construction of optimal boundaries for the (a) weight-6, (b) weight-5, and (c) weight-7 end-cycle codes of colour code hex-grid morphing circuits in~\cref{fig:colour_code_end_cycle}. (a) and (b) outline the full procedure using the error strings that trigger red $X$-stabilisers, starting in (i) with the well-known optimal boundaries (thick dotted line) of (a) colour or (b) surface code. In (ii), we analyse the error string structure for the end-cycle code, with insets showing the $Z$-error mechanisms that extend an error string from one red $X$-stabiliser to the next. For the weight-6 end-cycle code, all edges represent weight-2 errors; while for the weight-5 end-cycle code, the single edges represent weight-1 errors and the doubled edges represent weight-2 errors. By comparing the error string for the end-cycle code with that of the colour/surface code, we can impose optimal planar boundary conditions (thick dotted lines) and transfer them to the mid-cycle code in (iii). Note that all the boundaries drawn here correspond to the boundary as shown in \cref{fig:padding_boundaries}, and the microscopic details will depending on the colour/Pauli type of the boundary being imposed -- the boundaries imposed here are analogous to the thick red line in \cref{fig:padding_boundaries}. The morphing circuit for the weight-6 end-cycle code is shown in the main text in \cref{fig:half_shifted_hex} for $d^{\mathrm{end}}=9$, and also shows some minimum-weight end-cycle logical errors. (c), we show the error mechanisms in the weight-7 end-cycle code that give error strings of a given colour, where we consider the weight-1 string that triggers a pair of blue and green stabilisers as ``red'' error string. With this, the error strings travel along the same graph as the weight-5 end-cycle code and therefore the optimal boundary conditions are the same as those in (b)(iii). The optimisation problem for the $X$-errors is the same as for the $Z$-errors.}
    \label{fig:end_cycle_optimal_boundaries}
\end{figure*}

In this appendix, we outline the method we use to optimise the boundary conditions of the infinite end-cycle codes of morphing circuits for the colour code, see the flowchart in \cref{fig:optimisation_flowchart}. In particular, we will analyse the structure of the logical ``strings'' in each end-cycle code, and show that their structure mirrors either that of the surface or the colour code. Then, we can use the fact that the rotated surface code and the equilateral-triangular colour codes are the optimal planar boundary conditions for the surface and colour codes respectively to construct similar optimal boundary conditions for each end-cycle code. Finally, we will demonstrate how to transfer the optimal end-cycle boundary conditions to the mid-cycle code to construct the morphing circuits that we presented in \cref{subsec:colour_code}.

Let us begin by reviewing how to obtain the rotated surface code from the infinite surface code (which can be formally defined as an infinite Abelian 2BGA code in~\cref{eq:infinite_surface_code}). In order to obtain such a single-logical qubit surface code, one needs to impose quadrilateral boundary conditions with a pair of $X$-type boundaries and a pair of $Z$-type boundaries on opposite edges. The optimal boundary geometry is the rotated surface code for the following reasons. Error strings traverse a square grid as shown in \cref{fig:end_cycle_optimal_boundaries}(b)(i); or, in other words, the minimum-weight of an error string between a pair of endpoints is given by the $L_{1}$-distance between the endpoints. In the planar surface code, $X$-logical operators are formed by $X$-error strings crossing from one $X$-boundary to the other. Because of the square-grid structure of the error strings, it is therefore optimal to orient the $X$-boundaries diagonally with respect to this square grid. The same is true for the $Z$-boundaries, giving rise to the well-known ``diamond'' boundary conditions of the rotated surface code as shown in \cref{fig:end_cycle_optimal_boundaries}(b)(i).

%In the infinite surface code, an error string is a set of Pauli errors that only triggers a pair of stabilisers at the endpoints of the string.

In the colour code, the story is similar albeit somewhat more subtle. In this case, there are three ``colours'' of error strings that exist, depending on the colour of the stabilisers that are triggered at the endpoints of the string. These coloured error strings traverse a triangular grid instead of a square grid, as shown for the red error strings in \cref{fig:end_cycle_optimal_boundaries}(a)(i). Moreover, an error string of one colour can ``split'' into two strings in the other two colours (green and blue). To obtain a planar colour code with a single logical qubit, we impose triangular boundaries where each boundary has a different type: one red, one green, and one blue. Logical operators are formed by a red, green, and blue string originating from a boundary of matching colour and meeting at some point in the interior of the colour code. We wish to optimise the geometry for the distance and the number of qubits by using the fact that each error string traverses a triangular grid. Once again, the optimal geometry is well-known and is given by equilateral triangular boundary conditions.

Thus, to optimise the boundary conditions of the end-cycle codes in \cref{fig:colour_code_end_cycle}, we first need to determine what the structure of the error strings is in each end-cycle code, beginning with the weight-6 end-cycle code in \cref{fig:end_cycle_optimal_boundaries}(a)(ii). The key is to find pairs of errors that trigger pairs of stabilisers of the same colour -- such pairs make up the building blocks of any error string. In this case, we see that the structure of the error strings in the weight-6 end-cycle code resembles that of the colour code itself because error strings travel along a triangular grid. The optimal boundary conditions are therefore equilateral triangular boundary conditions \textit{with respect to} the triangular grid of the end-cycle error strings. Interestingly, though, the triangular grid of the end-cycle error strings is not aligned with the triangular grid of the mid-cycle colour code, so when we transfer the boundaries to the mid-cycle code in \cref{fig:end_cycle_optimal_boundaries}(a)(iii), the resulting optimised boundary conditions are those of a \textit{right-angled} triangle.

The story is different for the weight-5 end-cycle code. Here, even though the error strings can be red, green, or blue like in the colour code, the error strings move along a square grid like the surface code, as shown in \cref{fig:end_cycle_optimal_boundaries}(b)(ii). For the purpose of optimising the boundary conditions, it is this latter fact that is important, and we therefore wish to impose diamond boundary conditions like the rotated surface code, as shown in \cref{fig:end_cycle_optimal_boundaries}(b)(iii). We can do this in two ways: either by alternating between $X$- and $Z$-type boundaries around the diamond, or using red and blue boundaries. In both cases, however, we encode two logical qubits instead of one, which is undesirable because we can no longer transversally implement the full Clifford group. Nevertheless, it may be interesting since (asymptotically) it has better code parameters than two copies of the surface code. It is of course possible to still impose triangular boundary conditions on the weight-5 end-cycle code as was originally done in Ref.~\cite{Gidney2023}; however, in our attempts to optimise this we were not able to find codes with competitive $[\![n,k,d]\!]$ parameters.

For the weight-7 end-cycle code, it is convenient to also allow error strings to trigger four detectors, two on each end of the string, as shown in \cref{fig:end_cycle_optimal_boundaries}(c). With this minor modification, one can see that the weight-7 end-cycle code has the same error string structure as the weight-5 end-cycle code, with the strings moving along a square grid. As a result, we once again find that the optimal end-cycle codes are obtained by imposing diamond boundary conditions.

\subsection{Optimal Morphing Circuits for Weight-4 Abelian 2BGA Codes}\label{subsec:w4_optimisation_appendix}

As a brief aside, let us outline how to obtain optimal morphing circuits for weight-4 Abelian 2BGA codes, to complement our numerical search in \cref{subsec:Abelian_2BGA} over weight-5 and 6 Abelian 2BGA codes. We have checked numerically that, for all purely contracting homomorphism-based two-round morphing circuits with weight-4 Abelian 2BGA mid-cycle codes, the corresponding end-cycle codes are also weight-4 Abelian 2BGA codes. This differs from the weight-5 and 6 case in \cref{subsec:Abelian_2BGA}, where the weight of the end-cycle stabilisers often increased compared to that of the mid-cycle code. It is known that weight-4 Abelian 2BGA codes are equivalent to some number of copies of toric codes, possibly with rotated periodic boundary conditions~\cite{wang2022distance}. Moreover, it is also known how to optimise the periodic boundary conditions of the toric code~\cite{Kovalev_2012}, and we have listed these as Abelian 2BGA codes in \cref{tab:Abelian_2BGA_code_examples}. Therefore, given a morphing circuit, we simply need to find the weight-4 Abelian 2BGA mid-cycle code whose end-cycle codes correspond to the optimised rotated toric codes; this can be done by transferring the periodic boundary conditions of the end-cycle code back to the mid-cycle code just like we did for the planar colour code in \cref{fig:end_cycle_optimal_boundaries}. The mid-cycle codes that achieve this optimum with the hex-grid morphing circuit in~\cref{fig:hex_grid_TC} are listed in~\cref{tab:optimal_w4_Abelian_2BGA_morphing}.

\begin{table}
    \caption{Optimal mid-cycle codes for the hex-grid toric code morphing circuits, with the corresponding end-cycle codes, for even distance (first row) and odd distance (second row). Note that the end-cycle codes coincide with the rotated toric codes listed in \cref{tab:Abelian_2BGA_code_examples}. These circuits therefore achieve the optimal end-cycle code parameters for weight-4 Abelian 2BGA codes with two-round, purely contracting, homomorphism-based morphing circuits.}
    \centering
    \begin{tabular}{|c|c|c||c|c|c|}
        \hline
        \multicolumn{3}{|c||}{Mid-cycle code} & \multicolumn{3}{c|}{End-cycle code} \\\hline
        $G$ & $A$ & $B$ & $G$ & $A$ & $B$\\\hline
        \hspace{- 3 pt}$\mathbb{Z}_{d}{\times}\mathbb{Z}_{d}$\hspace{- 3 pt}\vphantom{.}&\hspace{- 3 pt}$\{1,x\}$\hspace{- 3 pt}\vphantom{.}&\hspace{- 3 pt}$\{1,y\}$\hspace{- 3 pt}\vphantom{.}& \hspace{- 3 pt}$\mathbb{Z}_{d}{\times}\mathbb{Z}_{d/2}$\hspace{- 3 pt}\vphantom{.} &\hspace{- 3 pt}$\{1,x\}$\hspace{- 3 pt}\vphantom{.}&\hspace{- 3 pt}$\{1,xy\}$\hspace{- 3 pt}\vphantom{.}\\
        \hspace{- 3 pt}$\mathbb{Z}_{d^{2}+1}$\hspace{- 3 pt}\vphantom{.}&\hspace{- 3 pt}$\{1,x\}$\hspace{- 3 pt}\vphantom{.}&\hspace{- 3 pt}$\{1,x^{d}\}$\hspace{- 3 pt}\vphantom{.}& \hspace{- 3 pt}$\mathbb{Z}_{(d^{2}{+}1)/2}$\hspace{- 3 pt}\vphantom{.} &\hspace{- 3 pt}$\{1,x\}$\hspace{- 3 pt}\vphantom{.} &\hspace{- 3 pt}$\{1,x^{d}\}$\hspace{- 3 pt}\vphantom{.}\\\hline
    \end{tabular}
    \label{tab:optimal_w4_Abelian_2BGA_morphing}
\end{table}

\section{Circuit-Level Distance and Numerical Simulations of Morphing Circuits for Colour Codes}\label{sec:colour_code_numerics}

\begin{table*}[]
    \caption{Comparison between the optimised morphing circuits for the planar hexagonal lattice colour code and existing colour code and surface code circuits. The triangular weight-6 and diamond weight-7 morphing circuits are novel constructions introduced in this section, with ``triangular'' and ``diamond'' referring to the shape of the colour code boundaries and ``weight-6/7'' referring to the weight of the stabilisers in the end-cycle code as shown in~\cref{fig:colour_code_end_cycle}. For each circuit, we write the \textit{total} number of data and ancilla qubits $n_{\mathrm{tot}}$ required to achieve an \textit{end-cycle} code distance of $d^{\mathrm{end}}$, the number of logical qubits $k$, the circuit-level distance $d_{\mathrm{circ}}$, and the CNOT depth of the syndrome extraction circuit. The circuit-level distance in each case is supported by numerical calculations.}
    \centering
    \begin{tabular}{c|c|c|c|c|c|c}
        Circuit & $k$ & $\mathrm{depth}(F_{2}^{\vphantom{\dag}}\circ F_{1}^{\dag})$ & Connectivity & $d^{\mathrm{end}}$ & $d_{\mathrm{circ}}$ & $n_{\mathrm{tot}}$\\
        \hline
        \multirow{4}{*}{Superdense~\cite{Gidney2023}} & \multirow{4}{*}{1} & \multirow{4}{*}{8} & \multirow{4}{*}{Hex-square grid} & 3 & 3 & 13\\
        &&&& 5 & 5 & 37\\
        &&&& 7 & 7 & 73\\
        &&&& 9 & 9 & 121\\
        \hline
        \multirow{4}{*}{Gidney \& Jones~\cite{Gidney2023} Morphing} & \multirow{4}{*}{1} & \multirow{4}{*}{6} & \multirow{4}{*}{Hex grid} & 3 & 3 & 23\\
        &&&& 5 & 5 & 69\\
        &&&& 7 & 7 & 139\\
        &&&& 9 & 9 & 233\\
        \hline
        \multirow{4}{*}{Triangular Weight-6 (Original)} & \multirow{4}{*}{1} & \multirow{4}{*}{6} & \multirow{4}{*}{Hex grid} & 3 & 2 & 19\\
        &&&& 5 & 4 & 49\\
        &&&& 7 & 5 & 91\\
        &&&& 9 & 6 & 145\\
        \hline
        \multirow{4}{*}{Triangular Weight-6 (Extra Layer)} & \multirow{4}{*}{1} & \multirow{4}{*}{8} & \multirow{4}{*}{Hex grid} & 3 & 3 & 19\\
        &&&& 5 & 4 & 49\\
        &&&& 7 & 6 & 91\\
        &&&& 9 & 7 & 145\\
        \hline
        \multirow{4}{*}{Triangular Weight-6 (Extra Qubits)} & \multirow{4}{*}{1} & \multirow{4}{*}{8} & \multirow{4}{*}{Hex grid} & 3 & 3 & 19\\
        &&&& 5 & 5 & 65\\
        &&&& 7 & 7 & 119\\
        &&&& 9 & 9 & 201\\
        \hline
        \multirow{3}{*}{Diamond Weight-7} & \multirow{3}{*}{2} & \multirow{3}{*}{6} & \multirow{4}{*}{Hex grid} & 3 & 3 & 28\\
        &&&& 5 & 5 & 100\\
        &&&& 7 & 7 & 179\\
        \hline
        \multirow{4}{*}{Hex-grid surface code~\cite{McEwen23}} & \multirow{4}{*}{1} & \multirow{4}{*}{4} & \multirow{4}{*}{Hex grid} & 3 & 3 & 17\\
        &&&& 5 & 5 & 49\\
        &&&& 7 & 7 & 97\\
        &&&& 9 & 9 & 161
    \end{tabular}
    \label{tab:n_tot_comparison}
\end{table*}

We explained in \cref{subsec:colour_code,sec:code_boundary_optimisation} how we obtained the morphing circuits for the colour code that we analyse in this work. In this appendix, we analyse the circuit-level distance of the triangular weight-6 and diamond weight-7 morphing colour codes considering one measurement round (see Prop.~\ref{prop:two_round_circuit_distance}), and we present our numerical results. The overview of the properties of each circuit that we analyse are presented in \cref{tab:n_tot_comparison}. All the numerical results in this section are for memory experiments conducted over $d$ measurement rounds under SI1000 depolarising noise (see \cref{tab:noise_model}) and decoded using a minimum-weight decoder with Gurobi~\cite{Gurobi}.

\subsection{Triangular Weight-6 Morphing Colour Code}\label{sec:weight_6_details}

\begin{figure*}
    \centering
    \includegraphics[width=\linewidth]{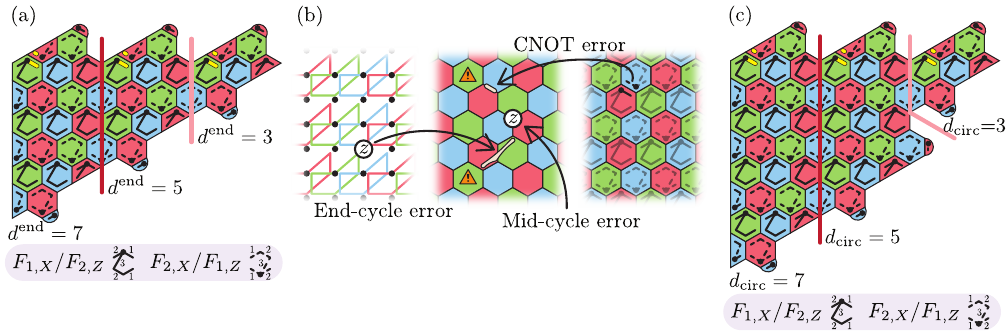}
    \caption{Triangular weight-6 colour code morphing circuits, and their circuit-level errors. (a) The contraction tree diagram for the triangular weight-6 colour code, optimised for the end-cycle code distance using the methods in Appendix~\ref{sec:code_boundary_optimisation}. The full diagram gives a distance-7 end-cycle code, to obtain the distance-5 or 3 counterparts simply ``chop'' the diagram at the vertical lines indicated on the figure. Since the two contraction circuits are related by a transversal Hadamard gate $F_{2}=H^{\otimes n} F_{1} H^{\otimes n\dag}$, we display the trees for both $X$ and $Z$ stabilisers on the same diagram as indicated in the key. Moreover, in the key we indicate the ordering of these CNOTs, because this order is important for circuit-level errors. The CNOT gates responsible for the first circuit-level-distance-reducing error mechanism discussed in the text are highlighted in yellow. To overcome this in the Extra Layer circuit, we move the problematic highlighted CNOT gates from layer 1 into a new layer inserted between layers 2 and 3. This is allowed because the highlighted CNOT gates commute with the other CNOT gates in their contraction round. (b) A weight-3 circuit-level $Z$-error that gives a syndrome that can only be created by a weight-4 end-cycle error in \cref{fig:end_cycle_optimal_boundaries}(a). It is easiest to visualise the full circuit-level error by forward- or backward-propagating it to its nearest mid-cycle code as shown in \cref{fig:circuit_level_errors}. The error consists of an end-cycle code error (which propagates to a weight-3 mid-cycle error), a mid-cycle error, and a CNOT error (which propagates to a weight-2 mid-cycle error). These combine to trigger a pair of green $X$-stabilisers. (c) The contraction tree diagram for the Extra Qubits circuit, with the re-optimised boundary conditions that take into account the circuit-level error mechanism in (b) and the extra gate layer for the highlighted CNOT gates.}
    \label{fig:weight_6_errors}
\end{figure*}

The contraction tree diagram of the weight-6 morphing circuit for the colour code with triangular planar boundaries is shown in \cref{fig:weight_6_errors}(a) for distances 3, 5, and 7. The methods we use in Appendix~\ref{sec:code_boundary_optimisation} optimise the boundary conditions for the end-cycle code, but this may result in a circuit-level distance that is lower than the end-cycle distances: this is indeed what we find in the weight-6 morphing circuit. We identified such two error mechanisms that cause this problem. The first mechanism consists of a CNOT gate error on the top spatial boundary of the code leading to a reduction in the circuit-level distance by one, as shown in \cref{fig:weight_6_errors}(a). The circuit-level distance is reduced by one because the blue stabilisers can be connected to the spatial boundary by a single CNOT gate error, even though in the end-cycle code they require a weight-2 error to connect to the same spatial boundary. The second mechanism consists of three circuit-level errors that combine to give a syndrome that, in the end-cycle code, is caused by a weight-\textit{four} error string. This is problematic because the boundaries were optimised in Appendix~\ref{sec:code_boundary_optimisation} to equalise the weights of the end-cycle logical error strings, but now with circuit-level errors some of these can have their weight reduced by a factor of 3/4.

We overcome these two issues as follows. For the first mechanism, we can simply perform a reordering of CNOT gates to remove the error mechanism. This reordering requires an additional gate layer in the contraction round but does not require any additional CNOT gates -- we call this morphing circuit the ``Extra Layer'' circuit. And for the second mechanism, we can re-optimise the boundary conditions taking into account this additional error mechanism, using additional qubits in the process. We call the circuit with both the extra layer and the re-optimised boundaries the ``Extra Qubits'' circuit. We checked using Gurobi~\cite{Gurobi} that the circuit-level distance of the Extra Qubits circuit was as expected for distances 3, 5, 7, and 9. Note that the $d^{\mathrm{end}}=3$ Extra Layer circuit also achieves $d_{\mathrm{circ}}=3$ since it is the same as the $d_{\mathrm{circ}}=3$ Extra Qubits circuit. This is not the case for the higher distance codes with the Extra Layer circuit only achieving $d_{\mathrm{circ}}=\lceil 3d^{\mathrm{end}}/4\rceil$ (again, we checked this up to and including distance 9).

We compare the original, Extra Layer, and Extra Qubits circuits numerically in \cref{fig:colour_code_morphing_and_superdense}(a). Note in particular that the original circuit does not achieve $d_{\mathrm{circ}}=3$ even for the $d^{\mathrm{end}}=3$ circuit, which is visible in a reduction in performance at lower physical error rate. Although the Extra Qubits circuit clearly outperforms both the Original and Extra Layer circuits for $d^{\mathrm{end}}=5$ and 7, it is not completely clear whether this is due to the improvement in the circuit-level distance, or simply due to the increased number of qubits, reducing the \textit{number} of low-weight logical operators in the circuit. Further investigations at lower physical error rate would be needed to investigate this.

As a last comment, in the final stages of preparing this manuscript, we had a new idea to improve the circuit-level distance of the triangular weight-6 morphing circuit that we are now actively investigating, and which uses fewer qubits than the Extra Qubits circuit -- please contact the authors for further information.

In \cref{fig:colour_code_morphing_and_superdense}(b) we compare the Extra Qubits circuit to two colour code circuits already existing in the literature: the superdense colour code circuit and the original morphing colour code circuit which both were designed by Gidney and Jones in Ref.~\cite{Gidney2023}. Interestingly, we found that both morphing circuits significantly outperformed the superdense circuit at the same distance. However, both morphing circuits use significantly \textit{more} physical qubits than the superdense circuit, see~\cref{tab:colour_code_circuits}. It is therefore an interesting question to ask whether the extra qubits used in the morphing circuits give enough of a performance improvement to be ``worth it'' at various physical error rates. 
%To investigate this, following Ref.~\cite{Gidney2023} we produced some plots of logical error rate vs the square root of the number of physical qubits at a fixed physical error rate, but the error bars in these plots were too large to draw any meaningful conclusions.

\subsection{Diamond Weight-7 Morphing Colour Code}\label{sec:weight_7_details}

\begin{figure}
    \centering
    \includegraphics[width=\linewidth]{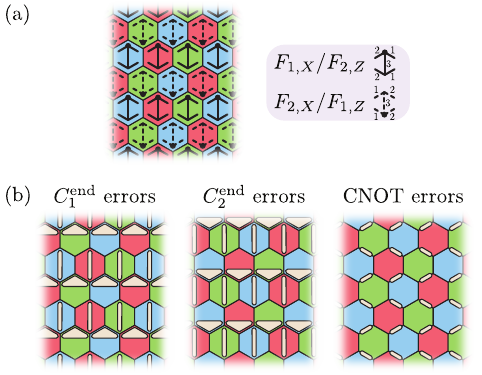}
    \caption{All possible circuit-level errors in the weight-7 morphing circuit (without planar boundaries). (a) The contraction tree diagram, where to save space again we use the fact that $F_{2}=H^{\otimes n} F_{1} H^{\otimes n\dag}$. (b) The mid-cycle code representation of the circuit-level $Z$-errors in this circuit, following \cref{fig:circuit_level_errors}. Each possible circuit-level $Z$-error in a single-round memory experiment is propagated to the mid-cycle code and is equivalent to either: a single-qubit error (e.g.~a mid-cycle error, not explicitly shown above), an end-cycle error, or a CNOT error in the first gate layer of a contraction circuit. These latter two types of error are indicated in the diagram by a beige shape whose vertices correspond to the qubits that the $Z$-error is supported on.}
    \label{fig:weight_7_errors}
\end{figure}

Next, we move to the diamond weight-7 morphing circuit. Unlike the weight-6 morphing circuit, we conjecture that the optimised boundary conditions from Appendix~\ref{sec:code_boundary_optimisation} do indeed give circuits with $d_{\mathrm{circ}}=d^{\mathrm{end}}$.  
To visualise all the possible circuit-level errors in one place, we forward- or backward-propagate each circuit-level error to the mid-cycle code. We have done this for the $Z$-errors in the weight-7 morphing circuit (without planar boundaries) in \cref{fig:weight_7_errors}. We are unable to find any error mechanisms that can create an error string with lower weight than what can be created by errors in an end-cycle code, and therefore we conjecture that the circuit is, in the bulk, fault-tolerant.

\begin{figure*}
    \centering
    \includegraphics[width=\linewidth]{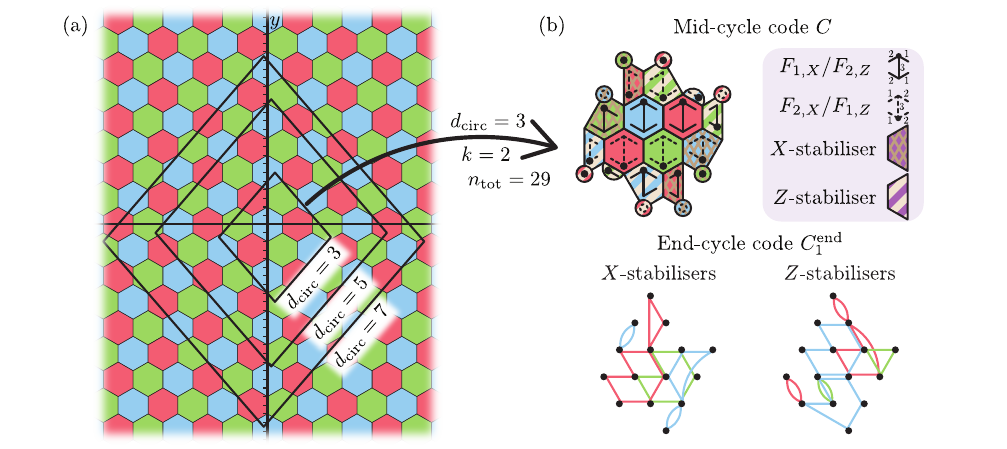}
    \caption{Diamond weight-7 morphing colour code circuits. (a) The boundaries used to construct the $d_{\mathrm{circ}}=3$, 5, and 7 circuits. The coordinate system used to specify these boundaries via their $y$-intercepts is also shown. (b) The contraction tree diagram and end-cycle code $C_{1}^{\mathrm{end}}$ for the $d_{\mathrm{circ}}=3$ circuit. Because we impose Pauli boundary types on the code (as defined in \cref{subsec:boundaries}), the $X$-stabilisers and $Z$-stabilisers are no longer identical as they often are in the colour code, as shown in the key. Note that the single-qubit stabilisers that are contracted in $F_{2}$ are indicated by using a dashed circle in the contraction tree. The second end-cycle code $C_{2}^{\mathrm{end}}$ is identical to the first up to a $\pi$-rotation.}
    \label{fig:diamond_weight_7}
\end{figure*}

\begin{table}[]
    \caption{The boundary conditions of the diamond weight-7 morphing circuits. We list the best arrangement of the diamond weight-7 morphing circuits by specifying the $y$-intercepts of the four boundary edges that define the code (anticlockwise starting with the top-right edge and ending with the bottom-right). The Cartesian coordinates defining these $y$-intercepts are shown in \cref{fig:diamond_weight_7}. The gradients of these four edges are given by the optimisation procedure in Appendix~\ref{sec:code_boundary_optimisation}. We also list the total number of qubits (data plus ancilla) involved in each code.}
    \centering
    \begin{tabular}{c|c|c}
        $d_{\mathrm{circ}}$ & $y$-intercepts & $n_{\mathrm{tot}}$\\\hline
        2 & $3,5,-8,-6$ & 16 \\
        3 & $7,5,-8,-10$ & 28 \\
        4 & $11,9,-12,-14$ & 58 \\
        5 & $15,14,-16,-17$ & 100 \\
        6 & $17,18,-20,-19$ & 136 \\
        7 & $19,20,-24,-23$ & 179 \\
        8 & $23,22,-26,-27$ & 233
    \end{tabular}
    \label{tab:diamond_weight_7}
\end{table}

However, we have found that constructing explicit circuit-level-distance-preserving planar boundaries is somewhat challenging. In Appendix~\ref{sec:code_boundary_optimisation} we mentioned that the diamond boundaries can be imposed either using two alternating coloured boundaries or two alternating Pauli boundaries (as defined in \cref{subsec:boundaries}). Here, we chose to impose Pauli boundaries since this introduces fewer irregularities in the boundary geometry. Then, we took the gradient of the diamond boundary edges from Appendix~\ref{sec:code_boundary_optimisation} and used a brute-force numerical search to find the smallest sized diamond that gives rise to a morphing circuit with a desired circuit-level distance. Our results are shown in \cref{fig:diamond_weight_7,tab:diamond_weight_7}, where we specify the microscopic location of the boundaries of the diamond by listing the $y$-intercepts of the four boundary edges with respect to a set of Cartesian coordinates. We found that, compared with the surface code and superdense colour code, for some target circuit-level distances the diamond weight-7 circuit had a low number of qubits and for others it was larger. It remains an open problem why this is the case (perhaps there is a circuit-level error mechanism arising due to the presence of the boundaries) but we leave such investigations to future work.

Finally, we numerically benchmark the performance of the diamond weight-7 morphing circuits in \cref{fig:colour_code_morphing_and_superdense}(c). Our results show that the diamond weight-7 circuits have numerical performance comparable to the triangular weight-6 circuit for the three distances that we tested.

\begin{figure*}
    \includegraphics{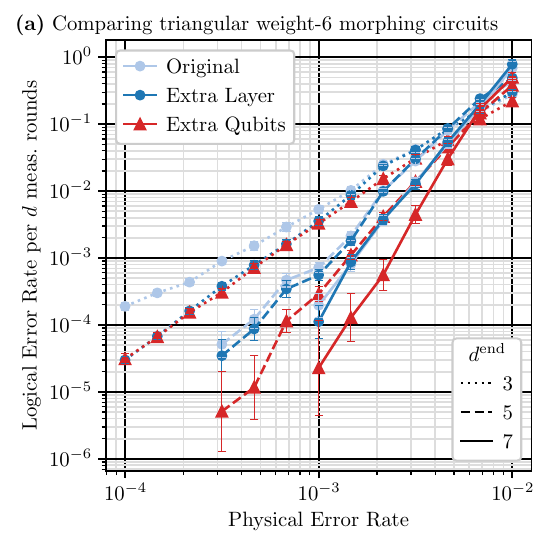}\includegraphics{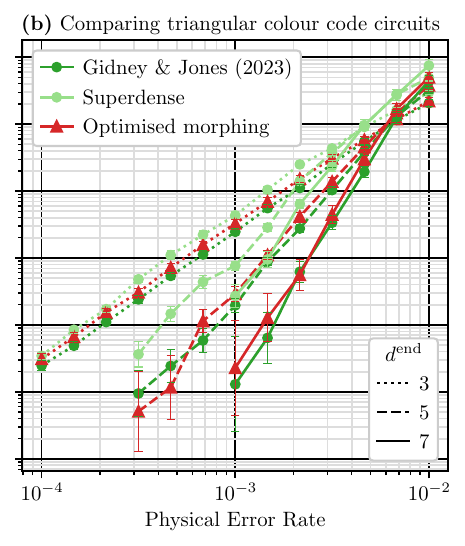}\\[-0.25 cm]
    \includegraphics{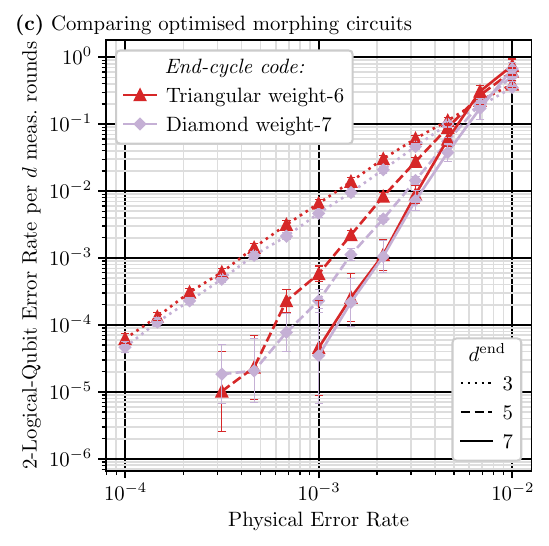}
    \caption{Numerical performance of colour code circuits in a $d$-round memory experiment under SI1000 circuit-level noise (\cref{tab:noise_model}) and decoded using a minimum-weight decoder. Error bars represent 99\% confidence intervals.}
    \label{fig:colour_code_morphing_and_superdense}
\end{figure*}

\bibliography{my_bib}
\end{document}